\documentclass[11pt, a4paper, oneside, reqno]{amsart}

\usepackage{dsfont, amssymb, amsmath, graphicx, enumitem, multicol, bm, bbm, makecell, multirow}
\usepackage{amsfonts, mathtools, mathrsfs, amsthm} 
\usepackage{float, graphicx, subcaption}
\usepackage[table,dvipsnames]{xcolor}

\RequirePackage[colorlinks,citecolor=blue,linkcolor=blue,urlcolor=ForestGreen,pagebackref]{hyperref}

\usepackage[numbers]{natbib}
\newtheorem{definition}{Definition}
\newtheorem{theorem}{Theorem}
\newtheorem{lemma}{Lemma}
\newtheorem{remark}{Remark}
\newtheorem{corollary}{Corollary}
\newtheorem{proposition}{Proposition}

\newcommand{\be}{\begin{equation}}
\newcommand{\ee}{\end{equation}}
\newcommand{\bea}{\begin{equation*}\begin{aligned}}
\newcommand{\eea}{\end{aligned}\end{equation*}}

\newcommand{\R}{\mathbb{R}}

\newcommand{\Max}{\max\limits_}
\newcommand{\Min}{\min\limits_}
\newcommand{\Sup}{\sup\limits_}
\newcommand{\Inf}{\inf\limits_}
\newcommand{\Tr}[1]{\Trace \big[ #1 \big]}

\newcommand{\wh}{\widehat}
\newcommand{\mc}{\mathcal}
\newcommand{\mbb}{\mathbb}

\newcommand{\cov}{\Sigma} 
\newcommand{\covsa}{\wh{\cov}}

\newcommand{\M}{\mc M}

\newcommand{\p}{\mbb P}
\newcommand{\PP}{\mbb P}

\newcommand{\QQ}{\mbb Q}

\newcommand{\Ambi}{\mc P}

\DeclareMathOperator{\Trace}{Tr}

\DeclareMathOperator{\st}{s.t.}

\newcommand{\Sym}{\mathbb{S}}
\newcommand{\PSD}{\mathbb{S}_{+}} 
\newcommand{\PD}{\mathbb{S}_{++}} 

\newcommand{\opt}{^\star}
\newcommand{\eps}{\varepsilon}

\newcommand{\Wass}{\mathds{W}}
\newcommand{\V}{\mc{V}}
\newcommand{\Q}{\mbb{Q}}
\newcommand{\EE}{\mathds{E}}

\newcommand{\m}{\mu}
\newcommand{\msa}{\wh \mu}

\newcommand{\U}{\mc U}

\newcommand{\Var}{\mathbb{V}\text{ar}} % Variance

\newcommand{\half}{\frac{1}{2}}

\DeclareMathOperator{\CVaR}{CVaR}
\DeclareMathOperator{\VaR}{VaR}

\newcommand{\Gelbrich}{\mathds{G}}

\newcommand{\w}{w}

\newcommand{\Pnom}{\wh \p}
\newcommand{\Lc}{\mc L}

\newcommand{\risk}{\mc R}
\newcommand{\diff}{\mathrm{d}}

\newcommand{\DS}{\displaystyle}

\newcommand{\RR}{\mathbb R}

\title[Mean-Covariance Robust Risk Measurement]{Mean-Covariance Robust Risk Measurement}

\title{Mean-Covariance Robust Risk Measurement}
\author{Viet~Anh~Nguyen, Soroosh~Shafiee, Damir Filipovi\'{c}, and Daniel~Kuhn}
\thanks{The authors are with the Department of Systems Engineering and Engineering Management, Chinese University of Hong Kong (\texttt{nguyen@se.cuhk.edu.hk}), Department of Operations Research and Information Engineering, Cornell University (\texttt{shafiee@cornell.edu}), EPFL and Swiss Finance Institute ({\tt damir.filipovic@epfl.ch}), the Risk Analytics and Optimization Chair, EPFL (\texttt{daniel.kuhn@epfl.ch}).}

\begin{document}

\begin{abstract}
We introduce a {principled} framework for mean-covariance robust risk measurement and portfolio optimization. We model uncertainty in terms of the Gelbrich distance on the mean-covariance space, along with prior structural information about the population distribution. Our approach is related to the theory of optimal transport and exhibits statistical and computational properties superior to those of existing models. We find that for a large class of risk measures, mean-covariance robust portfolio optimization boils down to the Markowitz model, subject to a regularization term given in closed form. This includes the finance standards, value-at-risk, and conditional value-at-risk, and can be solved highly efficiently.

\end{abstract}

\maketitle

%\begin{keyword}
%\kwd{Robust optimization}
%\kwd{risk measurement}
%\kwd{optimal transport}
%%\kwd{portfolio theory}
%\end{keyword}
%
%\end{frontmatter}

\section{Introduction}
Portfolio managers distribute their funds over multiple assets with the aim of optimizing future returns. The workhorse in the finance industry is the Markowitz model~\cite{ref:markowitz1952portfolio}, which maximizes the expected portfolio return adjusted for its standard deviation. Its optimal portfolio is given in closed form in terms of the mean and covariance matrix of the asset returns. The classical Markowitz model assumes symmetrically distributed asset returns, so that minimizing the standard deviation is equivalent to minimizing risk. However, asset returns are known to be skewed, and the standard deviation is unable to distinguish undesirable deviations below the mean from desirable deviations above the mean. Numerous propositions for a more appropriate risk assessment have made research on downside risk measures a vibrant field spanning economics and finance. These notably include the industry standards, value-at-risk (VaR) \cite{ref:jorion1996var, JPM_96, duf_pan_97} and conditional value-at-risk (CVaR) \cite{ref:artzner1999coherent, ref:rockafellar2000optimization}, as well as a plethora of more general distribution-based risk measures that have emerged over the last two decades, see, {\em e.g.}, \cite{ref:kusuoka2001coherent, ref:acerbi2002spectral}. 

A distribution-based risk measurement requires precise knowledge of the joint distribution of the underlying asset returns, which is unobservable in practice. For assets traded on public markets, one may attempt to estimate this distribution from historical data. However, \cite{Roy52} pointed out that the mean and covariance matrix are the only quantities that can reasonably be distilled from financial time series. This observation has spurred interest in robust risk measurement by maximizing a given distribution-based risk measure across all asset return distributions in a Chebyshev ambiguity set, which consists of all distributions with a fixed mean and covariance matrix \cite{ref:ghaoui2003worst, ref:yu2009projection, ref:zymler2013wcvar,ref:rujeerapaiboon2016robust, ref:li2018risk, ref:cai2020distributionally}.

While taking the worst case over a Chebyshev ambiguity set can immunize the risk measurement against uncertainty in the {\em shape} of the asset return distribution, it does not take into account the uncertainty in its {\em mean} and {\em covariance matrix}. This is worrying because estimators for the mean display a notoriously high variance irrespective of the sampling frequency of historical return data~\cite[\S~8.5]{ref:luenberger1997investment}. In addition, estimating high-dimensional covariance matrices is a formidable challenge in statistics that requires structural information~\cite{ref:ledoit2003improved}. Unfortunately, estimation errors in the mean and covariance matrix have a detrimental impact on the solution of a portfolio optimization problem because the optimal portfolio's out-of-sample performance falls severely short of its in-sample performance \cite{ref:michaud1989markowitz, ref:best1991sensitivity, ref:chopra2011errors}.

In this paper, we introduce a {principled} framework for robustifying any distribution-based risk measurement against mean-covariance uncertainty. To this end, we propose the Gelbrich distance~\cite{ref:gelbrich1990formula} as a metric in the space of mean-covariance pairs. We define the Gelbrich ambiguity set as the family of all asset return distributions with a given structure whose mean-covariance pairs reside in a Gelbrich ball around an empirical mean-covariance pair estimated from data. Here, the structure of a distribution refers to any of its properties that are complementary to location and scale. Examples include symmetry, unimodality, log-concavity or Gaussianity. We think of the structure as reflecting domain knowledge that is uninformed by data. Unlike the Chebyshev ambiguity set, the Gelbrich ambiguity set takes account of the uncertainty in the shape {\em as well as} the mean and covariance matrix of the asset return~distribution. For any given distribution-based risk measure, we then define the Gelbrich risk as the worst-case risk over the Gelbrich ambiguity set. 

We find that, if the underlying risk measure is law invariant, translation invariant and positive homogeneous, then the Gelbrich risk reduces to a regularized mean-standard deviation risk measure. The Gelbrich risk minimization problem is therefore equivalent to a regularized Markowitz portfolio selection problem, which can be solved highly efficiently. We thus revive and legitimate the Markowitz model, subject to a fully explicit and tractable regularization, irrespective its aforementioned shortcomings. The underlying risk measure and the structure of the asset return distribution impact the Gelbrich risk only through a scalar coefficient, which can be computed offline and is available in closed form for all coherent, spectral and distortion risk measures. In addition, the weight of the regularization term scales with the radius of the Gelbrich ambiguity set. As a corollary, we obtain that the equally weighted portfolio minimizes the Gelbrich risk under extreme mean-covariance uncertainty, that is, in the limit of an infinitely large Gelbrich ambiguity set.

For any fixed portfolio, we then analytically characterize the worst-case asset return distributions that maximize the underlying risk measure over the Gelbrich ambiguity set. These distributions reveal the portfolio's vulnerabilities and are straightforwardly applicable for stress tests.

The Gelbrich risk is intimately related to the theory of optimal transport \cite{ref:villani2008optimal}. Indeed, the Gelbrich distance between two mean-covariance pairs coincides with the 2-Wasserstein distance between the corresponding Gaussian distributions~\cite{ref:gelbrich1990formula}. Defining the Wasserstein risk as the worst-case risk over all distributions with a given structure that reside in a 2-Wasserstein ball around a nominal distribution estimated from data, we show that the Gelbrich risk upper bounds the Wasserstein risk if the underlying ambiguity sets have the same radius. This bound is sharp and collapses to an equality if the underlying ambiguity sets contain only Gaussian distributions.

We also show that the Gelbrich risk provides a finite-sample upper confidence bound on the true risk under the population distribution if the radius of the Gelbrich ambiguity set scales with the inverse square root of the sample size. This finite-sample bound is dimension-free in the sense that the rate does not depend the number of assets. This result contrasts sharply with existing out-of-sample guarantees for the Wasserstein risk \cite{ref:esfahani2018data}, which rely on a measure concentration result by~\cite{ref:fournier2015rate} and suffer from a curse of dimensionality. Our result is also orthogonal to the dimension-free finite-sample guarantees for the Wasserstein risk by~\cite{ref:gao2020finite}, which rely on concepts of hypothesis complexity such as covering numbers or Rademacher complexities and which apply only to worst-case expectations but may not easily generalize to other risk measures.

The study of decision problems under distributional uncertainty has a long and distinguished history in economics dating back at least to~\cite{K21:treatise} and~\cite{K21:risk}. \cite{Ellsberg_1961} was the first to document that most individuals have a low tolerance for distributional uncertainty. This phenomenon is often used to motivate the maxim that different decision alternatives should be ranked in view of their worst-case performance with respect to all distributions in some ambiguity set. A rigorous axiomatic justification for this decision rule is due to \cite{ref:gilboa1989maxmin}. In operations research, decision problems under uncertainty are typically framed as distributionally robust optimization problems, and research focuses primarily on deriving tractable reformulations and efficient solution algorithms \cite{ref:delage2010distributionally,ref:goh2010distributionally,  ref:wiesemann2014distributionally}.

% The analysis of the worst-case expectation dates back to the seminal papers of \cite{ref:scarf1958min} for newsvendor problems, where the ambiguity set contains distributions with the same mean and variance. \cite{ref:vzavckova1966minimax} then derived tractable reformulations for the worst-case expectation of a piecewise linear function, as long as only the support and mean of the random variables are available. For general loss functions, Monte-Carlo sampling was traditionally used to evaluate the worst-case expectation~\cite{ref:shapiro2002minimax, ref:shapiro2004class}. Recent reformulations, however, are based on duality results for moment problems \cite{ref:isii1960extrema, ref:shapiro2001on}. In particular, 

Considerable efforts were directed to investigating various distributionally robust portfolio optimization models. When modeling distributional uncertainty via Chebyshev ambiguity sets, the worst-case risk admits tractable reformulations for the VaR \cite{ref:ghaoui2003worst, ref:zymler2013wcvar, ref:rujeerapaiboon2016robust, ref:rujeerapaiboon2018chebyshev}, the CVaR \cite{ref:natarajan2010utility, ref:chen2011tight, ref:zymler2013wcvar} as well as for all spectral risk measures \cite{ref:li2018risk} and all distortion risk measures~\cite{ref:cai2020distributionally, ref:pesenti202optimizing}. These reformulations will emerge as special cases of our results because any Chebyshev ambiguity set constitutes a Gelbrich ambiguity set with a vanishing radius. We also stress that if the mean and covariance matrix of the asset returns are estimated from data, the corresponding Chebyshev ambiguity set will {\em not} contain the data-generating distribution with probability one. Thus, the Chebyshev risk fails to provide a safe estimate for the true risk. To obtain a safe estimate, one could inflate the Chebyshev ambiguity set by allowing the mean and covariance matrix of the asset returns to range over simple box-type or semidefinite-representable confidence sets \cite{ref:ghaoui2003worst, ref:delage2010distributionally, ref:zymler2013wcvar, ref:rujeerapaiboon2016robust}. However, the design of these confidence sets is driven by computational rather than economic or statistical considerations. Moreover, tractability results are limited to special risk measures such as VaR or CVaR. Other ambiguity sets used in distributionally robust portfolio optimization impose asymmetric moment bounds \cite{ref:chen2010cvar, ref:natarajan2008incorporating, ref:natarajan2018asymmetry}, marginal moment bounds \cite{ref:doan2015robustness} or structural properties such as symmetry, unimodality or tail convexity etc.\ \cite{ref:popescu2005semidefinite, ref:yu2009projection, ref:van2016generalized, ref:lam2017tail}. Ambiguity sets based on factor models for the asset returns~\cite{ ref:nemirovski2007convex}, information divergences \cite{ref:ghaoui2003worst,ref:ben2013robust} or multiple priors~\cite{ref:garlappi2007portfolio} have also been studied.

Distributionally robust portfolio optimization problems with 1-Wasserstein ball ambiguity sets for discrete asset return distributions are addressed by~\cite{ref:pflug2007ambiguity} via exhaustive search methods from global optimization. Using robust optimization techniques, \cite{ref:postek2016risk} show that these problems are actually equivalent to tractable convex programs and thus amenable to efficient iterative algorithms. When the ambiguity set may contain generic non-discrete distributions, the worst case of any convex distortion risk measure over a $p$-Wasserstein ball for $p\ge 1$ coincides with the nominal risk adjusted by a regualization term penalizing some dual norm of the portfolio weight vector~\cite{ref:wozabal2014robustifying}; see also~\cite[\S~4]{ref:pichler2013evaluations} for a related discussion.
% he regularization coefficient in \cite{ref:wozabal2014robustifying} is not easily computable as it involves $L_p$ norm of a random variable.
However, all of these results apply only to {\em convex} risk measures (thus excluding VaR) and fail to account for structural distributional information (meaning that the ambiguity set may contain unrealistic pathological distributions). Also, the only universal statistical guarantees on the true risk known to date suffer from a curse of dimensionality \cite{ref:esfahani2018data}.

The regularization terms penalizing the norm of the portfolio weight vector, which emerge from our reformulation of the Gelbrich risk, can be viewed as {\em implicit} norm constraints. In the context of a Markowitz model, \cite{ref:jagannathan2003risk} show that imposing {\em explicit} norm constraints is equivalent to shrinking the estimator of the covariance matrix. \cite{ref:demiguel2009generalized} offer a Bayesian interpretation for the resulting optimal portfolios and show empirically that they display an excellent out-of-sample performance. In contrast, our paper offers a new probabilistic interpretation for norm regularization terms and unveils how they depend on the risk measure and any available structural information.

The Gelbrich ambiguity set was first proposed by \cite{ref:nguyen2019bridging} to robustify minimum mean square error estimation problems against distributional uncertainty. Similar ideas were also used by \cite{ref:nguyen2018distributionally} and~\cite{ref:shafieezadeh2018wasserstein} in the context of inverse covariance matrix estimation and Kalman filtering, respectively.

The remainder of the paper is structured as follows. Section~\ref{sec:statement} formally introduces the Gelbrich risk. Section~\ref{sec:properties} investigates its attractive conceptual and statistical properties. Section~\ref{sec:linearportfolio} demonstrates its efficient computability for law invariant, translation invariant and positive homogeneous base risk measures in a classical portfolio selection setting. Section~\ref{sec:nonlinearportfolio} extends these findings to generalized portfolio selection and index tracking problems, and Section~\ref{sec:numerical} reports on numerical results. All proofs as well as some additional tractability results are relegated to the appendix.

\textit{Notation.} The $2$-norm of a vector $x \in \R^n$ is denoted by $\| x \|$. Similarly, we use $\| A \|$ and $\| A \|_F$ to denote the spectral and Frobenius norms of a square matrix $A \in \R^{n \times n}$, respectively. The space of all symmetric matrices in~$\R^{n\times n}$ is denoted as~$\Sym^n$, while $\PD^n$ ($\PSD^n$) stands for the cone of all positive (semi)definite matrices in~$\Sym^n$. For any $A \in \Sym^n$, $\lambda_{\min}(A)$ and $\lambda_{\max}(A)$ denote the minimum and maximum eigenvalues of $A$, respectively. We denote the Borel $\sigma$-algebra on $\R^n$ by $\mc B(\R^n)$, the space of Borel-measurable functions from~$\R^n$ to~$\mathbb R$ by~$\mc L_0$, and the set of all probability distributions on~$\mc B(\R^n)$ by~$\mc M$. The expectation of a random variable $\ell\in\mc L_0$ under $\Q\in\M$ is denoted by $\EE_\QQ[\ell]$. We denote by $\mc M_2$ the set of all~$\Q\in\mc M$ with finite second moments, that is, with~$\EE_\QQ[\|\xi\|^2] <\infty$.

\section{Problem Statement}
\label{sec:statement}
We study decision problems under distributional uncertainty that aim to minimize the risk of a loss affected by a vector~$\xi \in \R^n$ of {\em risk factors}. Formally, a {\em loss function} $\ell\in \mc L_0$ assigns each realization~$\xi \in \R^n$ of the risk factors a loss~$\ell(\xi)\in\R$. Different loss functions correspond to different decision alternatives available to a risk-averse decision maker. If the risk factors are governed by a probability distribution~$\PP\in\mc M$, then the decision maker ranks the loss functions according to a risk measure~$\risk_\PP: \mc L_0 \to \R \cup \{+\infty\}$, which usually depends on~$\PP$. We thus define the {\em risk} of any loss function~$\ell\in \Lc_0$ under~$\PP$ as~$\risk_\PP(\ell)$. In addition, we define the {\em optimal risk} corresponding to a set~$\mc L \subseteq \mc L_0$ of {\em admissible} loss functions under~$\PP$ as~$\inf_{\ell\in\mc L}\risk_\PP(\ell)$. Unfortunately, the true probability distribution~$\PP$ is almost never known in practice, and thus neither the risk of a fixed loss function nor the optimal risk can be evaluated reliably.

We henceforth assume that the decision maker has only access to limited statistical information, along with some prior structural information, about~$\PP$. Formally, she only knows that $\PP$ lies in some ambiguity set $\mc P\subseteq\M$. For a given family of risk measures $\{\mc R_\QQ\}_{\QQ\in\mc M}$ this leads to the corresponding {\em worst-case risk} of any loss function $\ell\in\mc L_0$ given by
\begin{equation}\label{wcrm}
    \mc R_{\mc P}(\ell) = \sup_{\QQ \in\mc P} \mc R_\QQ(\ell).
\end{equation}
Ranking different loss functions by their worst-case risk, the decision maker thus solves a distributionally robust optimization problem that finds the {\em optimal worst-case risk}
\begin{equation}\label{optimalwcr}
   \mc R_{\mc P}(\mc L) = \inf_{\ell\in\mc L} \mc R_{\mc P}(\ell) = \inf_{\ell\in\mc L} \sup_{\QQ \in\mc P} \mc R_\QQ(\ell).
\end{equation}

In the following we discuss the choice of the ambiguity set~$\mc P$. Specifically, in Section~\ref{sec:structural} we first formalize the modeling of structural information, and in Section~\ref{ref:statistical} we address the modeling of statistical information. We thereby focus on mean-covariance uncertainty.

\subsection{Structural Information}
\label{sec:structural}

We encode structural information about the unknown probability distribution~$\PP$ by a structural ambiguity set defined as follows. 

\begin{definition}[Structural ambiguity set]
    \label{def:structural}
    A structural ambiguity set~$\mc S$ is a subset of~$\mc M_2$ that is closed under positive semidefinite affine pushforwards. That is, for any $\QQ \in \mc S$ and any transformation $f: \R^n \to \R^n$ of the form $f(\xi) = A\xi + b$ for some $A \in \PSD^{n}$ and $b \in \R^n$, the pushforward distribution $\QQ \circ f^{-1}$ belongs to $\mc S$.
\end{definition}

The entire set~$\mathcal M_2$ trivially constitutes a structural ambiguity set. Below we discuss non-trivial examples of structural ambiguity sets that will be addressed in this paper.

\begin{definition}[Symmetric distribution]
    \label{definition:symmetric3}
    The probability distribution $\QQ\in\mc M_2$ is symmetric if there exists $\m \in \R^n$ with $\QQ[\xi \le \m- \tau]=\QQ[\xi \ge \m+\tau]$ for all $\tau \in \R^n$.
\end{definition}

By definition, $\QQ$ is symmetric about $\mu$ if and only if the random vectors $\xi-\mu$ and $\mu-\xi$ have the same cumulative distribution function under $\QQ$, which is equivalent to the condition that $\QQ\left[(\xi-\mu)\in B\right] = \QQ\left[(\mu-\xi) \in B \right]$ for all Borel sets $B \subseteq \R^n$. Hence, the set of symmetric probability distributions is closed under positive semidefinite affine pushforwards and thus constitutes a structural ambiguity set; see also ~\cite[Lemma~1]{ref:yu2009projection}.

While there is consensus about what it means for a {\em univariate} distribution to be unimodal, there are several non-equivalent notions of unimodality for {\em multivariate} distributions. In the following we will argue that two of these notions, namely linear unimodality and log-concavity, give rise to structural ambiguity sets.

\begin{definition}[Linear unimodal distribution]
    \label{definition:linear-unimodal}
    The probability distribution $\QQ\in\mc M_2$ is  linear unimodal if there exists $\mu\in\R^n$ such that the cumulative distribution function of $\w^\top \xi$ under $\QQ$ is convex on $(-\infty, w^\top\mu]$ and concave on $[w^\top\mu, +\infty)$ for all $\w \in \R^n$.
\end{definition}

It is easy to show that the set of linear unimodal distributions is closed under positive semidefinite affine pushforwards and thus also constitutes a structural ambiguity set.

\begin{definition}[Log-concave distribution]
    \label{definition:log-concave}
    The probability distribution $\QQ \in \mc M_2$ is log-concave if for any Borel sets~$B_1, B_2 \in \mc B(\R^n)$ and for any scalar weight~$\theta \in [0, 1]$ we have
    \begin{align*}
        \mathbb Q( \theta B_1 + (1 - \theta) B_2 ) \geq \mathbb Q(B_1)^\theta \mathbb Q(B_2)^{1 - \theta},
    \end{align*}
    where the convex combination of $B_1$ and $B_2$ is understood in the sense of Minkowski.
\end{definition}

Log-concave distributions play an important role in statistics and optimization. Many standard distributions such as the uniform distributions on convex sets as well as the Gaussian, Wishart or Dirichlet distributions are log-concave. All log-concave distributions have sub-exponential tails \cite{ref:borell1983convexity}. Moreover, by \cite[Lemma~2.1]{ref:dharmadhikari1988unimodality} the family of all log-concave distributions is closed under positive semidefinite affine pushforwards and thus constitutes a structural ambiguity set.

\begin{definition}[Elliptical distribution]
    \label{definition:elliptical}
    The probability distribution $\QQ\in\mc M_2$ is elliptical if its characteristic function $\EE_\QQ[\exp(i \tau^\top \xi)]$ can be written as $\exp(i \tau^\top \m) \phi(\tau^\top \cov \tau)$ for some location parameter $\m \in \R^n$, dispersion matrix $\cov \in \PSD^n$ and characteristic generator $\phi: \R_+ \to \R$, where $i = \sqrt{-1}$ denotes the imaginary unit.
\end{definition}

%The elliptical distributions generalize common distribution families such as the Gaussian, logistic or $t$-distributions. 
By \cite[Theorem~1]{ref:cambanis1981theory}, the family of all elliptical distributions with the same characteristic generator is closed under positive semidefinite affine pushforwards and thus constitutes a structural ambiguity set. The same theorem implies that every elliptical distribution is symmetric. However, not every elliptical distribution is unimodal. Examples of elliptical distributions that fail to be unimodal include certain Kotz-type or multivariate Bessel distributions. The location parameter $\m$ of an elliptical distribution~$\QQ$ always matches the mean of~$\QQ$. In addition, one can show that the generator $\phi$ of~$\QQ$ may not be chosen freely but is only admissible if $\phi(\|\tau\|^2)$ represents the characteristic function of some probability distribution $\QQ'$ with finite second moments. Note that $\QQ'$ differs from $\QQ$ unless $\m=0$ and $\cov=I_n$. As probability distributions are normalized, this condition can only hold if $\phi(0) = 1$. And as any characteristic function is continuous thanks to the dominated convergence theorem, this condition further implies that $\phi$ must be continuous. One can also show that the covariance matrix of~$\QQ$ is given by $-2 \phi'(0)\cov$, where~$\phi'(0)$ stands for the right derivative of~$\phi(u)$ at $u=0$. Assuming that $\QQ$ has finite second moments is thus equivalent to assuming that~$\phi'(0)$ exists and is finite. In the remainder of the paper we will assume without loss of generality that $\phi'(0)=-\frac{1}{2}$, which means that the dispersion matrix $\cov$ coincides with the covariance matrix of~$\QQ$. Indeed, changing the characteristic generator to $\phi(\frac{-u}{2\phi'(0)})$ and the dispersion matrix to $-2 \phi'(0)\cov$ does not change $\QQ$ but ensures that the dispersion matrix and the covariance matrix of~$\QQ$ coincide. 
%For details on these properties of elliptical distributions see~\cite{ref:cambanis1981theory}.
Note that the families of all symmetric, linear unimodal or log-concave distributions as well as the family of all elliptical distributions with a given generator fail to be convex. For example, the set of Gaussian distributions, which is obtained by setting $\phi(u)=e^{-u/2}$, is non-convex because mixtures of Gaussian distributions are generically multimodal and thus not Gaussian. 

We now define the smallest structural ambiguity set that contains a given~$\Pnom\in\mc M_2$. 

\begin{definition}[Structural ambiguity set generated by~$\Pnom$] \label{def:element}
     The structural ambiguity set generated by $\Pnom\in\mc M_2$ is the family of all positive semidefinite affine pushforwards of $\Pnom$.
\end{definition}

By \cite[Theorem~1]{ref:cambanis1981theory}, the set of all elliptical distributions with generator~$\phi$ can be viewed as the structural ambiguity set generated by the standardized elliptical distribution~$\Pnom$ with generator~$\phi$, mean~$\mu=0$ and covariance matrix~$\cov=I$. In contrast, the sets of all symmetric, linear unimodal or log-concave distributions are not generated by any single distribution. Also, not every distribution in the structural ambiguity set~$\mc S$ generated by~$\Pnom$ generates all of~$\mc S$. For example, the Dirac distribution~$\delta_0$ that concentrates unit mass at~0 constitutes a (degenerate) Gaussian distribution with mean~$\mu=0$ and covariance matrix~$\cov=0$. However, $\delta_0$ fails to generate the family of all Gaussian distributions because any positive semidefinite affine pushforward of~$\delta_0$ is also Dirac distribution.

\subsection{Statistical Information}
\label{ref:statistical}

In addition to structural information captured by the ambiguity set~$\mc S$, the decision maker may have access to a finite training sample that provides statistical information about the unknown probability distribution~$\PP$. In a financial context, such training samples are routinely used to construct estimators~$\msa$ and~$\covsa$ for the mean and covariance matrix of~$\PP$, respectively. Indeed, \cite{Roy52} asserts that the first and second moments of~$\PP$ `{\em are the only quantities that can be distilled out of our knowledge of the past}.' Moreover, he adds that `{\em the slightest acquaintance with problems of analysing economic time series will suggest that this assumption is optimistic rather than unnecessarily restrictive.}' Roy's warning alerts us that the true mean~$\mu$ and the true covariance matrix~$\cov$ of~$\PP$ typically differ from their noisy estimators~$\msa$ and~$\covsa$. In the following, we quantify the corresponding estimation errors via the Gelbrich distance in the space of mean-covariance pairs.

\begin{definition}[Gelbrich distance] \label{def:Gelbrich}
    The Gelbrich distance between two mean-covari\-ance pairs~$(\m_1, \cov_1)$ and $(\m_2, \cov_2)$ in~$\R^n \times \PSD^n$ is given by
\begin{equation*}
\Gelbrich\big( (\m_1, \cov_1), (\m_2, \cov_2) \big) = \sqrt{ \| \m_1 - \m_2 \|^2 + \Tr{\cov_1 + \cov_2 - 2 \big( \cov_2^{\half} \cov_1 \cov_2^{\half} \big)^{\half} } }.
\end{equation*}
\end{definition}
One can show that the Gelbrich distance is non-negative, symmetric and subadditive and that it vanishes if and only if $(\m_1, \cov_1) = (\m_2, \cov_2)$, which implies that it represents a metric on $\R^n \times \PSD^n$ \cite[pp.~239]{ref:givens1984class}. 

% The expression for the Gelbrich distance simplifies also if $\cov_1$ and $\cov_2$ commute ($\cov_1 \cov_2= \cov_2 \cov_1$). In this special case, one easily verifies that it collapses to the Euclidean distance between~$(\mu_1,\cov_1^\half)$ and~$(\mu_2,\cov_2^\half)$, that is,
% \[
%     \Gelbrich(\cov_1, \cov_2) = \sqrt{\|\m_1 - \m_2\|^2 + \| \cov_1^\half - \cov_2^\half \|_F^2}.
% \]

If $\mu_1=\mu_2$, then the Gelbrich distance reduces to the Bures distance that measures the dissimilarity between density matrix operators in quantum information theory~\cite{ref:bhatia2018strong, ref:bhatia2018bures}. In this case, the Gelbrich distance induces a Riemannian metric on the space of positive semidefinite matrices. If, in addition, the two covariance matrices are diagonal, then the Gelbrich distance simplifies to the Hellinger distance, which is closely related to the Fisher-Rao metric ubiquitous in information theory~\cite{ref:liese1987convex}.

We can thus define the mean-covariance uncertainty set as the ball of radius~$\rho\ge 0$ centered at the estimators~$\msa$ and~$\covsa$ of the mean and covariance matrix of~$\PP$,
\[
\U_{\rho}(\msa, \covsa) = 
\left\{ (\m,\cov) \in \R^n\times \PSD^n: \Gelbrich\big( (\m, \cov), (\msa, \covsa) \big) \leq \rho \right\}.
\]
We can then introduce the {\em Gelbrich ambiguity set} as the family of all probability distributions of~$\xi$ that are consistent with the available structural and statistical information. Formally, the Gelbrich ambiguity set is defined as the pre-image of $\mathcal U_{\rho}(\msa, \covsa)$ under the transformation that maps~$\QQ \in \mc S$ to its first and second moments.
\begin{definition}[Gelbrich ambiguity set]
    \label{def:gelbrich-hull}
    The Gelbrich ambiguity set is given by
    \begin{equation*}
    \label{eq:Gelbrich:ball}
    \mc{G}_{\rho}(\msa, \covsa) = \left\{ \QQ \in \mc S: 
    \left(\EE_{\QQ}[\xi] , \EE_\QQ[(\xi-\EE_{\QQ}[\xi])(\xi-\EE_{\QQ}[\xi])^\top] \right)  \in \mathcal U_{\rho}(\msa, \covsa)
    \right\}.
    \end{equation*}
\end{definition}

Note that the Gelbrich ambiguity set contains all distributions in the structural ambiguity set~$\mc S$ whose mean-covariance pairs have a Gelbrich distance of at most~$\rho$ from the estimators~$(\msa,\covsa)$. If the estimation error of~$(\msa,\covsa)$ is at most~$\rho$, then the unknown true distribution~$\PP$ is thus guaranteed to belong to the Gelbrich ambiguity set. In this case the true risk of a Borel-measurable loss function~$\ell\in \Lc_0$ cannot be evaluated because~$\PP$ is unknown (except in trivial cases, {\em e.g.}, when~$\ell$ is constant). However, we can evaluate the worst-case risk of~$\ell$ with respect to all probability distributions in the Gelbrich ambiguity set~$\mc{G}_{\rho}(\msa, \covsa)$. We thus define the {\em Gelbrich risk} $\risk_{\mc{G}_{\rho}(\msa, \covsa)}(\ell)$ as the worst-case risk~\eqref{wcrm} with respect to the ambiguity set~$\mc P = \mc{G}_{\rho}(\msa, \covsa)$. 
In addition, we define the {\em optimal Gelbrich risk} $\risk_{\mc{G}_{\rho}(\msa, \covsa)}(\Lc)$ corresponding to a set~$\mc L \subseteq \mc L_0$ of admissible loss functions as the infimum of the Gelbrich risk over all loss functions~$\ell\in\mc L$, defined as in~\eqref{optimalwcr} with $\mc P = \mc{G}_{\rho}(\msa, \covsa)$.
Note that by solving the optimal Gelbrich risk problem, the decision maker anticipates the worst possible probability distribution in the Gelbrich ambiguity set and seeks a decision alternative that results in the least possible risk under this worst-case distribution. Computing $\risk_{\mc{G}_{\rho}(\msa, \covsa)}(\Lc)$ thus amounts to solving a distributionally robust optimization problem.

\section{Properties of the Gelbrich Ambiguity Set}
\label{sec:properties}
We now provide conceptual, statistical and computational justification for modeling distributional uncertainty via Gelbrich ambiguity sets. In Section~\ref{sec:wasserstein-risk} we first show that the Gelbrich ambiguity set is closely related to the Wasserstein ambiguity set, which is widely used in distributionally robust optimization. In Section~\ref{sect:guarantee} we then investigate the statistical and in Section~\ref{sect:computational-perperties} the computational properties of the Gelbrich ambiguity set.

\subsection{Relation between the Gelbrich and Wasserstein Ambiguity Sets}
\label{sec:wasserstein-risk}

Instead of directly estimating the mean and covariance matrix from a training sample, one could construct a {\em nominal} probability distribution~$\Pnom$ representing a best guess of the unknown true distribution~$\PP$. Throughout the rest of the paper we assume that~$\Pnom$ belongs to the structural ambiguity set~$\mc S$ and that~$\msa$ and~$\covsa$ coincide with the mean and covariance matrix of~$\Pnom$, respectively. Note that the first and second moments of~$\Pnom$ exist because~$\mc S\subseteq\mc M_2$. Of course, the decision maker should not put full trust into the nominal distribution~$\Pnom$, which can be viewed as a noisy estimator for~$\PP$, and the corresponding estimation errors are conveniently measured by the 2-Wasserstein distance on the space~$\mc M_2$.

\begin{definition}[Wasserstein distance]
    \label{definition:wasserstein}    
    The 2-Wasserstein distance between two distributions $\QQ_1,\QQ_2\in \mc M_2$~is 
    \be
    \notag
    \Wass(\QQ_1, \QQ_2) = \Min{\pi\in\Pi(\QQ_1,\QQ_2)}  \left( \int_{\R^n \times \R^n} \|\xi_1 - \xi_2\|^2\, \pi({\rm d}\xi_1, {\rm d} \xi_2) \right)^{\half},
    \ee
    where $\Pi(\QQ_1,\QQ_2)$ denotes the set of all couplings of $\QQ_1$ and $\QQ_2$, that is, the set of all joint distributions of $\xi_1\in\R^n$ and $\xi_2\in\R^n$ with marginal distributions $\QQ_1$ and  $\QQ_2$, respectively. 
\end{definition}

The 2-Wasserstein distance is non-negative, symmetric and subadditive, and it vanishes only if $\QQ_1 = \QQ_2$, which implies that it represents a metric on~$\mc M_2$ \cite[p.~94]{ref:villani2008optimal}. In addition, the minimization problem over~$\pi$ is always solvable \cite[Theorem~5.9]{ref:villani2008optimal}, and $\Wass(\QQ_1, \QQ_2)^2$ can be viewed as the minimum cost of transporting the distribution $\QQ_1$ to $\QQ_2$, assuming that the cost of moving a unit probability mass from~$\xi_1$ to~$\xi_2$ amounts to~$\|\xi_1-\xi_2\|^2$. The variable~$\pi$ thus encodes a (probability) mass transportation plan. 

We can now define the Wasserstein ambiguity set with structural information
\begin{equation*}
\label{eq:ball}
\mc W_{\rho}(\Pnom) = \left\{ \QQ \in \mathcal S: \Wass(\Pnom, \QQ) \leq \rho \right\}
\end{equation*}
as the ball of radius $\rho \geq 0$ in the structural ambiguity set $\mathcal S$ centered at the nominal distribution $\Pnom$ with respect to the 2-Wasserstein distance. Intuitively, the radius~$\rho$ of this ambiguity set quantifies the decision maker's distrust in the nominal distribution~$\Pnom$. Note that since~$\Pnom\in\mc S$ and~$\Wass(\Pnom, \Pnom) = 0$, the Wasserstein ambiguity set is non-empty for every~$\rho\ge 0$.

If the true probability distribution~$\PP$ of the risk factors is unknown, then the decision maker could rank different loss functions in view of their worst-case risk over the Wasserstein ambiguity set. 
 We thus define the {\em Wasserstein risk} $\risk_{\mc{W}_{\rho}(\Pnom)}(\ell)$ as the worst-case risk~\eqref{wcrm} with $\mc P = \mc{W}_{\rho}(\Pnom)$. 
In addition, we define the {\em optimal Wasserstein risk} $\risk_{\mc{W}_{\rho}(\Pnom)}(\Lc)$ corresponding to a set~$\mc L \subseteq \mc L_0$ of admissible loss functions as the infimum of the Wasserstein risk over all loss functions~$\ell\in\mc L$, defined as in~\eqref{optimalwcr} with $\mc P = \mc{W}_{\rho}(\Pnom)$.

Distributionally robust optimization with Wasserstein ambiguity sets enjoys increasing popularity in economics and operations research because it offers attractive out-of-sample performance and asymptotic consistency guarantees while being computationally tractable \cite{ref:kuhn2019tutorial}. For example, if the risk measure~$\risk_\QQ$ coincides with the expected value under~$\QQ$, the loss function~$\ell$ is representable as a pointwise maximum of finitely many concave functions, $\Pnom$ is discrete and~$\mc S=\mc M_2$, then, under mild technical conditions, the optimal Wasserstein risk $\risk_{\mc{W}_{\rho}(\Pnom)}(\Lc)$ can be computed by solving a tractable convex optimization problem~\cite[\S~6]{ref:zhen2019distributionally}. However, tractability results for more general risk measures, nominal distributions and structural ambiguity sets are scarce. And even if a tractable reformulation exists, its size typically scales with the cardinality of the support of~$\Pnom$. In this section we will show that the Gelbrich ambiguity set provides an outer approximation of the Wasserstein ambiguity set and that this approximation becomes exact if the structural ambiguity set~$\mc S$ is generated by the nominal distribution~$\Pnom$. This result is significant for several reasons. First, it implies that if the Wasserstein ambiguity set is designed as a confidence region for~$\PP$, then the corresponding Gelbrich ambiguity set also constitutes a confidence region for~$\PP$ with the same coverage probability. The Gelbrich risk thus inherits any known statistical guarantees for the Wasserstein risk. Moreover, this result will later allow us to construct tractable conservative approximations or exact tractable reformulations of the Wasserstein risk that do not grow with the cardinality of the support of the nominal distribution. We emphasize that these approximations and reformulations are available for a broad range of risk measures, nominal distributions and structural ambiguity sets for which there currently exist no tractability results.

To show that the Gelbrich risk upper bounds the Wasserstein risk, we recall that the Gelbrich distance provides a lower bound on the 2-Wasserstein distance between two probability distributions that depends exclusively on their mean vectors and covariance matrices. 

\begin{theorem}[{Gelbrich bound \cite[Theorem~2.1]{ref:gelbrich1990formula}}]
    \label{theorem:gelbrich}
    For any distributions $\QQ_1,\QQ_2\in\mc M_2$ with mean vectors $\mu_1$, $\mu_2 \in \R^n$ and covariance matrices $\cov_1$, $\cov_2 \in \PSD^n$, respectively, we have $\Wass(\QQ_1, \QQ_2) \geq \Gelbrich ((\m_1, \cov_1), (\m_2, \cov_2))$.
\end{theorem}

The Gelbrich bound of Theorem~\ref{theorem:gelbrich} may be useful when the exact Wasserstein distance between two probability distributions is inaccessible. Indeed, computing Wasserstein distances is generically $\#$P-hard \cite[Theorem~2.2]{ref:taskesen2019complexity}. Even though the Gelbrich distance is non-convex, the squared Gelbrich distance is jointly convex in both of its arguments. This is evident from the proof of Theorem~\ref{theorem:gelbrich}, which shows that $\Gelbrich^2 ( (\m_1, \cov_1), (\m_2, \cov_2) )$ equals the optimal value of a semidefinite program, and because convexity is preserved under partial minimization \cite[Proposition~3.3.1]{ref:bertsekas2009convex}.  Maybe surprisingly, the Gelbrich bound is tight in several cases of practical interest. 

\begin{theorem}[Tightness of the Gelbrich bound] \label{theorem:Wasserstein=Gelbrich}
    Suppose that the distributions $\QQ_1,\QQ_2\in\mc M_2$ have mean vectors $\mu_1, \mu_2 \in \R^n$ and covariance matrices $\cov_1, \cov_2 \in \PSD^n$, respectively. If~$\cov_1\succ 0$ and $\QQ_2$ is a positive semidefinite affine pushforward of~$\QQ_1$, then we have \[ \Wass(\QQ_1, \QQ_2) = \Gelbrich \big((\m_1, \cov_1), (\m_2, \cov_2) \big).\]
\end{theorem}

Recall that $\QQ_2$ is a positive semidefinite affine pushforward of~$\QQ_1$ if there exists an affine function $f(\xi) = A \xi + b$ with  $A \in \PSD^n$ and $b \in \R^n$  such that $\QQ_2 = \QQ_1 \circ f^{-1}$. The proof of Theorem~\ref{theorem:Wasserstein=Gelbrich} reveals that $f$ is uniquely determined by $\m_1$, $\m_2$, $\cov_1$ and~$\cov_2$ via the relations
\begin{equation}
\label{eq:optimal-affine}
A = \cov_1^{-\half} \big( \cov_1^{\half} \cov_2 \cov_1^{\half} \big)^{\half} \cov_1^{-\half} \qquad \text{and} \qquad b = \mu_2 - A \mu_1.
\end{equation}
Note that the inverse of $\cov_1^{-\frac{1}{2}}$ exists because of our assumption that $\cov_1$ is positive definite. We emphasize, however, that Theorem~\ref{theorem:Wasserstein=Gelbrich} remains valid if $\cov_1$ is only positive {\em semi}definite and if $\QQ_2 \circ P_{\cov_1}^{-1} = \QQ_1 \circ f^{-1}$, where $f$ is parametrized as in~\eqref{eq:optimal-affine} with $\cov_1^{-1}$ representing the Moore-Penrose inverse of $\cov_1$, while $P_{\cov_1}$ denotes the orthogonal projection onto the column space of $\cov_1$ \cite[Theorem~2.1]{ref:gelbrich1990formula}. To keep this paper self-contained, we prove Theorem~\ref{theorem:Wasserstein=Gelbrich}, which is weaker but sufficient for our purposes, in the online appendix. 

If~$\QQ$ belongs to the structural ambiguity set generated by~$\Pnom$, then it constitutes a positive semidefinite affine pushforward of~$\Pnom$. If in addition~$\covsa\succ 0$, then Theorem~\ref{theorem:Wasserstein=Gelbrich} implies that the 2-Wasserstein distance between~$\QQ$ and~$\Pnom$ coincides with the Gelbrich distance between their mean-covariance pairs.

We now show that $\mc U_\rho(\msa, \covsa)$ covers the projection of the 2-Wasserstein ball~$\mc W_{\rho}(\Pnom)$ onto the space of mean-covariance pairs. A sharper results is available if~$\covsa \succ 0$, in which case~$\mc U_\rho(\msa, \covsa)$ coincides exactly with the aforementioned projection.

\begin{proposition}[Projection of $\mc W_{\rho}(\Pnom)$ onto the mean-covariance space]
    \label{prop:gelbrich-projection} 
    If the nominal distribution $\Pnom$ has mean $\msa\in\R^n$ and covariance matrix $\covsa\in \PSD^n$, then we have
    \begin{equation}
    \label{eq:Gelbrich-inclusion}
    \left\{ \left(\EE_{\QQ}[\xi] , \EE_\QQ[(\xi-\EE_{\QQ}[\xi])(\xi-\EE_{\QQ}[\xi])^\top] \right) : \QQ\in \mc W_{\rho}(\Pnom)\right\} \subseteq \U_{\rho}(\msa, \covsa).
    \end{equation}
    If in addition $\covsa \succ 0$, then the inclusion becomes an equality.
\end{proposition}
The above results culminate in the following main theorem.

\begin{theorem}[Relation between the Gelbrich and Wasserstein ambiguity sets]
    \label{theorem:gelbrich:amb}
    If the nominal distribution $\Pnom$ has mean $\msa \in \R^n$ and covariance matrix~$\covsa \in \PSD^n$, then we have $\mc W_{\rho}(\Pnom) \subseteq \mc{G}_{\rho}(\msa, \covsa)$. In addition, if $\mc S$ is the structural ambiguity set generated by~$\Pnom$ and if $\covsa \succ 0$, then the inclusion becomes an equality.
\end{theorem}

Theorem~\ref{theorem:gelbrich:amb} implies that the Gelbrich ambiguity set~$\mc{G}_{\rho}(\msa, \covsa)$ constitutes an outer approximation for the Wasserstein ambiguity set~$\mc W_{\rho}(\Pnom)$ that ignores any information about the nominal distribution except for its structure, mean and covariance matrix. Discarding all higher-order moments can be interpreted as a compression of the available information. Below we will argue that this compression can be leveraged to construct tractable reformulations of the (optimal) Wasserstein risk. An immediate consequence of Theorem~\ref{theorem:gelbrich:amb} is that the (optimal) Gelbrich risk provides an upper bound on the (optimal) Wasserstein risk. We formalize this insight in the following corollary, which we state without proof. 

\begin{corollary}[Gelbrich risk versus Wasserstein risk]
    \label{cor:gelbrich-risk}
    If the nominal distribution~$\Pnom$ has mean $\msa \in \R^n$ and covariance matrix $\covsa \in \PSD^n$, then we have $\risk_{\mc W_{\rho}(\Pnom)}(\ell)\le \risk_{\mc G_{\rho}(\msa, \covsa)}(\ell)$ for all $\ell\in\Lc$ and $\risk_{\mc W_{\rho}(\Pnom)}( \Lc) \le \risk_{\mc G_{\rho}(\msa, \covsa)}(\Lc)$. In addition, if $\mc S$ is the structural ambiguity set generated by $\Pnom$ and if $\covsa \succ 0$, then these inequalities become equalities.
\end{corollary}

Corollary~\ref{cor:gelbrich-risk} implies that any finite-sample guarantee for the (optimal) Wasserstein risk immediately leads to a finite-sample guarantee for the (optimal) Gelbrich risk. For example, if~$\Pnom$ is set to the discrete empirical distribution of the training sample, then the measure concentration results by \cite{ref:fournier2015rate} can be used to calibrate the radius~$\rho$ to the sample size so that the unknown true distribution~$\PP$ belongs to the Wasserstein ambiguity set with probability~$1-\eta$ for any given~$\eta\in (0,1)$ \cite[Theorem~3.4]{ref:esfahani2018data}. For this choice of~$\rho$, the Wasserstein and the Gelbrich risk exceed the true risk with probability at least~$1-\eta$. However, the underlying measure concentration results suffer from a fundamental curse of dimensionality, which implies that~$\rho$ must decay extremely slowly with the sample size even if~$\xi$ has moderate dimension. In the next section we will demonstrate that this curse of dimensionality can be circumvented by working with the Gelbrich instead of the Wasserstein ambiguity set.

\subsection{Statistical Properties of the Gelbrich Ambiguity Set}
\label{sect:guarantee}
    
We now show that the Gelbrich distance enjoys attractive measure concentration properties, which make it ideal to construct ambiguity sets for sub-Gaussian distributions. These measure concentration properties thus lend themselves for calibrating the radius~$\rho$ of the Gelbrich ambiguity set~$\mc{G}_{\rho}(\msa, \covsa)$. Throughout this section we denote by~$\PP$ the unknown true distribution of the vector~$\xi\in\R^n$ of risk factors, and we assume that~$\PP$ has a finite mean~$\mu=\EE_\PP[\xi]$, second moment matrix~$M=\EE_\PP[\xi\xi^\top]$ and covariance matrix~$\cov=M-\mu\mu^\top$. Similarly, we use $\widehat \mu_N$ to denote an estimator for $\mu$ and $\widehat M_N$ to denote an estimator for $M$ constructed from $N$ independent data points $\xi_1, \dots, \xi_N$ sampled from $\mathbb P$. These estimators can then be combined to construct an estimator $\covsa_N = \widehat M_N - \msa_N \msa_N^\top$ for $\cov$. The following theorem shows that any finite-sample guarantees for the estimators $\msa_N$ and $\widehat M_N$ give rise to a finite-sample guarantee for the Gelbrich distance between $(\msa_N, \covsa_N)$ and $(\mu,\cov)$.

\begin{theorem}[Finite-sample guarantees]
    \label{theorem:finite-sample}
    
    Suppose that the estimators $\widehat \mu_N$ and $\widehat M_N$ satisfy
    \begin{align*}
        \PP^N [ \| \msa_N - \mu \| \leq \rho_\mu(\eta_\mu) ] \geq 1 - \eta_\mu \quad \text{and} \quad \PP^N [ \| \widehat M_N - M \| \leq \rho_M(\eta_M) ] \geq 1 - \eta_M
    \end{align*}
    for any significance levels $\eta_\mu, \eta_M \in (0, 1]$, where $\rho_\mu$ and $\rho_M$ are non-increasing functions from $(0, 1]$ to $\R_+$, and set $\covsa_N = \widehat M_N - \msa_N \msa_N^\top$. 
    If $\cov \succ 0$, then there exist positive constants $c_1, c_2$ and $c_3$ that depend on $\PP$ only through $\mu$, $\cov$ and $n$ such that for any~$\eta_\mu, \eta_M \in (0, 1]$ we have
    \begin{align*}
    \PP^N \left[ \Gelbrich \big( (\msa_N, \covsa_N), (\mu, \cov) \big) \leq \rho(\eta_\mu, \eta_M) \right] \geq 1-\eta_\mu - \eta_M,
    \end{align*}
    where $\rho(\eta_\mu, \eta_M)= c_1 \rho_{\mu}(\eta_\mu) + c_2 \rho_{\mu}(\eta_\mu)^2 + c_3 \rho_M(\eta_M)$.
\end{theorem}
Theorem~\ref{theorem:finite-sample} guarantees that if $\rho\geq \rho(\eta_\mu, \eta_M)$, then the Gelbrich ambiguity set~$\mc{G}_{\rho}(\msa, \covsa)$ contains the unknown data-generating distribution~$\PP$ with probability at least $1-\eta_\mu - \eta_M$ under $\PP^N$. Moreover, if $\rho_\mu$ and $\rho_M$ scale as $\mathcal O(1 / \sqrt{N})$, then $\rho(\eta_\mu, \eta_M)$ scales as $\mathcal O(1 / \sqrt{N})$, too. % ensures that the Gelbrich ambiguity set contains the true distribution with high probability.
This happens, for example, if $\msa_N$ is the sample mean and $\widehat M$ is the sample second moment matrix of a sub-Gaussian distribution~$\PP$.   

\begin{definition}[Sub-Gaussian probability distribution]
    \label{definition:sub-Gaussian}
    The probability distribution $\PP$ of $\xi\in\R^n$ is sub-Gaussian with variance proxy~$\sigma^2 \in \R_+$ if 
    \[ 
    \EE_{\PP} \left[ \exp \left( z^\top \!\! \left( \xi - \EE_{\PP} \left[\xi\right] \right) \right) \right] \leq \exp \left({\textstyle\frac{1}{2}}\| z \|^2 \sigma^2 \right) \quad \forall z \in \R^n. 
    \]
\end{definition}

\begin{corollary}[Empirical estimators]
    \label{corollary:finite-sample:SAA}
    If $\PP$ is sub-Gaussian with mean $\mu$, covariance matrix $\cov \succ 0$, and variance proxy $\sigma^2$, and if %$\msa_N$ the sample mean, by $\widehat M_N$ the sample second moment, and set by $\covsa_N$ the sample covariance matrix corresponding to~$N$ independent points sampled from $\PP$, that is,
    \begin{align*}
        \msa_N = \frac{1}{N} \sum_{i=1}^N  \xi_i, \quad \wh M_N = \frac{1}{N} \sum_{i=1}^N  \xi_i \xi_i^\top \quad \text{and} \quad \covsa_N = \wh M_N-\msa_N\msa_N^\top
    \end{align*}
    are the sample mean, the sample second moment matrix and the sample covariance matrix, respectively, then there exist $c_1,c_2>0$ that depend on $\PP$ only through $\mu$, $\cov$, $\sigma^2$ and $n$ such that for any~$\eta \in (0, 1]$ we have $\PP^N [ \Gelbrich ( (\msa_N, \covsa_N), (\mu, \cov) ) \leq \rho(\eta) ] \geq 1-\eta$, where $\rho(\eta)= \left(c_1 + c_2 \log(1/\eta)\right)/\sqrt{N}$.
\end{corollary}    

Corollary~\ref{corollary:finite-sample:SAA} implies that if we use the na\"ive empirical estimators for the first and second moments of~$\PP$ and if $\rho\geq \rho(\eta)$, % scales as $\mathcal O(1 / \sqrt{N})$,
%$\rho\geq (c_1 + c_2 \log(1/\eta))/\sqrt{N}$, 
then the Gelbrich ambiguity set~$\mc{G}_{\rho}(\msa, \covsa)$ constitutes a $(1-\eta)$-confidence set for the true distribution~$\PP$.  %at least~$1-\eta$ under~$\PP^N$. This in turn immediately implies that the Gelbrich risk~$\risk_{\mc{G}_{\rho}(\msa, \covsa)}(\ell)$ of 
Consequently, the Gelbrich risk~$\risk_{\mc{G}_{\rho}(\msa, \covsa)}(\ell)$ provides a $(1-\eta)$-upper confidence bound on the true risk~$\risk_\PP(\ell)$ for any fixed loss function~$\ell\in\mathcal L_0$, and the optimal Gelbrich risk~$\risk_{\mc{G}_{\rho}(\msa, \covsa)}(\mathcal L)$ provides a $(1-\eta)$-upper confidence bound on the true optimal risk~$\inf_{\ell \in\mathcal L}\risk_\PP(\ell)$. %, with probability at least~$1-\eta$ under~$\PP^N$. 

Corollary~\ref{corollary:finite-sample:SAA} provides the first {\em non-asymptotic} concentration bound for the Gelbrich distance between the empirical and true mean-covariance pairs, based on data points from {\em any sub-Gaussian} distribution. Prior work, such as \cite[Theorem~2.1]{ref:rippl2016limit}, studied only the {\em asymptotic} behavior of the Wasserstein distance between the empirical and true distributions (or, equivalently, the Gelbrich distance between the corresponding mean-covariance pairs), assuming that the data points follow a Gaussian distribution. The asymptotic convergence rate was shown to be $\mathcal O(1/\sqrt{N})$. In contrast, Corollary~\ref{corollary:finite-sample:SAA} proves a {\em finite-sample} guarantee with the same rate for {\em any} sub-Gaussian distribution.

The sample mean and the sample covariance matrix are not appropriate for dealing with heavy-tailed distributions or high-dimensional data. There is a vast literature on robust estimators for the mean that are based on the median of means principle \cite{ref:nemirovskij1983problem,ref:catoni2012challenging}. They are harder to compute than the sample mean but still offer finite-sample guarantees with error terms that scale as $\mathcal O(1 / \sqrt{N})$ even though the data-generating distribution may fail to be sub-Gaussian \cite{ref:lugosi2019mean}. In addition, there are sophisticated covariance estimators based on factor models~\cite{ref:green1992will,ref:chen2016efficient}, shrinkage estimators \cite{ref:ledoit2003improved,ledoit2004well,demiguel2013size,ref:nguyen2018distributionally}, M-estimators \cite{maronna1976robust, huber1981robust, ref:demiguel2009portfolio} and S-estimators \cite{lopuhaa1989relation, campbell1998calculation, ref:demiguel2009portfolio}; see also \cite{ref:fan2016overview}. These estimators are again harder to compute than the sample covariance matrix, but they offer $\mathcal O(1 / \sqrt{N})$ finite-sample guarantees even for distributions that fail to be sub-Gaussian. Theorem~\ref{theorem:finite-sample} indicates that any of these robust estimators can be used to construct a Gelbrich ambiguity set that offers rigorous finite-sample guarantees.

\subsection{Computational Properties of the Gelbrich Ambiguity Set}
\label{sect:computational-perperties}

We now establish basic properties of the Gelbrich ambiguity set that are conducive to the solvability and computational tractability of %the worst-case risk evaluation and the 
distributionally robust optimization problems with Gelbrich ambiguity sets. %risk $\risk_{\mc{G}_{\rho}(\msa, \covsa)}(\ell)$ and and the optimal Gelbrich risk $\risk_{\mc{G}_{\rho}(\msa, \covsa)}(\Lc)$. 
To this end, we first study the ball $\mc U_\rho(\msa, \covsa)$ in the space of mean-covariance pairs. It is natural to expect that this set is compact and convex because $\Gelbrich$ represents a metric on~$\R^n \times \PSD^n$. The next proposition formalizes this intuition.

\begin{proposition}[Properties of $\mathcal U_{\rho}(\msa, \covsa)$] \label{prop:compact:U}
    For any $\msa \in \R^n$, $\covsa \in \PSD^n$ and $\rho \in \R_+$, the set $\mathcal U_{\rho}(\msa, \covsa)$ is compact and convex.
\end{proposition}

Consider now the Gelbrich ambiguity set $\mc G_\rho(\msa, \covsa)$ of Definition~\ref{def:gelbrich-hull}. Note that the mean $\mu= \EE_{\QQ}[\xi]$ and the second moment matrix $M=\EE_{\QQ}[\xi\xi^\top]$ are linear in the underlying probability distribution~$\QQ$, whereas the covariance matrix $\cov=M-\mu\mu^\top$ is indefinite quadratic in~$\QQ$. The constraint requiring the mean and covariance matrix of~$\QQ$ to fall into the convex set~$\mathcal U_\rho(\msa, \covsa)$ thus appears to be non-convex in~$\QQ$. One might therefore suspect that the Gelbrich ambiguity set is non-convex and that evaluting the Gelbrich risk $\risk_{\mc{G}_{\rho}(\msa, \covsa)}(\ell)$ is hard. We will now show that the Gelbrich ambiguity set is nonetheless convex. To this end, we define the following transform of $\mathcal U_{\rho}(\msa, \covsa)$, 
\[
\mathcal V_{\rho}(\msa, \covsa) = 
\left\{ (\m,M) \in \R^n\times \PSD^n:  (\mu, M-\mu\mu^\top) \in \mathcal U_{\rho}(\msa, \covsa) \right\}.
\]
Even though it is constructed as the pre-image of a convex set under an indefinite quadratic transformation, one can show that $\mathcal V_{\rho}(\msa, \covsa)$ is convex. 

\begin{proposition}[Properties of $\mathcal V_{\rho}(\msa, \covsa)$] \label{prop:compact:V}
    For any $\msa \in \R^n$, $\covsa \in \PSD^n$ and $\rho \in \R_+$, the set $\mathcal V_{\rho}(\msa, \covsa)$ is compact and convex.
\end{proposition}

Proposition~\ref{prop:compact:V} implies that the Gelbrich ambiguity set~$\mc G_\rho(\msa, \covsa)$ is convex because it can be viewed as the pre-image of the convex set~$\mathcal V_{\rho}(\msa, \covsa)$ under the linear transformation that maps any probability distribution to its mean and second moment matrix. We formalize this insight in the next corollary, which we state without proof.

\begin{corollary}[Convexity of $\mc G_\rho(\msa, \covsa)$]
    The Gelbrich ambiguity set is convex.
\end{corollary}

To close this section, we establish a decomposition of the Gelbrich ambiguity set that will prove useful for evaluating $\risk_{\mc{G}_{\rho}(\msa, \covsa)}(\ell)$ and $\risk_{\mc{G}_{\rho}(\msa, \covsa)}(\Lc)$.  
By definition, the Gelbrich ambiguity set $\mc{G}_{\rho}(\msa, \covsa)$ encompasses all distributions in $\mc S$ whose mean vectors and covariance matrices belong to $\mathcal U_{\rho}(\msa, \covsa)$. Denoting by $ \mathcal C(\mu,\Sigma)$ the {\em structured Chebyshev ambiguity set} that contains all distributions in $\mc S$ with mean~$\m$ and covariance matrix~$\cov$, the Gelbrich ambiguity set can be decomposed as 
\begin{equation}
\label{eq:gelbrich-union}
\mc{G}_{\rho}(\msa, \covsa) = \bigcup_{(\mu,\Sigma) \in \mathcal U_{\rho}(\msa, \covsa)} \mathcal C(\mu,\Sigma).
\end{equation}
In particular, we have $\mc{G}_{0}(\msa, \covsa)=\mc{C}(\msa, \covsa)$.
The decomposition~\eqref{eq:gelbrich-union} indicates that if $\mc{G}_{\rho}(\msa, \covsa)$ contains a particular distribution~$\QQ$, then it contains {\em all} distributions in $\mc S$ with the same mean and covariance matrix as $\QQ$. Moreover, it allows us to represent the Gelbrich risk of any loss function~$\ell\in\Lc_0$ as %the optimal value of the  nested optimization problem
\begin{subequations}\label{eq:two-layer}
    \begin{align}
    \label{eq:two-layer-a}
    \risk_{\mc{G}_{\rho}(\msa, \covsa)}(\ell)~&=\Sup{(\mu, \cov) \in \mathcal U_\rho(\msa, \covsa)} ~ \Sup{\QQ \in \mathcal C(\mu,\Sigma)} \risk_\QQ(\ell) \\
    & = \Sup{(\mu, M) \in \mathcal V_\rho(\msa, \covsa)} ~ \Sup{\QQ \in \mathcal C(\mu,M-\mu\mu^\top)} \risk_\QQ(\ell) \label{eq:two-layer-b},
    \end{align}
\end{subequations}
where the second equality holds because~$(\m,M)\in \mathcal V_\rho(\msa, \covsa)$ if and only if~$(\m,M-\m\m^\top)\in \mathcal U_\rho(\msa, \covsa)$. The reformulation~\eqref{eq:two-layer-b} suggests that the Gelbrich risk evaluation problem may be computationally tractable in situations of practical interest. To see this, note that the inner maximization problem in~\eqref{eq:two-layer-b} simply evaluates the worst-case risk over the Chebyshev ambiguity set of all probability distributions with mean~$\m$ and second moment matrix~$M$. If the risk measure~$\risk_\QQ(\ell)$ is linear in~$\QQ$ ({\em e.g.}, if it represents the expected loss) or concave in~$\QQ$ ({\em e.g.}, if it represents the variance, the VaR or the CVaR of the loss), then this inner problem constitutes a convex maximization problem over all probability distributions~$\QQ$ that satisfy the linear equality constraints $\EE_\QQ[\xi]=\mu$ and $\EE_\QQ[\xi\xi^\top]=M$. As concavity is preserved under partial maximization, the optimal value of the inner problem in~\eqref{eq:two-layer-a} is jointly concave in the constraint right hand sides~$\m$ and~$M$. The  outer problem in~\eqref{eq:two-layer-b} thus maximizes a concave function ({\em i.e.}, the worst-case Chebyshev risk) over all mean vectors and second moment matrices in the convex set~$\mathcal V_\rho(\msa, \covsa)$; see Proposition~\ref{prop:compact:V}. Consequently, both the inner and the outer maximization problems in~\eqref{eq:two-layer-b} are convex, and thus there is hope that both of them are computationally tractable. We will further investigate the decomposition~\eqref{eq:two-layer} in Section~\ref{sec:linearportfolio}.

Conceptually, the inner problems in~\eqref{eq:two-layer-a} and~\eqref{eq:two-layer-b} hedge against uncertainty in the shape and the outer problems hedge against uncertainty in the location and dispersion of the distribution of the risk factors.
We are not the first to study two-layer distributionally robust optimization problems with an outer layer that hedges against mean-covariance uncertainty. As the worst-case risk over a Chebyshev ambiguity set is concave in~$(\mu, M)$ but non-concave in~$(\mu,\cov)$ for most common risk measures, moment uncertainty has mostly been modeled through convex uncertainty sets for~$(\mu,M)$. This choice leads to convex outer-layer problems. For example, uncertainty sets that restrict~$\mu$ to an ellipsoid and~$M$ to the intersection of two positive semi-definite cones were proposed by~\cite{ref:delage2010distributionally}, whereas rectangular uncertainty sets for~$(\mu,M)$ were studied by~\cite{ref:zymler2013wcvar} and~\cite{ ref:hanasusanto2015distributionally}. Generic convex uncertainty sets for~$(\mu,\cov)$ render the outer-layer problems convex only in special situations, {\em e.g.}, when the loss function is quadratic and the VaR is used as a risk measure; see~\cite{ref:ghaoui2003worst, ref:rujeerapaiboon2016robust}. Our new convex uncertainty set~$\mathcal U_\rho(\msa, \covsa)$ for $(\mu,\cov)$ is not only remarkable due to its connection to the Wasserstein ambiguity set but also because it leads to a convex outer-layer problem in~\eqref{eq:two-layer-a} irrespective of the loss function (as long as the risk measure is concave in~$\QQ$).

\section{Gelbrich Risk of Linear Portfolio Loss Functions}
\label{sec:linearportfolio}

We now derive explicit formulas for the Gelbrich risk of linear portfolio loss functions of the form~$\ell(\xi) = -\w^\top \xi$, where~$\xi$ stands for the vector of asset returns, and~$w$ collects the portfolio weights. Thus,~$\ell(\xi)$ represents the negative portfolio return. As a preparation, Section~\ref{sec:risk-measures} reviews several basic properties of risk measures. Section~\ref{sec:regularized-mean-deviation} then shows that if the risk measure at hand is law-invariant, translation invariant and positive homogeneous (but not necessarily convex) and the structural ambiguity set satisfies a stability condition, then the Gelbrich risk simplifies to a regularized mean-standard deviation risk measure, which is convex and can be minimized efficiently. In addition, we analytically characterize the extremal distributions that attain the supremum in evaluating $\risk_{\mc{G}_{\rho}(\msa, \covsa)}(\ell)$. Remarkably, the Gelbrich risk, its optimal portfolios as well as the corresponding extremal distributions depend on the underlying risk measure only through a scalar, which we term the {\em standard risk coefficient}, and which can be calculated offline. In Section~\ref{sec:standard-risk} we thus provide closed-form expressions for the standard risk coefficients of the VaR, the CVaR and the mean-standard deviation risk measure. We also derive the standard risk coefficients of all spectral risk measures, all risk measures that admit a Kusuoka representation and all distortion risk measures. The online appendix shows that most results of Sections~\ref{sec:regularized-mean-deviation}--\ref{sec:standard-risk} extend to the mean-variance risk measures, even though they fail to be positive homogeneous.

\subsection{Basic Properties of Risk Measures}
\label{sec:risk-measures}
Virtually all risk measures used in economics and finance are law-invariant~\cite[\S~4.5]{ref:follmer2008stochastic}. Such risk measures are usually defined in view of the probability distribution~$\PP$ of the relevant risk factors. As we study situations in which~$\PP$ is ambiguous, we now extend the notion of law-invariance to {\em families} of risk measures $\{\risk_\QQ\}_{\QQ \in \mc M}$. %parameterized by the probability distributions in the structural ambiguity set~$\mc M$.

\begin{definition}[Law-invariant family of risk measures] \label{def:law-invariant}
    The family of risk measures $\{\risk_\QQ\}_{\QQ \in \mc M}$ is law-invariant if~$\risk_{\QQ_1}(\ell_1) = \risk_{\QQ_2}(\ell_2)$ for any loss functions $\ell_1,\ell_2\in\Lc_0$ and probability distributions $\QQ_1, \QQ_2 \in \mc M$ such that the distribution of~$\ell_1(\xi)$ under~$\QQ_1$ matches the distribution of $\ell_2(\xi)$ under~$\QQ_2$. 
\end{definition}

Note that if the family of risk measures $\{\risk_\QQ\}_{\QQ \in \mc M}$ is law-invariant in the sense of Definition~\ref{def:law-invariant}, then the risk measure $\risk_{\QQ}$ is law-invariant in the sense of \cite[\S~4.5]{ref:follmer2008stochastic} for any fixed~$\QQ \in \mc M$. Conversely, the following remark shows that any law-invariant risk measure~$\risk_\PP$ associated with a continuous probability distribution~$\PP \in \mc M$ naturally induces a law-invariant family of risk measures~$\{ \risk_\QQ \}_{\QQ \in \mc M}$.

\begin{remark}[Constructing a law-invariant family of risk measures]
    \label{remark:construction}
    Assume that~$\risk_\PP$ is a law-invariant risk measure associated with a continuous probability distribution~$\PP\in\mc M$, and let~$\QQ\in\mc M$ be any other probability distribution. In particular, $\QQ$ does not have to be continuous. Denote by~$\varphi_\PP:\R^n\rightarrow[0,1]^n$ the Rosenblatt transformation corresponding to~$\PP$ \cite{ref:rosenblatt1952} and by $\psi_\QQ:[0,1]^n\rightarrow \R^n$ the inverse Rosenblatt transformation corresponding to~$\QQ$ \cite[\S~2.5]{ref:chen2011markov}. As~$\PP$ is continuous, one can show that~$\varphi_\PP(\xi)$ is uniformly distributed on~$[0,1]^n$ under~$\PP$. In addition, $\psi_\QQ(\varphi_\PP(\xi))$ follows the distribution~$\QQ$ under~$\PP$. We can now define a risk measure~$\risk_\QQ$ corresponding to~$\QQ$ by setting~$\risk_\QQ(\ell)=\risk_\PP(\ell(\psi_\QQ(\varphi_\PP(\xi))))$ for all~$\ell\in\Lc_0$. The family~$\{ \risk_\QQ \}_{\QQ \in \mc M}$ constructed in this way is law-invariant in the sense of Definition~\ref{def:law-invariant} as~$\risk_\PP$ is law-invariant in the usual~sense.
\end{remark}

Note that `risk measures' in colloquial English ({\em e.g.}, the `variance,' the `VaR' or the `CVaR' etc.) make no reference to a specific probability distribution and are therefore naturally interpreted as families of risk measures of the form~$\{\risk_\QQ\}_{\QQ \in \mc M}$. All of these standard families of risk measures are in fact law-invariant. We now recall some basic properties displayed by many popular risk measures.

\begin{definition}[Properties of risk measures] \label{def:risk:measure}
    A risk measure $\risk_\QQ$ associated with a probability distribution~$\QQ\in\mc M$ is
    \begin{itemize}[leftmargin=.2in]
        \item[$\diamond$] translation invariant if $\risk_\QQ(\ell + \lambda) = \risk_\QQ(\ell) + \lambda$ for all $\ell \in \mc L_0$, $\lambda \in \R$;
        \item[$\diamond$] positive homogeneous if $\risk_\QQ(\lambda \ell) = \lambda \risk_\QQ(\ell)$ for all $\ell \in \mc L_0$, $\lambda \in\R_+$;
        \item[$\diamond$] monotonic if $\risk_\QQ(\ell_1) \leq \risk_\QQ(\ell_2)$ for all $\ell_1, \ell_2 \in \mc L_0$ such that $\ell_1 \leq \ell_2$ $\QQ$-almost surely;
        \item[$\diamond$] convex if $\risk_\QQ(\lambda \ell_1 + (1- \lambda) \ell_2) \leq \lambda \risk_\QQ( \ell_1) + (1-\lambda) \risk_\QQ(\ell_2)$ for all $\ell_1, \ell_2 \in \mc L_0$, $\lambda \in [0, 1]$.
    \end{itemize}
    A risk measure is called coherent if it satisfies all of the above properties.
\end{definition}

\subsection{Law Invariant, Translation Invariant and Positive Homogeneous Risk Measures}
\label{sec:regularized-mean-deviation}

We will now demonstrate that many commonly used families of risk measures impact the Gelbrich risk of a linear loss function only through a scalar, which we define as follows.

\begin{definition}[Standard risk coefficient]
\label{def:alpha}
    The standard risk coefficient of a family~$\{ \risk_\QQ\}_{\QQ \in \mc M}$ of risk measures corresponding to a structural ambiguity set~$\mc S$, mean~$\m \in \R^n$, covariance matrix~$\cov \in \PSD^n$ and portfolio vector~$\w \in \R^n$ with~$\w^\top \cov \w \neq 0$ is
    \be
    \alpha(\mu,\cov,w) = \Sup{\QQ \in \mathcal C(\m, \cov)} \; \risk_\QQ \left( - \frac{\w^\top (\xi- \m)}{\sqrt{\w^\top \cov \w}} \right). \label{eq:alpha:def}
    \ee
\end{definition} 

Recall that $\mathcal C(\m, \cov)$ denotes the structured Chebyshev ambiguity set of all distributions in~$\mathcal{S}$ with mean~$\m$ and covariance matrix~$\cov$. For generic families of risk measures, the standard risk coefficient of Definition~\ref{def:alpha} depends on~$\m$, $\cov$ and~$\w$. However, if the family of risk measures is law-invariant and the structural ambiguity set is stable in the sense of the following definition, then the standard risk coefficient is constant in these parameters and depends solely on the family of risk measures and the structural ambiguity set at hand. 

\begin{definition}[Stable structural ambiguity set]
    \label{def:stable}
    A structural ambiguity set~$\mathcal S$ is stable if it is closed under arbitrary affine pushforwards and convolutions. Thus, if~$\QQ \in \mc S$ and $f: \R^n \to \R^n$ is of the form $f(\xi) = A\xi + b$ for some $A \in \mathbb R^{n \times n}$ and $b \in \R^n$, then~$\QQ \circ f^{-1}\in \mc S$. Similarly, if~$\mathbb Q_1, \mathbb Q_2 \in \mathcal S$, then~$\mathbb Q_1 * \mathbb Q_2\in\mc S$, where the convolution is defined through $\mathbb Q_1 * \mathbb Q_2(B) = \int_{\mathbb R^n} \int_{\mathbb R^n} \mathds 1_{\xi_1 + \xi_2\in B}\, \diff \mathbb Q_1(\xi_1) \,\diff \mathbb Q_2(\xi_2) $ for all Borel sets~$B \in \mc B(\mathbb R^n)$.
\end{definition}

To motivate our terminology, recall that a distribution is called stable if any linear combination of two independent random variables with this distribution has the same distribution up to location and scaling. For example, the Gaussian, Cauchy and L\'{e}vy distributions are stable. Definition~\ref{def:stable} generalizes this notion to ambiguity sets. Indeed, it implies that if two independent random vectors~$\xi_1$ and~$\xi_2$ have distributions~$\mathbb Q_1$ and~$\mathbb Q_2$, respectively, both of which belong to the same stable structural ambiguity set~$\mathcal S$, and if~$A_1,A_2 \in \mathbb R^{n \times n}$, then the probability distribution of the linear combination~$\xi = A_1 \xi_1 + A_2 \xi_2$ also belongs to~$\mathcal S$. The structural ambiguity set~$\mathcal M_2$ of all distributions with finite second moment is trivially stable. By~\cite[\S~2]{ref:yu2009projection}, the sets of all symmetric, all symmetric linear unimodal and all log-concave distributions with finite second moments also constitute stable ambiguity sets. In addition, the ambiguity set of all Gaussian distributions is stable. However, some structural ambiguity sets fail to be stable. The set of all elliptical distributions with the same characteristic generator, for example, is not necessarily stable. Indeed, the convolution of two Laplace distributions is not a Laplace distribution, for instance. One can also show that the structural ambiguity set of all linear unimodal (but not necessarily symmetric) distributions also fails to be stable even in the univariate case. One can further show that the structural ambiguity set generated by a distribution~$\Pnom$ is stable only if~$\Pnom$ is Gaussian.
We can now state the announced result, which is inspired by~\cite[Theorems~1 and~2]{ref:yu2009projection}.

\begin{proposition}[Standard risk coefficient] \label{proposition:alpha}
    If $\{\risk_\QQ\}_{\QQ \in \mc M}$ is a law-invariant family of risk measures and the structural ambiguity set~$\mathcal S$ is stable, then the corresponding standard risk coefficient $\alpha$ is independent of~$\m$, $\cov$ and~$\w$.
\end{proposition}

Proposition~\ref{proposition:alpha} is a key ingredient to prove our following main result.

\begin{theorem}[Gelbrich risk of linear loss functions]
    \label{theorem:meanstd}
    If $\{\risk_\QQ \}_{\QQ \in \mc M}$ is a law-invariant family of translation invariant and positive homogeneous risk measures, the structural ambiguity set~$\mc S$ is stable and the corresponding standard risk coefficient satisfies $0 \leq \alpha < +\infty$, then the Gelbrich risk of the portfolio loss function~$\ell(\xi) = -\w^\top \xi$ is given~by
    \be    \label{eq:MS:1}
        \Sup{\QQ \in \mc{G}_{\rho}(\msa, \covsa)} \; \risk_\QQ \left( - \w^\top \xi \right) =
        -\msa^\top \w + \alpha \sqrt{\w^\top \covsa \w } + \rho \sqrt{1+ \alpha^2}\, \| \w \|.
    \ee
    In addition, if $\mc S$ is the structural ambiguity set generated by a Gaussian nominal distribution~$\Pnom$ with~$\covsa \succ 0$, then the Wasserstein risk coincides with the Gelbrich risk.
\end{theorem}

We emphasize that the standard risk coefficient~$\alpha$ can in general be negative. In this case, evaluating the Gelbrich risk of $-\w^\top \xi$ requires the solution of a non-convex optimization problem, and Theorem~\ref{theorem:meanstd} no longer holds (see problem~\eqref{eq:primal} in the proof of Theorem \ref{theorem:meanstd}). A sufficient condition for the non-negativity of~$\alpha$ is described in the following proposition.
%if the structural ambiguity set contains at least one symmetric probability distribution and if~$\{\risk_\QQ\}_{\QQ \in \mc M}$ is a law-invariant family of coherent risk measures. 
\begin{proposition}[Non-negative standard risk coefficient]
    \label{prop:alpha:positive}
    If~$\{\risk_\QQ\}_{\QQ \in \mc M}$ is a law-invariant family of coherent risk measures and the structural ambiguity set~$\mathcal S$ contains a symmetric distribution, then the corresponding standard risk coefficient~$\alpha$ is non-negative.
\end{proposition}

From now on we assume that~$\{\risk_\QQ \}_{\QQ \in \mc M}$ is a law-invariant family of translation invariant and positive homogeneous risk measures and that~$\mc S$ is a stable ambiguity set with standard risk coefficient~$\alpha \ge 0$. Denoting by~$\Lc$ the set of portfolio loss functions~$\ell(\xi)=-w^\top \xi$ with portfolio weights~$w$ belonging to a set $\Omega\subseteq \R^n$, Theorem~\ref{theorem:meanstd} allows us to reformulate the optimal Gelbrich risk $\risk_{\mc{G}_{\rho}(\msa, \covsa)}(\Lc)$~as
\begin{align}
    \label{eq:optimize-w-gelbrich}
    \Min{\w \in \Omega} \Sup{\QQ \in \mc{G}_{\rho}(\msa, \covsa)} \risk_{\QQ}(-\w^\top \xi) 
    = \Min{\w \in \Omega} ~ - \msa^\top \w + \alpha \sqrt{\w^\top \covsa \w} + \rho \sqrt{1 + \alpha^2} \| \w\|.
\end{align}
Note that if~$\Omega$ is convex, then~\eqref{eq:optimize-w-gelbrich} constitutes a finite convex program. In particular, if~$\Omega$ is representable via second-order cone constraints, then problem~\eqref{eq:optimize-w-gelbrich} reduces to a tractable second-order cone program that can be solved highly efficiently with off-the-shelf solvers. As the Gelbrich risk upper bounds the Wasserstein risk by virtue of Corollary~\ref{cor:gelbrich-risk}, the convex program~\eqref{eq:optimize-w-gelbrich} provides a conservative and efficiently computable proxy for the optimal Wasserstein risk~$\risk_{\mc{W}_{\rho}(\Pnom)}(\Lc)$, which may be hard to compute exactly. Note also that~\eqref{eq:optimize-w-gelbrich} can be interpreted as a regularized Markowitz portfolio selection problem with an $\ell_2$-regularization term that scales with the size parameter~$\rho$ of the Gelbrich ambiguity set.

\begin{remark}[Portfolio constraints]
    There is a vast literature on the impact of portfolio constraints on performance. For example, it has been argued that the portfolio's out-of-sample risk can be reduced by imposing norm constraints~\cite{ref:jagannathan2003risk,ref:demiguel2009generalized}, no-short-sales constraints~\cite{ref:jagannathan2003risk} or combinations thereof~\cite{ref:kan2007optimal,ref:tu2011markowitz,ref:zhao2019portfolio}. Our Theorem~\ref{theorem:meanstd} implies that robustifying the risk measure with respect to a Gelbrich ambiguity set is equivalent to penalizing the portfolio's 2-norm---irrespective of the feasible set~$\Omega$. Thus, one may simultaneously robustify the risk measure and restrict the feasible set to improve performance. %to of the poWhile our regularization result is equivalent to imposing norm constraints on $w$, the rest of the constraints can still be incorporated into the set $\Omega$, and the equivalence in~\eqref{eq:optimize-w-gelbrich} still holds.
\end{remark}

Under mild conditions on~$\Omega$, one can show that the optimizer of the Gelbrich risk portfolio selection problem~\eqref{eq:optimize-w-gelbrich} converges to the equally weighted portfolio as~$\rho$ tends to infinity ({\em i.e.}, in the limit of extreme uncertainty). A similar result was proved by~\cite{ref:pflug20121/N} for a specific class of {\em convex} risk measures and for a Wasserstein ambiguity set.

\begin{corollary}[The equally weighted portfolio is optimal under high uncertainty]\label{corol:1/N}
    If the assumptions of Theorem~\ref{theorem:meanstd} hold, 
    $\Omega\subseteq \{w\in\mathbb R^n:e^\top \w = 1\}$ is a closed set of portfolio weights with $e\in\R^n$ being the vector of all ones, and if~$\Omega$ contains the equally weighted portfolio~$\frac{1}{n} e$, then the unique minimizer of~\eqref{eq:optimize-w-gelbrich} converges to~$\frac{1}{n}e$ as~$\rho$ tends to~$\infty$.
\end{corollary}

Worst-case distributions that maximize the risk of a fixed portfolio over the Gelbrich ambiguity set can expose potential threats to the portfolio or may be useful for stress test experiments. Therefore, we now aim to characterize the worst-case distributions~$\QQ\opt$ that attain the supremum in evaluating the Gelbrich risk~$\risk_{\mc{G}_{\rho}(\msa, \covsa)}(\ell)$ for~$\ell(\xi) = -\w^\top \xi$. %The next proposition shows that the first two moments of $\QQ\opt$ are unique and can be computed analytically.

\begin{proposition}[Worst-case moments]
    \label{prop:WC}
    If $\{\risk_\QQ \}_{\QQ \in \mc M}$ is a law-invariant family of translation invariant and positive homogeneous risk measures, the structural ambiguity set~$\mc S$ is stable, the standard risk coefficient satisfies $0 < \alpha < +\infty$ and~$\covsa \succ 0$, then any extremal distribution $\QQ\opt$ that attains the Gelbrich risk of the loss function~$\ell(\xi) = -\w^\top \xi$ has the same mean $\m\opt \in \R^n$ and covariance matrix $\cov\opt \in \PSD^n$, where
    \begin{align*}
    \m\opt &= \msa - \frac{\rho}{\sqrt{1+\alpha^2} \|\w\|} w \quad \text{and} \\
    \cov\opt &= \left( I + \frac{\rho \alpha \w\w^\top}{\sqrt{1+\alpha^2} \|\w\| \sqrt{\w^\top \covsa \w}} \right) \covsa \left( I + \frac{\rho \alpha \w\w^\top}{\sqrt{1+\alpha^2} \|\w\| \sqrt{\w^\top \covsa \w}} \right).
    \end{align*}
\end{proposition}

Proposition~\ref{prop:WC} characterizes only the first two moments of the extremal distributions that attain the Gelbrich risk. In general, $\mc S$ may contain multiple distributions with these moments. When~$\mathcal S$ is the set of all Gaussian distributions, however, the Gelbrich risk is uniquely attained by the Gaussian distribution with mean~$\mu^\star$ and covariance matrix~$\cov^\star$.

%We emphasize that the family~$\{\risk_\QQ \}_{\QQ \in \mc M}$ of risk measures and the structural ambiguity set~$\mc S$ impact the Gelbrich portfolio selection problem~\eqref{eq:optimize-w-gelbrich} only indirectly through the standard risk coefficient~$\alpha$, which can be computed offline. 

\subsection{Calculation of the Standard Risk Coefficient}
\label{sec:standard-risk}

We now show that the standard risk coefficient~$\alpha$ is given closed form for a large class of risk measures. First, we focus on the value-at-risk (VaR). For any probability distribution~$\QQ\in\mc M$, the VaR at level~$\beta \in (0, 1)$ of any loss function~$\ell \in\Lc_0$ is defined as
\[
\QQ\text{-}\VaR_{\beta} (\ell) = \inf \left\{ \tau \in \R : \QQ[\ell(\xi) \leq \tau] \geq 1-\beta \right\}.
\]
VaR fails to be convex, yet it is widely used by financial institutions and regulators~\cite{ref:jorion1996var, JPM_96, duf_pan_97}. In addition, VaR induces a law-invariant family of translation invariant and positive homogeneous risk measures, and thus
Proposition~\ref{proposition:alpha} and Theorem~\ref{theorem:meanstd} apply whenever the structural ambiguity set~$\mc S$ is stable. The following proposition describes situations in which~$\alpha$ is available in closed form.

\begin{proposition}[Standard risk coefficient for VaR] 
    \label{proposition:VaR}
    If~$\beta\in(0,1)$ and~$\risk_\QQ=\QQ\text{\em -VaR}_\beta$ for~$\QQ\in\mc M$, then~$\alpha$ is available in closed form for several stable structural ambiguity sets.
    \begin{enumerate}[label = (\roman*)]
        \item If~$\mc S=\mc M_2$, then $\alpha = \sqrt{(1-\beta)/\beta}$.
        \item If~$\mc S$ is the set of all symmetric distributions in~$\mc M_2$, then~$\alpha = \sqrt{1/(2\beta)}$ for~$\beta<\half$ and~$\alpha = 0$ for~$\beta\ge \half$.
        \item If~$\mc S$ is the set of all symmetric linear unimodal distributions in~$\mc M_2$, then~$\alpha = 2/(3\sqrt{2\beta})$ for~$\beta <\half$ and~$\alpha = 0$ for~$\beta \ge \half$.
        \item If~$\mc S$ is the set of all Gaussian distributions, then $\alpha = \Phi^{-1}(1-\beta)$, where~$\Phi$ denotes the cumulative distribution function of the standard Gaussian distribution.
    \end{enumerate}
\end{proposition}

In assertions~\emph{(i)}, \emph{(ii)} and~\emph{(iii)} of Proposition~\ref{proposition:VaR} the standard risk coefficient~$\alpha$ is non-negative for all~$\beta\in(0,1)$ even though Proposition~\ref{prop:alpha:positive} does not apply (as VaR fails to be convex). In assertion~\emph{(iv)}, on the other hand, $\alpha$ becomes negative for~$\beta>\half$. Hence, Theorem~\ref{theorem:meanstd} does {\em not} apply to the VaR at level~$\beta>\half$ if~$\mc S$ is the family of all Gaussian distributions.

We now address the conditional value-at-risk (CVaR). For any distribution~$\QQ\in\mc M$, the CVaR at level~$\beta \in (0, 1)$ of any loss function~$\ell \in\Lc_0$ is defined as
\[
\QQ\text{-}\CVaR_{\beta} (\ell) = \Inf{\tau \in \R} \left\{ \tau + \frac{1}{\beta} \EE_\QQ\left[ \max \{ \ell(\xi) - \tau, 0 \} \right] \right\}.
\]
It is well known that CVaR induces a law-invariant family of coherent risk measures \cite{ref:artzner1999coherent, ref:rockafellar2000optimization}, and thus
Proposition~\ref{proposition:alpha} and Theorem~\ref{theorem:meanstd} apply whenever the structural ambiguity set~$\mc S$ is stable. The following proposition shows that~$\alpha$ is again available in closed form in several situations of practical interest.

\begin{proposition}[Standard risk coefficient for CVaR]
    \label{proposition:CVaR}
    If~$\beta\in(0,1)$ and~$\risk_\QQ=\QQ\text{\em -CVaR}_\beta$ for~$\QQ\in\mc M$, then~$\alpha$ is available in closed form for several stable structural ambiguity sets.
    \begin{enumerate}[label = (\roman*)]
        \item If~$\mc S=\mc M_2$, then $\alpha = \sqrt{(1-\beta)/\beta}$.
        \item If~$\mc S$ is the set of all symmetric distributions in~$\mc M_2$, then~$\alpha = \sqrt{1/(2\beta)}$ for~$\beta< \half$ and~$\alpha = \sqrt{1-\beta}/(\sqrt{2}\beta)$ for~$\beta\ge \half$.
        \item If~$\mc S$ is the set of all symmetric linear unimodal distributions in~$\mc M_2$, then~$\alpha = 2/(3\sqrt{\beta})$ for~$\beta \leq \frac{1}{3}$, $\alpha = \sqrt{3}(1-\beta)$ for~$\frac{1}{3}<\beta \leq \frac{2}{3}$ and~$\alpha = 2\sqrt{1-\beta} /(3\beta)$ for~$\beta >\frac{2}{3}$.
        \item If~$\mc S$ is the set of all Gaussian distributions, then $\alpha = (\sqrt{2\pi} \beta)^{-1}\exp(-(\Phi^{-1}(1-\beta))^2/2)$, where~$\Phi$ denotes the cumulative distribution function of the standard Gaussian distribution.
    \end{enumerate}
\end{proposition}

Propositions~\ref{proposition:VaR}~\emph{(i)} and~\ref{proposition:CVaR}~\emph{(i)} imply that the standard risk coefficients for VaR and CVaR coincide if~$\mc S$ represents the family of all distributions with finite second moments. This result is reminiscent of the observation that distributionally robust chance constraints are equivalent to distributionally robust CVaR constraints when the distributional uncertainty is modeled by a Chebyshev ambiguity set~\cite[Theorem~2.2]{ref:zymler2013distributionally}. We emphasize that the standard risk coefficients for VaR and CVaR differ under the structural ambiguity sets of assertions~\emph{(ii)}, \emph{(iii)} and~\emph{(iv)} of Propositions~\ref{proposition:VaR} and~\ref{proposition:CVaR}, respectively.

Next, we study mean-standard deviation risk measures that are ubiquitous in classical portfolio theory~\cite{ref:rockafellar2002deviation}. For any distribution~$\QQ \in \mc M$, the mean-standard deviation risk measure with risk-aversion coefficient~$\beta \ge 0$ of any loss function~$\ell \in \Lc_0$ is defined as $\risk_\QQ(\ell) = \EE_\QQ[ \ell(\xi)] + \beta (\Var_\QQ(\ell( \xi)))^{1/2}$, where $\Var_\QQ(\ell(\xi))$ denotes the variance of $\ell(\xi)$ under $\QQ$. For any fixed~$\beta$, the mean-standard deviation risk measure induces a law-invariant family of translation invariant and positive homogeneous risk measures, and thus
Proposition~\ref{proposition:alpha} and Theorem~\ref{theorem:meanstd} apply if~$\mc S$ is stable. Note that this includes the expected loss, for $\beta=0$. Proposition~\ref{proposition:meanstd} below derives~$\alpha$ again in closed form.

\begin{proposition}[Standard risk coefficient for mean-standard deviation risk measures]
    \label{proposition:meanstd}
    If~$\risk_\QQ$ is the mean-standard deviation risk measure with risk-aversion coefficient~$\beta \ge 0$ for every~$\QQ\in\mc M$ and if the structural ambiguity set~$\mathcal S$ is stable, then~$\alpha = \beta$.
\end{proposition}

Consider now the family of spectral risk measures introduced by~\cite{ref:acerbi2002spectral}. In the following discussion, for any~$\ell \in\Lc_0$ and~$\QQ\in\mc M$ we use~$F^\QQ_{\ell(\xi)}$ to denote the cumulative distribution function of~$\ell(\xi)$ under~$\QQ$. In addition, we define the quantile function~$(F^\QQ_{\ell(\xi)})^{-1}$ through $(F^\QQ_{\ell(\xi)})^{-1}(\tau) = \inf\{q \in \R: F^\QQ_{\ell(\xi)}(q) \ge \tau\}$ for all~$\tau\in(0,1)$.

\begin{definition}[Spectral risk measures]
    \label{def:spectral-risk-measure}
    An admissible spectrum is a right-continuous and non-decreasing function~$\psi: [0, 1) \to \R_+$ with~$\int_0^1 \psi(\tau) \mathrm{d} \tau = 1$. The spectral risk measure $\risk_{\QQ}$ induced by~$\psi$ under a given distribution~$\QQ\in\mc M$ of the risk factors is defined through 
    \[
    \risk_\QQ(\ell) = \int_0^1 \psi(\tau) (F^\QQ_{\ell(\xi)})^{-1}(\tau) \mathrm{d} \tau \quad \forall \ell\in\mc L_0.
    \]
\end{definition}

For any fixed~$\QQ$, the set of all spectral risk measures coincides with the family of all coherent, law-invariant and comonotonic risk measures that satisfy a nonrestrictive Fatou property \cite[Theorem~7]{ref:kusuoka2001coherent}. On the other hand, any fixed admissible spectrum~$\psi$ induces a family of spectral risk measures parametrized by the distributions~$\QQ\in\mc M$. This family is law-invariant by construction, and thus
Proposition~\ref{proposition:alpha} and Theorem~\ref{theorem:meanstd} apply whenever~$\mc S$ is stable. The following proposition evaluates~$\alpha$ in closed form for~$\mc S=\mc M_2$.

\begin{proposition}[Standard risk coefficient for spectral risk measures] \label{proposition:spectral}
    If there exists a square-integrable admissible spectrum $\psi$ such that~$\risk_{\QQ}$ is the spectral risk measure induced by~$\psi$ for every~$\mathbb Q \in \mathcal M$ and if~$\mc S=\mc M_2$, then~$\alpha = (\int_0^1 \psi(\tau)^2 \mathrm{d} \tau - 1)^{\half}$.
\end{proposition}

Note that the CVaR at level $\beta \in (0, 1)$ is a spectral risk measure with spectrum $\psi(\tau) = \beta^{-1} \mathbbm{1}_{[1- \beta, 1)}(\tau)$; see, {\em e.g.}, \cite[Definition~4.43 and Lemma~4.46]{ref:follmer2008stochastic}. By Proposition~\ref{proposition:spectral}, the standard risk coefficient of the CVaR is thus given by $(\int_0^1 \psi(\tau)^2 \mathrm{d} \tau-1)^\half = \sqrt{(1-\tau)/\tau}$, which confirms the formula derived in Proposition~\ref{proposition:CVaR}~\emph{(i)}.

Next, we address risk measures that admit a Kusuoka representation~\cite{ref:kusuoka2001coherent, ref:shapiro2013kusuoka} and can be expressed as suprema over families of spectral risk measures. 

\begin{definition}[Kusuoka representation] 
    A risk measure $\risk_{\QQ}$ admits a Kusuoka representation under the distribution~$\QQ\in\mc M$ of the risk factors if there exists a set~$\Psi$ of admissible spectra in the sense of Definition~\ref{def:spectral-risk-measure} such that
    \[\risk_\QQ(\ell) = \Sup{\psi \in \Psi} \; \int_0^1 \psi(\tau) (F^\QQ_{\ell(\xi)})^{-1}(\tau) \mathrm{d} \tau\quad\forall \ell\in\mc L_0.
    \]
\end{definition}

For any fixed~$\QQ$, the set of all risk measures that admit a Kusuoka representation coincides with the family of all coherent law-invariant risk measures satisfying the Fatou property \cite[Theorem~10]{ref:kusuoka2001coherent}. On the other hand, any fixed set~$\Psi$ of admissible spectra induces a law-invariant family of coherent risk measures parametrized by the distributions~$\QQ\in\mc M$, and thus
Proposition~\ref{proposition:alpha} and Theorem~\ref{theorem:meanstd} apply whenever~$\mc S$ is stable. The following proposition presents a closed-form expression for~$\alpha$ if~$\mc S=\mc M_2$. 

\begin{proposition}[Standard risk coefficient for risk measures with a Kusuoka representation] \label{proposition:coherent}
    If there exists a set~$\Psi$ of square-integrable admissible spectra such that~$\risk_{\QQ}$ admits a Kusuoka representation induced by~$\Psi$ for every $\mathbb Q \in \mathcal M$ and if~$\mc S=\mc M_2$, then~$\alpha = \sup_{\psi \in \Psi} (\int_0^1 \psi(\tau)^2 \mathrm{d} \tau - 1)^\half$.
\end{proposition}

Lastly, we study the family of distortion risk measures, which measure the risk of an uncertain loss function by its expected value under a distorted distribution~\cite{ref:yaari1987dual}.

\begin{definition}[Distortion risk measures]
    \label{def:distortion-risk-measure}
    An admissible distortion is a non-decreasing function~$h: [0, 1] \to [0, 1]$ with~$\lim_{\tau \downarrow 0} h(\tau) = h(0) = 0$ and $\lim_{\tau \uparrow 1} h(\tau) = h(1) = 1$. The distortion risk measure $\risk_{\QQ}$ induced by~$h$ under the distribution~$\QQ\in\mc M$ is defined through 
    \[
    \risk_\QQ(\ell) 
    = \int_0^\infty \left( 1 - h \left( F^\QQ_{\ell(\xi)}(\tau) \right) \right) \hspace{0.1ex} \mathrm{d} \tau - \int_{-\infty}^0 h \left( F^\QQ_{\ell(\xi)}(\tau) \right) \hspace{0.1ex} \mathrm{d} \tau \quad \forall \ell\in\mc L_0.
    \]
\end{definition}
\noindent If the distortion $h$ is right-continuous, then we have
\begin{align}
    \label{eq:rep}
    \risk_\QQ(\ell) 
    = \int_0^1 \left(F_{\ell(\xi)}^{\mathbb Q}\right)^{-1}\!\!(\tau) \, \mathrm{d} h(\tau)
    = \int_{\mathbb R} \tau \, \mathrm{d} h\!\left( F_{\ell(\xi)}^{\mathbb Q}(\tau) \right),
\end{align}
where the first equality follows from \cite[Lemma~1]{ref:cai2020distributionally}, whereas the second equality follows from the definition of the Lebesgue-Stieltjes integral~\cite[\S~3]{ref:riesz1990functional}. The resulting reformulation reveals that~$\risk_\QQ(\ell)$ can be viewed as the expected value of the distorted cumulative distribution function~$ h\circ F_{\ell(\xi)}^{\mathbb Q}$. The simplest distortion risk measure is the ordinary expectation, which is induced by the trivial distortion~$h(\tau) = \tau$. Other examples include the VaR and the CVaR at level~$\beta \in (0, 1)$, which are induced by the distortions $h(\tau) = \mathbbm 1_{[\beta, 1]}(\tau)$ and $h(\tau) = \max\{ \tau - 1 + \beta, 0 \} / \beta$, respectively; see, {\em e.g.}, \cite[Definition~4.43 and Lemma~4.46]{ref:follmer2008stochastic}. Moreover, one readily verifies that the spectral risk measure with admissible spectrum $\psi$ coincides with the distortion risk measure induced by the convex continuous distortion $h(\tau) = \int_{0}^\tau \psi(t) \mathrm{d} t$. A comprehensive list of distortion risk measures is provided by~\cite{ref:cai2020distributionally}.

Note that~$(F_{\ell(\xi)}^{\mathbb Q})^{-1}(\tau)=\QQ\text{-}\VaR_{1-\tau} (\ell)$ for all~$\tau\in(0,1)$, and recall that VaR induces a law-invariant family of translation invariant and positive homogeneous risk measures. The second expression in~\eqref{eq:rep} thus implies that the distortion risk measure corresponding to a right-continuous distortion~$h$ represents an average of VaRs with different levels~$\tau$. Hence, any such distortion risk measure induces a law-invariant family of translation invariant and positive homogeneous risk measures, which implies that Proposition~\ref{proposition:alpha} and Theorem~\ref{theorem:meanstd} apply if~$\mc S$ is stable. Next, we provide a closed-form expression for $\alpha$ when~$\mc S=\mc M_2$.

\begin{proposition}[Standard risk coefficient for distortion risk measures] \label{proposition:distortion}
    If there exists an admissible right-continuous distortion~$h$ such that~$\risk_{\QQ}$ is the distortion risk measure induced by~$h$ for every $\mathbb Q \in \mathcal M$ and if~$\mc S=\mc M_2$, then~$\alpha = (\int_0^1 h'_{\mathrm{cvx}}(\tau)^2 \mathrm{d} \tau - 1)^{\half}$, where $h'_{\mathrm{cvx}}$ denotes the derivative of the convex envelope of $h$, which exists almost everywhere.
\end{proposition}

\subsection{Mahalanobis-Gelbrich Risk}\label{sec:extension}
We now briefly discuss an extension of Theorem~\ref{theorem:meanstd} to a broader class of regularization functions. The Mahalanobis distance between two points $\xi_1,\xi_2\in\mathbb R^n$ induced by a matrix $H \in \mathbb S_{++}^n$ is defined as $\| \xi_1 - \xi_2 \|_H = \big((\xi_1 - \xi_2)^\top H (\xi_1 - \xi_2)\big)^{1/2}$.
Accordingly, the Mahalanobis-Wasserstein distance between two distributions $\QQ_1,\QQ_2\in \mc M_2$ is defined as 
\begin{align*}
    \Wass_H(\QQ_1, \QQ_2) = \Min{\pi\in\Pi(\QQ_1,\QQ_2)}  \left( \int_{\R^n \times \R^n} \|\xi_1 - \xi_2\|^2_H\, \pi({\rm d}\xi_1, {\rm d} \xi_2) \right)^{\half},
\end{align*}
whereas the Mahalanobis-Gelbrich distance between two mean-covariance pairs $(\mu_1, \cov_1), (\mu_2, \cov_2) \in \R^n \times \mathbb S_+^n$ is defined as 
\begin{align*}
    \Gelbrich_H \big( (\m_1, \cov_1), (\m_2, \cov_2) \big) = \sqrt{ \| \m_1 - \m_2 \|_H^2 + \Tr{\cov_1 H + \cov_2 H - 2 \big( \cov_2^{\half} H \cov_1 H \cov_2^{\half} \big)^{\half} } }.
\end{align*}
A straightforward generalization of their proofs shows that Theorems~\ref{theorem:gelbrich} and~\ref{theorem:Wasserstein=Gelbrich} apply verbatim with $\Wass(\QQ_1, \QQ_2)$ and $\Gelbrich\big( (\m_1, \cov_1), (\m_2, \cov_2) \big)$ replaced by $\Wass_H(\QQ_1, \QQ_2)$ and $\Gelbrich_H\big( (\m_1, \cov_1), (\m_2, \cov_2) \big)$, respectively. As a consequence, we can now define the Mahalanobis-Gelbrich ambiguity set with structural information as 
\[
    \mc G_{H,\rho}(\Pnom) = \{ \QQ \in \mathcal S: \Gelbrich_H(\Pnom, \QQ) \leq \rho \}.
\]
A straightforward generalization of its proof then shows that Theorem~\ref{theorem:meanstd} applies verbatim with Equation~\eqref{eq:MS:1} replaced by
\begin{align*}
        \Sup{\QQ \in \mc{G}_{H,\rho}(\msa, \covsa)} \; \risk_\QQ \left( - \w^\top \xi \right) = -\msa^\top \w + \alpha \sqrt{\w^\top \covsa \w } + \rho \sqrt{1+ \alpha^2}\, \| \w \|_{H^{-1}},
    \end{align*}
which is the announced extension.

\section{Gelbrich Risk of Nonlinear Portfolio Loss Functions}\label{sec:nonlinearportfolio}
Even though the Gelbrich risk of a {\em nonlinear} portfolio loss function is generically {\em not} available in closed form, it can sometimes be computed efficiently by solving a tractable convex program. Throughout this section we assume that the structural ambiguity set $\mc S = \mc M_2$ comprises all distributions with finite second moments. All of our tractability results rely on the following theorem.

\begin{theorem}[Gelbrich risk of nonlinear loss functions] \label{theorem:WC-exp}
    If $\ell\in\mc L_0$, $\mc S = \mc M_2$ and $\rho>0$, then we have
    \begin{align*}
        \Sup{\QQ \in \mathcal{G}_{\rho}(\widehat \mu, \covsa)} \EE_{\QQ} \left[ \ell(\xi) \right] =
        \left\{
        \begin{array}{cl}
            \inf &  y_0 + \gamma \big( \rho^2 - \| \widehat \mu \|^2 - \Tr{\covsa} \big) + z + \Tr{Z} \\
            \st &\gamma \in \mathbb R_+, \; z \in \mathbb R_+, \; Z \in \PSD^n, \; (y_0, y, Y) \in \mathcal Y \\[1ex]
            &\begin{bmatrix} \gamma I - Y & \gamma \covsa^\half \\ \gamma \covsa^\half & Z \end{bmatrix} \succeq 0, \; \begin{bmatrix} \gamma I - Y & \gamma \widehat \mu + y \\ (\gamma \widehat \mu + y)^\top & z \end{bmatrix} \succeq 0,
        \end{array}
        \right.
    \end{align*}
    where
    $$\mathcal Y = \Big\{ y_0 \in \mathbb R, \; y \in \mathbb R^n, \; Y \in \Sym^n: y_0 + 2y^\top \xi + \xi^\top Y \xi \geq \ell(\xi) ~ \forall \xi \in \mathbb R^n \Big\}.$$
\end{theorem}
Theorem~\ref{theorem:WC-exp} expresses the Gelbrich risk of {\em any} measurable loss function as the optimal value of a convex minimization problem. Note, however, that the convex set~$\mathcal Y$ is defined through infinitely many linear inequality constraints parametrized by~$\xi\in\R^n$ and may thus be difficult to handle algorithmically. We now investigate three distinct applications of Theorem~\ref{theorem:WC-exp} in which the Gelbrich risk of a nonlinear loss function can be evaluated tractably.

\subsection{Gelbrich VaR of Piecewise Linear Loss Functions}
\label{sect:polyhedralVaR}

Suppose that 
\[
\ell(\xi) = -\w^\top \max\{ A\xi + a, B\xi + b \}
\]
for some $A, B \in \mathbb R^{k \times n}$, $a, b \in \mathbb R^k$ and portfolio vector $w\in\mathbb R^k$, where the `max' operator applies element-wise. We stress that {\em any} piecewise linear concave loss function with $k$ kinks is representable in this way. Loss functions of this type arise if the portfolio includes (long positions) of European put or call options that expire at the end of the investment horizon \cite{ref:zymler2013wcvar}. Indeed, such options provide portfolio insurance, and their payoffs (negative losses) are piecewice linear in the asset returns~$\xi$.

% We now consider the problem of evaluating the Gelbrich VaR of a portfolio consisting of both stocks and their derivatives. A na\"{i}ve approach is to treat the derivatives as usual stocks and proceed with the evaluation of the Gelbrich VaR of a linear portfolio loss as introduced in Section~\ref{sec:linearportfolio}. However, because the Gelbrich risk depends only on the first- and second-moment of the joint stock-derivative return distribution, this approach fails to capture the inherent \textit{non}linear dependency of the derivative returns on the stock returns. Thus, we consider an explicit portfolio model that consists of an allocation $\w \in \W \subseteq \mathbb R^n$ over $n$ stocks and a derivative allocation $x \in \mathcal X \subseteq \mathbb R^k_+$ over $k$ derivatives. For simplicity, we consider only vanilla derivatives of the stocks, and in this case, the return of the derivatives can be written as a piecewise linear function of the stock return~$\xi$. Given a mixed portfolio $(\w, \x) \in \W \times \mathcal X$, the portfolio loss can be written as a polyhedral function of $\xi$ for some matrices $A, B \in \mathbb R^{k \times n}$ and vectors $a, b \in \mathbb R^k$ as
% \[
% \ell(\xi) = -\w^\top \max\{ A\xi + a, B\xi + b \},
% \]
% where the max operator is understood as the element-wise maximum of two vectors. We emphasize that the short selling of derivatives is prohibited in this model. 

The next theorem demonstrates that if the risk measure is set to the VaR, then the Gelbrich risk of our piecewise linear loss function can be computed efficiently by solving a semidefinite program.

\begin{theorem}[Polyhedral Gelbrich VaR] \label{theorem:poly-var}
    Suppose that $\ell(\xi) = -\w^\top\max\{ A\xi + a, B\xi + b \}$ for some $A, B \in \mathbb R^{k \times n}$, $a, b \in \mathbb R^k$ and $w\in\mathbb R^k$. If~$\beta\in(0,1)$, $\risk_\QQ=\QQ\text{\em -VaR}_\beta$ and $\mc S = \mc M_2$, then we have
    \begin{align*}
        \Sup{\QQ \in \mathcal{G}_{\rho}(\widehat \mu, \covsa)} \QQ\text{-}\VaR_{\beta}(\ell(\xi))
        = 
        \left\{
        \begin{array}{cl}
            \inf & \tau \\[-1ex]
            \st & \tau, v_0, y_0 \in \mathbb R, ~ \gamma, \eta, z \in \mathbb R_+, ~ v, y \in \mathbb R^n, ~ Y \in \PSD^n, ~ \zeta \in \mathbb R_+^k, ~ Z \in \PSD^n \\
            & \zeta \leq \w, ~ y_0 + \gamma \big( \rho^2 - \| \widehat \mu \|^2 - \Tr{\covsa} \big) + z + \Tr{Z} \leq \eta \beta \\
            & v =  \frac{1}{2} \left( (A - B)^\top \zeta + B^\top \w \right), ~ v_0 = \tau + (a-b)^\top \zeta + b^\top \w \\
            &\begin{bmatrix} \gamma I - Y & \gamma \covsa^\half \\ \gamma \covsa^\half & Z \end{bmatrix} \succeq 0, ~ \begin{bmatrix} \gamma I - Y & \gamma \widehat \mu + y \\ (\gamma \widehat \mu + y)^\top & z \end{bmatrix} \succeq 0 \\[3ex]
            & \begin{bmatrix} Y & y \\ y^\top & y_0 \end{bmatrix} \succeq 0, ~ \begin{bmatrix} Y & y + v \\ y^\top + v^\top & y_0 + v_0 - \eta \end{bmatrix} \succeq 0.
        \end{array}
        \right.
    \end{align*}
\end{theorem}

The constraints of the semidefinite program derived in Theorem~\ref{theorem:poly-var} are linear in~$\w$. Thus, the portfolio optimization problem that minimizes the Gelbrich VaR of our piecewise linear loss function over all $\w\in\Omega$ is equivalent to a tractable convex program whenever $\Omega$ is conic-representable. In analogy to~\cite[\S~6]{ref:zymler2013wcvar}, we can further show that the Gelbrich VaR coincides with the Gelbrich CVaR.

\begin{theorem}[Polyhedral Gelbrich CVaR] \label{theorem:poly-cvar}
    If $\risk_\QQ=\QQ\text{\em -CVaR}_\beta$ but all other assumptions of Theorem~\ref{theorem:poly-var} hold, then we have
    \begin{align*}
        \Sup{\QQ \in \mathcal{G}_{\rho}(\widehat \mu, \covsa)} \QQ\text{-}\CVaR_{\beta}(\ell(\xi)) ~ \, = \Sup{\QQ \in \mathcal{G}_{\rho}(\widehat \mu, \covsa)} \QQ\text{-}\VaR_{\beta}(\ell(\xi)).
    \end{align*}
\end{theorem}

\subsection{Gelbrich VaR of Quadratic Loss Functions}
\label{sect:quadVaR}

Suppose next that
\[
\ell(\xi) = - \theta(\w) - \Delta(\w)^\top \xi - \half \xi^\top \Gamma(\w) \xi
\]
for some $\theta(\w) \in \mathbb R$, $\Delta(\w) \in \mathbb R^n$ and $\Gamma(\w) \in \Sym^n$ that depend linearly on the portfolio vector~$\w\in\mathbb R^n$. We emphasize that $\Gamma(\w)$ may be indefinite, and thus {\em any} quadratic loss function is representable in this way. Quadratic loss functions arise if the portfolio includes arbitrary (exotic) derivatives whose payoffs are approximated by a second-order Taylor expansion (a {\em delta-gamma approximation}) \cite{ref:jaschke2002deltagamma}.

% The results of Section~\ref{sect:polyhedralVaR} rely fundamentally on the restriction that the stochastic returns of the derivatives can be modelled using a piecewise linear function of the stock returns $\xi$. In practice, an exotic derivative may exhibit strongly \textit{non}linear dependence on the stock returns $\xi$, hence the polyhedral model in Section~\ref{sect:polyhedralVaR} is inadequate for assessing the risk of a complex portfolio position. We consider in this section an approximation of the portfolio return using a second-order Taylor expansion which is commonly known as the delta-gamma approximation~\cite{ref:jaschke2002deltagamma}. Given a portfolio allocation $\w \in \mathbb R^n$ over both stocks and derivatives, the loss of a portfolio is approximated using a quadratic function of $\xi$ as
% \[
% \ell(\xi) = - \theta(\w) - \Delta(\w)^\top \xi - \half \xi^\top \Gamma(\w) \xi,
% \]
% where $\theta(\w) \in \mathbb R$, $\Delta(\w) \in \mathbb R^n$ and $\Gamma(\w) \in \Sym^n$ are parameters derived from the current portfolio position $\w$. We emphasize that $\ell$ is possibly a non-convex function of $\xi$ and that we do not prohibit the short selling of derivative in this section.

\begin{theorem}[Quadratic Gelbrich VaR] \label{theorem:quad-var}
    Suppose that $\ell(\xi) = - \theta(\w) - \Delta(\w)^\top \xi - \half \xi^\top \Gamma(\w) \xi$ for some $\theta(\w) \in \mathbb R$, $\Delta(\w) \in \mathbb R^n$ and $\Gamma(\w) \in \Sym^n$. If~$\beta\in(0,1)$, $\risk_\QQ=\QQ\text{\em -VaR}_\beta$ and $\mc S = \mc M_2$, then we have
    \begin{align*}
    \Sup{\QQ \in \mathcal{G}_{\rho}(\widehat \mu, \covsa)} \QQ\text{-}\VaR_{\beta}(\ell(\xi))
    =
    \left\{
    \begin{array}{cl}
        \inf & \tau \\[-1ex]
        \st & y_0, \tau \in \mathbb R, \; y \in \mathbb R^n, \; Y \in \PSD^n, \; \; \gamma,z,\eta \in \mathbb R_+, \; Z \in \PSD^n \\
        & y_0 + \gamma \big( \rho^2 - \| \widehat \mu \|^2 - \Tr{\covsa} \big) + z + \Tr{Z} \leq \eta \beta \\
        &\begin{bmatrix} \gamma I - Y & \gamma \covsa^\half \\ \gamma \covsa^\half & Z \end{bmatrix} \succeq 0, \; \begin{bmatrix} \gamma I - Y & \gamma \widehat \mu + y \\ (\gamma \widehat \mu + y)^\top & z \end{bmatrix} \succeq 0 \\[3ex]
        & \begin{bmatrix} Y & y \\ y^\top & y_0 \end{bmatrix} \succeq 0, \;\begin{bmatrix} Y & y \\ y^\top & y_0 \end{bmatrix} + \begin{bmatrix} \Gamma(\w) & \Delta(\w) \\ \Delta(\w)^\top & -\eta + 2(\tau + \theta(\w)) \end{bmatrix} \succeq 0.
    \end{array}
    \right. 
    \end{align*}
\end{theorem}

As $\theta(\w)$, $\Delta(\w)$ and $\Gamma(\w)$ are linear in~$\w$, the problem that minimizes the Gelbrich VaR of our quadratic loss function over all $\w\in\Omega$ is a tractable convex program if $\Omega$ is conic-representable.

\begin{theorem}[Quadratic Gelbrich CVaR] \label{theorem:quad-cvar}
    If $\risk_\QQ=\QQ\text{\em -CVaR}_\beta$ but all other assumptions of Theorem~\ref{theorem:quad-var} hold, then we have
    \[
    \Sup{\QQ \in \mathcal{G}_{\rho}(\widehat \mu, \covsa)} \QQ\text{-}\CVaR_{\beta}(\ell(\xi)) = \Sup{\QQ \in \mathcal{G}_{\rho}(\widehat \mu, \covsa)} \QQ\text{-}\VaR_{\beta}(\ell(\xi)).
    \]
\end{theorem}

\subsection{Gelbrich Error} \label{sec:indextracking}

Suppose finally that $\ell(\xi) = | w^\top \xi|^p$ for $p \in \{1, 2\}$. Loss functions of this type arise in index tracking, where one component of $\xi$ represents the return of a market index, and the corresponding component of $\w$ is fixed to $-1$. In this case, $\w^\top\xi$ quantifies the mismatch between the returns of the index and the corresponding tracking portfolio, while $| w^\top \xi |^p$ captures the tracking error. The worst-case expected tracking error with respect to all distributions in a Gelbrich ambiguity set can again be computed efficiently by solving a semidefinite program.

\begin{theorem}[Gelbrich error] \label{theorem:tracking} Suppose that $\ell(\xi) = | w^\top \xi |^p$ for $p \in \{1, 2\}$ and that~$\mc S = \mc M_2$.
\begin{enumerate}[label=(\roman*), leftmargin = 7mm]
        \item If $p=1$, then we have
     \[
    \Sup{\QQ \in \mathcal{G}_{\rho}(\widehat \mu, \covsa)} \EE_{\QQ} \left[ | w^\top \xi | \right] =
            \left\{
            \begin{array}{cl}
                \inf &  y_0 + \gamma \big( \rho^2 - \| \widehat \mu \|^2 - \Tr{\covsa} \big) + z + \Tr{Z} \\
                \st &\gamma \in \mathbb R_+, \; z \in \mathbb R_+, \; Z \in \PSD^n, \;  y_0 \in \R, \; y \in \R^n, \; Y \in \PSD^n\\
                &\begin{bmatrix} \gamma I - Y & \gamma \covsa^\half \\ \gamma \covsa^\half & Z \end{bmatrix} \succeq 0, \; \begin{bmatrix} \gamma I - Y & \gamma \widehat \mu + y \\ (\gamma \widehat \mu + y)^\top & z \end{bmatrix} \succeq 0 \\ [3ex]
                &\begin{bmatrix} 
            Y & y - \frac{1}{2} \w \\ y^\top - \frac{1}{2} \w^\top & y_0
        \end{bmatrix} \succeq 0, \quad
        \begin{bmatrix} 
            Y & y + \frac{1}{2} \w \\ y^\top + \frac{1}{2} \w^\top & y_0
        \end{bmatrix} \succeq 0.
            \end{array}
            \right.
    \]
    \item If $p=2$, then we have
  \[
    \Sup{\QQ \in \mathcal{G}_{\rho}(\widehat \mu, \covsa)} \EE_{\QQ} \left[ | w^\top \xi |^2 \right] =
            \left\{
            \begin{array}{cl}
                \inf &  y_0 + \gamma \big( \rho^2 - \| \widehat \mu \|^2 - \Tr{\covsa} \big) + z + \Tr{Z} \\
                \st &\gamma \in \mathbb R_+, \; z \in \mathbb R_+, \; Z \in \PSD^n, \;  y_0 \in \R, \; y \in \R^n, \; Y \in \PSD^n,\; M \in \PSD^n\\
                &\begin{bmatrix} \gamma I - Y & \gamma \covsa^\half \\ \gamma \covsa^\half & Z \end{bmatrix} \succeq 0, \; \begin{bmatrix} \gamma I - Y & \gamma \widehat \mu + y \\ (\gamma \widehat \mu + y)^\top & z \end{bmatrix} \succeq 0 \\[3ex]
                &
                \begin{bmatrix} 
                    M & y \\ y^\top &y_0 
                \end{bmatrix} \succeq 0, \quad
                \begin{bmatrix} 
                    Y - M & \w \\ \w^\top & 1
                \end{bmatrix} \succeq 0.
          \end{array}
          \right.
    \]
\end{enumerate}
\end{theorem}

Corollary~\ref{corollary:WC-pq} in the appendix uses the $\mc S$-lemma~\cite{ref:polik2007Slemma} to show that the worst-case expectation of any piecewise quadratic loss function admits a tractable semidefinite programming reformulation. Theorem~\ref{theorem:tracking} is a special case of this result, and its proof is thus omitted for brevity.

\section{Numerical Results}\label{sec:numerical}

{
Accurate estimation of the mean vector and the covariance matrix of the asset returns is crucial for constructing an effective Gelbrich ambiguity set. Empirical estimators, such as the sample mean and the sample covariance matrix, are known to be highly sensitive to outliers and are thus unsuitable for high-dimensional portfolio optimization problems. This sensitivity can lead to significant estimation errors, which can severely degrade the out-of-sample performance of optimized portfolios \citep{ref:demiguel2009generalized, ref:gotoh2011role, demiguel2013size}.

A standard approach to improve out-of-sample performance is to replace the empirical estimators with robust estimators that are less sensitive to outliers. We now assess the benefits of complementing robust estimation with robust optimization techniques by constructing Gelbrich ambiguity sets around different robust covariance estimators. Specifically, we consider linear shrinkage estimators \citep{ref:ledoit2003improved, ledoit2004well}, M-estimators \citep{maronna1976robust, huber1981robust}, S-estimators \citep{lopuhaa1989relation, campbell1998calculation}, and robust estimators based on the minimum covariance determinant principle~\citep{rousseeuw1984least, rousseeuw1999fast}. The Gelbrich ambiguity set introduces a second layer of robustness that offers protection against estimation errors. We compare the resulting Gelbrich portfolios against the solutions of distributionally robust portfolio optimization problems with Wasserstein and $\chi^2$-divergence ambiguity sets centered at empirical distributions. All optimization models are implemented in Python with interfaces compatible with the \texttt{scikit-learn} library. The code is available at \url{https://github.com/sorooshafiee/mean-covariance-robust-risk}. 
We use Gurobi 12.0.2 to solve all linear and second-order cone programs, and MOSEK 11.0.21 to solve all semidefinite programs through the JuMP interface \citep{dunning2017jump}.

We first review different robust covariance estimators constructed from a dataset $\xi_1, \dots, \xi_N \in \R^n$. Linear shrinkage estimators improve upon the sample covariance matrix by \emph{shrinking} it towards a more structured and stable target matrix. In particular, the linear shrinkage estimator is defined as a convex combination of the sample covariance matrix and a shrinkage target, where the weight of the target is referred to as the shrinkage intensity. \citet{ref:ledoit2003improved,ledoit2004well} derive a data-dependent, asymptotically optimal formula for the shrinkage intensity, thus eliminating the need for computationally expensive cross-validation. The corresponding \texttt{Ledoit-Wolf} shrinkage estimator provides more accurate covariance estimates, leading to more stable portfolio weights and improved out-of-sample performance \citep{ref:gotoh2011role, ref:demiguel2009generalized}. In our experiments we employ the \texttt{scikit-learn} implementation of the \texttt{Ledoit-Wolf} estimator, which uses the sample mean as the location estimator.

M-estimators, introduced by \citet{huber1964robust}, generalize the classical maximum likelihood framework to improve robustness against deviations from assumed distributional models, particularly in the presence of outliers. \citet{maronna1976robust} was the first to define M-estimators for the joint estimation of the mean vector $\widehat{\mu}_N$ and the covariance matrix $\widehat{\Sigma}_N$. \citet{huber1981robust} streamlined this framework by characterizing the M-estimators~$\widehat{\mu}_N$ and~$\widehat{\Sigma}_N$ as the solutions to the score equations
\begin{align}
    \label{eq:fixed:point}
    \begin{cases}
        \sum_{i=1}^N v_1(d_i) (\xi_i - \widehat \mu_N) = 0 \\
        %\widehat{\mu}_N = \frac{\sum_{i=1}^{N} g(d_i) \xi_i}{\sum_{i=1}^{N} g(d_i)} \\[1em]
        \sum_{i=1}^N \left( v_2(d_i^2) (\xi_i - \widehat \mu_N) (\xi_i - \widehat \mu_N)^\top - v_3(d_i) \widehat \Sigma_N \right) = 0,
        %\displaystyle \widehat{\Sigma}_N = \frac{1}{N} \sum_{i=1}^{n} h(d_i^2) (\xi_i - \widehat{\mu}_N)(\xi_i - \widehat{\mu}_N)^\top,
    \end{cases}
\end{align}
where $d_i^2 = (\xi_i - \widehat{\mu}_N)^\top \widehat{\Sigma}^{-1}_N (\xi_i - \widehat{\mu}_N)$ denotes the squared Mahalanobis distance between~$\xi_i$ and~$\widehat{\mu}_N$. The weight functions~$v_1(\cdot)$ and~$v_2(\cdot)$ are designed to down-weight observations~$\xi_i$ with large Mahalanobis distances, thereby mitigating the impact of outliers. Following Huber \citep{huber1981robust}, we set $v_1(d) =  1$ if $d \leq \eps$ and $v_1(d)= \eps / d$ otherwise, $v_2(d^2)=v_1(d)$, and $v_3(d) = 1$,
% \begin{align*}
%     v_2(d^2) = v_1(d) = 
%     \begin{cases}
%         1 & \text{if~} d \leq \eps \\
%         \eps / d & \text{otherwise,}
%     \end{cases}
% \end{align*}
where~$\eps > 0$ is a tuning parameter. We construct the resulting $\texttt{Huber-M}$ estimators using the iterative scheme described in \citep[\S~6.7.1]{maronna2006robust}.

S-estimators, introduced by \citet{rousseeuw1984robust}, have a high breakdown point and remain reliable under substantial data contamination. Portfolios constructed on the basis of S-estimators are thus less sensitive to extreme market events and display improved stability and performance predictability. \citet{lopuhaa1989relation} characterizes the S-estimators~$\widehat{\mu}_N$ and~$\widehat{\Sigma}_N$ as the solutions to the score equations~\eqref{eq:fixed:point}, where the weight functions satisfy the relations $v_2(d^2) = n v_1(d)$ and $v_3(d) = d^2 v_1(d) - \int_0^d y v_1(y) \diff y + b_0$, and the constant~$b_0$ is set to the closed-form expression proposed in \citep{campbell1998calculation}. Following Lopuha\"{a} \citep{lopuhaa1989relation}, we set $v_1(d)$ to Tukey's biweight function defined via $v_1(d) = ( 1 - (d / \eps)^2)^2 $ if $d \leq \eps$ and $v_1(d)= 0$ otherwise, where~$\eps > 0$ is a tuning parameter. We construct the resulting $\texttt{Tukey-S}$ estimators using the iterative scheme described in \citep{campbell1998calculation}.

Robust M- and S-estimators re-weight observations based on their Mahalanobis distances to~$\widehat\mu_N$. \citet{rousseeuw1984least} advocates an alternative approach to robust estimation, which uses only a subset of~$N'<N$ observations that are considered to be \emph{clean} and computes their mean vector and covariance matrix. The minimum covariance determinant (\texttt{MCD}) estimator is based on the subset of~$N'$ observations whose sample covariance matrix has the smallest determinant. Computing this estimator {\em exactly} is intractable. Therefore, we rely on the \texttt{scikit-learn} implementation of the fast algorithm by \citet{rousseeuw1999fast} to compute the \texttt{MCD} estimator corresponding to $N' = (n + N + 1)/2$ observations {\em approximately}. This estimator is highly resistant to outliers, making it a powerful tool for portfolio construction in financial markets that are prone to extreme events. Our publicly available code also includes backtests involving other robust estimators implemented in the \texttt{RobPy} package \citep{leyder2024robpy}; however, their results are not reported in the following sections.

\subsection{Linear Portfolio Selection}
\label{sec:linear-portfolio-experiment}
Theorem~\ref{theorem:meanstd} shows that the Gelbrich risk of any {\em linear} portfolio loss function is exactly equal to a Markowitz-type mean-standard deviation risk functional with an $\ell_2$-regularization term. The benefits of penalizing or---equivalently---constraining the norm of the portfolio vector in a Markowitz model are already well documented in the literature. Specifically, it has been shown empirically across several datasets that norm-constrained minimum-variance portfolios often display higher Sharpe ratios than standard baseline strategies \cite{ref:demiguel2009generalized,ref:gotoh2011role}. In this section, we assess the out-of-sample performance of our Gelbrich portfolios relative to that of other---more established---distributionally robust optimization benchmarks in the context of a synthetic numerical experiment. This experiment examines how the  sample size~$N$ influences the relative performance of different models, with the CVaR at level~$\beta$ used as the primary performance metric. We compute the standard risk coefficient~$\alpha$ using the formula of Proposition~\ref{proposition:CVaR}~\emph{(i)} and set~$\beta = 0.8$.

The experiment is based on the synthetic market considered in \citep{ref:bertsimas2018robust, ref:esfahani2018data}. We generate asset returns for $n = 20$ assets from a multivariate normal distribution, assuming each return $\xi_j$ decomposes into a systematic risk factor $\psi \sim \mathcal N(0,\sqrt{3\%})$, common to all assets, and an idiosyncratic risk factor $\zeta_j \sim \mathcal{N}(j \times 1\%, \sqrt{j \times 2\%})$, specific to asset~$j$. Thus, $\xi_j = \psi + \zeta_j$ for all $j \in [n]$. By construction, assets with higher indices promise higher mean returns at a higher risk. For each run of the experiment, $10^6$ out-of-sample returns are generated to allow for accurate performance evaluation. We consider three distinct scenarios: a small sample regime with $N = 20$, an intermediate sample regime with $N = 100$, and a large sample regime with $N = 1{,}000$. The entire experiment is repeated over $100$ simulation runs in order to ensure statistical robustness of our findings.

We compare various portfolio strategies on test data. Gelbrich portfolios based on different estimators for the mean vector and covariance matrix of asset returns are obtained by minimizing a regularized mean-standard deviation functional (see Theorem~\ref{theorem:meanstd}). The robustness of these portfolios arises from two sources: (i) the use of robust estimators for the mean and covariance matrix, and (ii) the regularization term weighted by the radius~$\rho$ of the Gelbrich ambiguity set. The Gelbrich approach is benchmarked against two prominent distributionally robust optimization models with a Wasserstein and a $\chi^2$-divergence ambiguity set of radius~$\rho$, respectively. Both ambiguity sets are centered at the empirical distribution of the $N$ training samples. We use \citep[Theorem~6.3]{ref:esfahani2018data} and \citep[Lemma~1]{duchi2021learning} to derive tractable reformulations for the resulting optimization models, respectively. %We evaluate the CVaR of each portfolio on a test dataset. %The primary metrics are the out-of-sample CVaR loss and the realized mean-standard deviation return, with $\alpha$ computed according to Proposition~\ref{proposition:CVaR}.

Figure~\ref{fig:syntethic} visualizes the CVaR (top row) and the mean-standard deviation return $\mathbb E_{\mathbb P}[w^\top \xi] - \alpha \sqrt{\Var_{\mathbb P}(w^\top \xi)}$ (bottom row) of the different portfolios on $10^6$~test samples as a function of~$\rho$, averaged over 100 simulation runs. Here, $\alpha$ denotes the usual standard risk coefficient. The Gelbrich portfolios benefit from a second layer of robustness, which mitigates the impact of estimation errors in the mean vectors and covariance matrices. This explains why the Gelbrich portfolios perform best for strictly positive values of~$\rho$. We observe that the benefits of robustification depend strongly on the training sample size. In the small-sample regime ($N=20$), estimation errors dominate. Here, the Gelbrich portfolio based on the \texttt{Ledoit-Wolf} estimator achieves the lowest CVaR and the highest mean-standard deviation return, outperforming not only all other Gelbrich portfolios but also the Wasserstein and $\chi^2$-divergence distributionally robust portfolios. This suggests that when training data is scarce, accurately estimating the entire asset return distribution---as implicitly done by the Wasserstein and $\chi^2$-divergence models---is more error-prone. A more conservative strategy that robustly estimates only the first and second moments of the distribution proves more effective. The Gelbrich portfolios perform particularly well when they are based on the \texttt{Ledoit-Wolf} estimator, which applies aggressive shrinkage and provides particularly stable covariance estimates.
%Figure~\ref{fig:syntethic} presents the out-of-sample performance for various methods. These results reveal that the proposed framework provides a valuable second layer of robustness for different mean-covariance estimators, with its effectiveness being highly dependent on the amount of available data. In the small sample regime, where estimation error is the predominant challenge, the Markowitz optimizer equipped with the \texttt{Ledoit-Wolf} estimator demonstrates superior performance when robustified with respect to the Gelrich ambiguity set. It consistently yields better out-of-sample CVaR and mean-standard deviation returns compared not only to other Markowitz estimators but also to the Wasserstein and $\chi^2$-divergence distributionally robust optimization formulations. This observation suggests that when data is scarce, the task of accurately estimating the entire underlying probability distribution, implicitly required by distributionally robust optimization models, is fraught with error. A more constrained approach that focuses on robustly estimating only the first and second moments proves more effective. The \texttt{Ledoit-Wolf} estimator, with its shrinkage properties, provides a stable and accurate estimate of the covariance matrix, which, when combined with our regularized optimizer, leads to a more genuinely robust portfolio. 
Similar patterns emerge in the moderate data regime ($N=100$), although the advantage of the Gelbrich portfolios with \texttt{Ledoit-Wolf} estimator vis-\`a-vis the Wasserstein portfolios erodes. Nonetheless, the Gelbrich approach remains more tractable because computing a Wasserstein portfolio requires solving a constrained nonlinear optimization problem with $\mathcal O(N)$ variables and constraints. In the large data regime ($N=1{,}000$), the dominance of the Gelbrich portfolios is broken. Indeed, the Wasserstein portfolios corresponding to small radii~$\rho$ (which are approximately optimized in view of the mean), begin to outperform the Gelbrich portfolios. Note that CVaR is sensitive to the tails of the asset return distribution and therefore depends on the entire distribution, not only on its first two moments. With a large dataset, the empirical distribution provides a high-fidelity approximation of the unknown true distribution. In this regime, Wasserstein portfolios can exploit richer (reliable) distributional information by optimizing against the worst-case distribution within a small neighborhood of the empirical distribution. As a result, they achieve better out-of-sample performance than Gelbrich portfolios, which rely only on the first and second moments. Our experiments highlight a critical trade-off. When data is scarce, robust moment estimation and an additional layer of distributional robustness are essential; whereas with abundant data, directly optimizing in view of the empirical distribution is a more powerful approach.

\begin{figure}[!tb]
    \centering
    \begin{subfigure}[b]{0.31\textwidth}
        \centering
        \includegraphics[width=\linewidth]{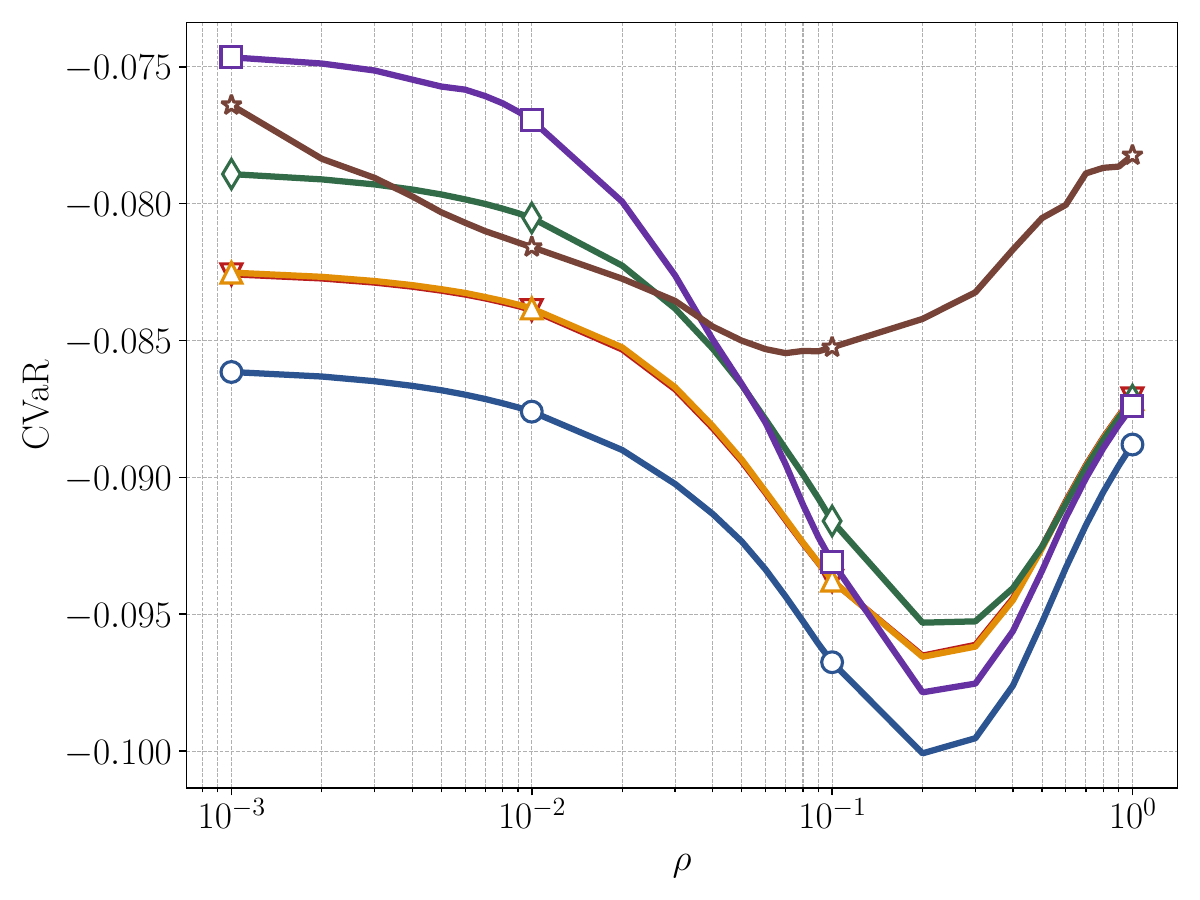}
        % \caption{Caption 1}
    \end{subfigure}
    \hfill
    \begin{subfigure}[b]{0.31\textwidth}
        \centering
        \includegraphics[width=\linewidth]{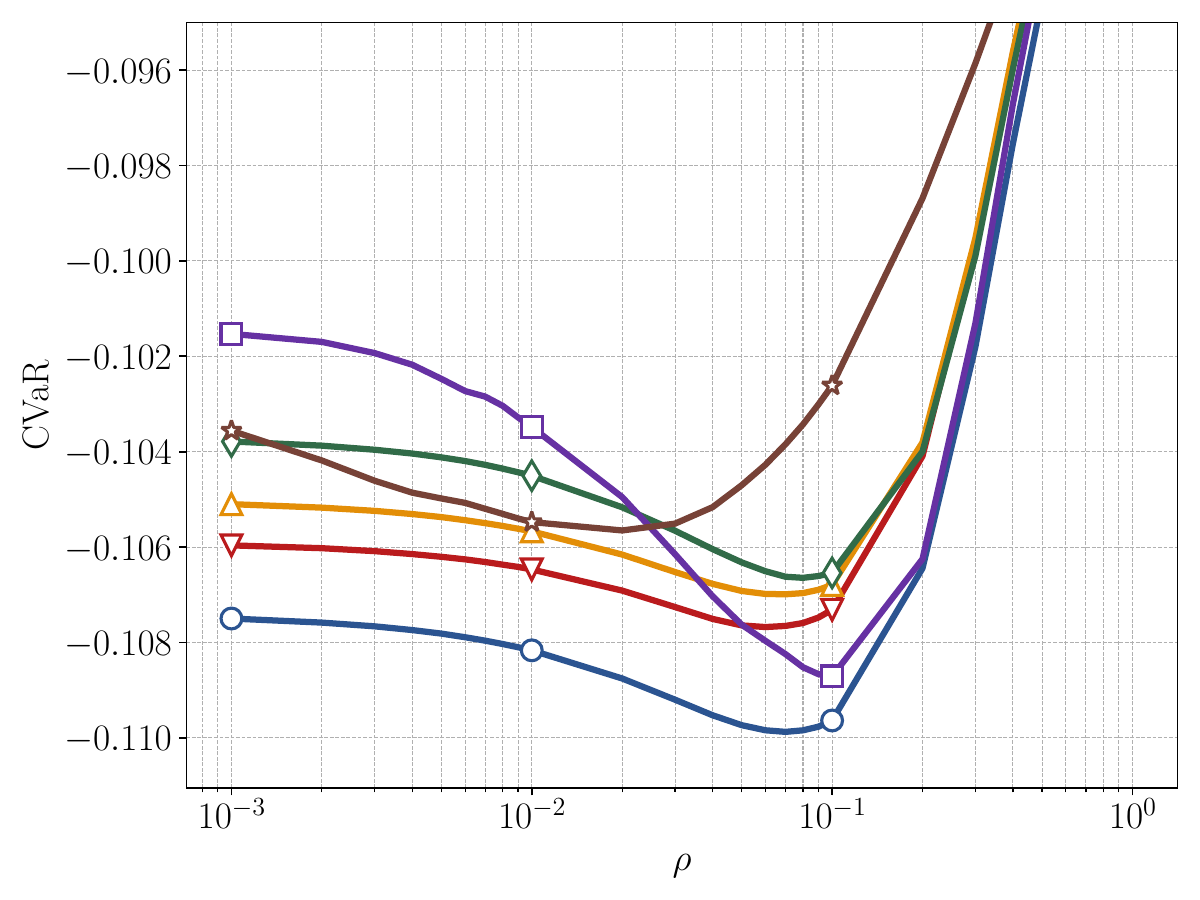}
        % \caption{Caption 2}
    \end{subfigure}
    \hfill
    \begin{subfigure}[b]{0.31\textwidth}
        \centering
        \includegraphics[width=\linewidth]{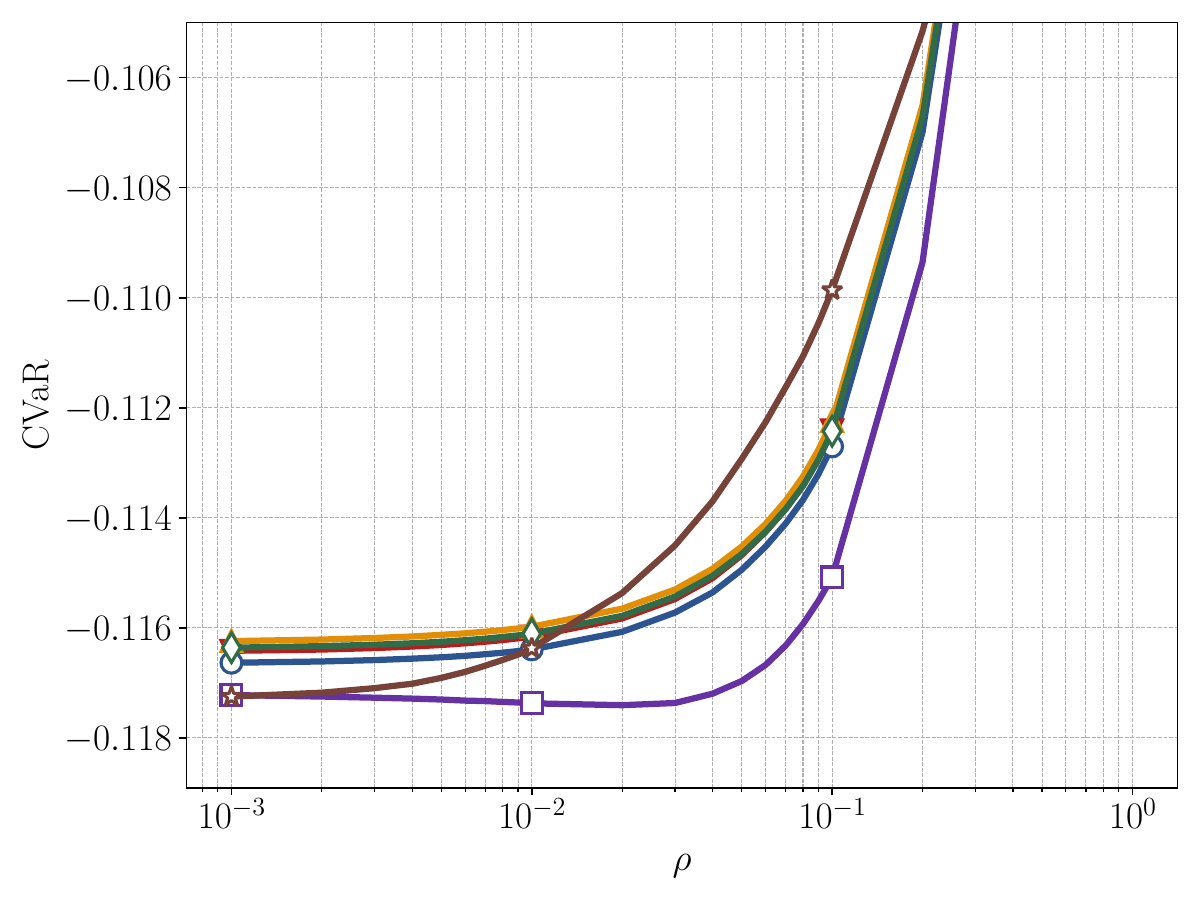}
        % \caption{Caption 3}
    \end{subfigure} \\
    \begin{subfigure}[b]{0.31\textwidth}
        \centering
        \includegraphics[width=\linewidth]{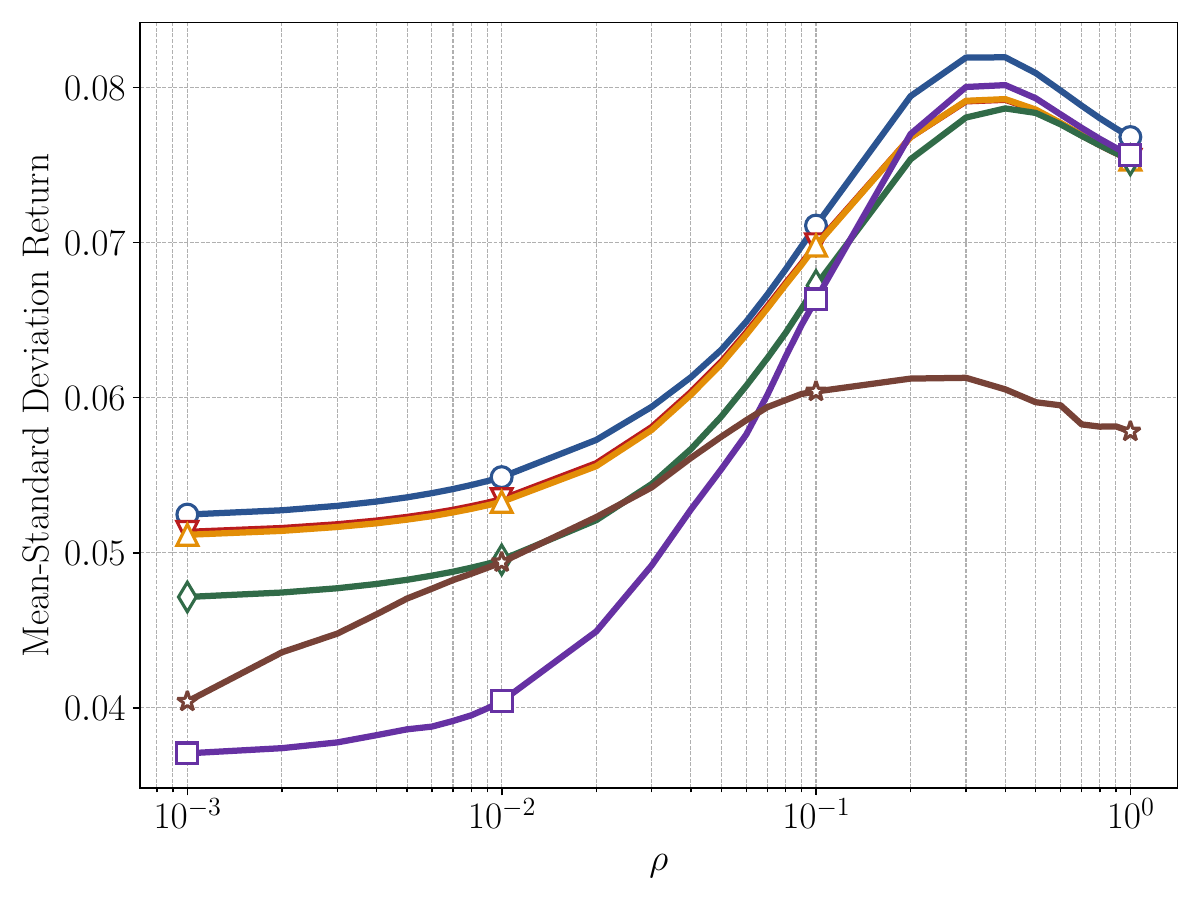}
        % \caption{Caption 1}
    \end{subfigure}
    \hfill
    \begin{subfigure}[b]{0.31\textwidth}
        \centering
        \includegraphics[width=\linewidth]{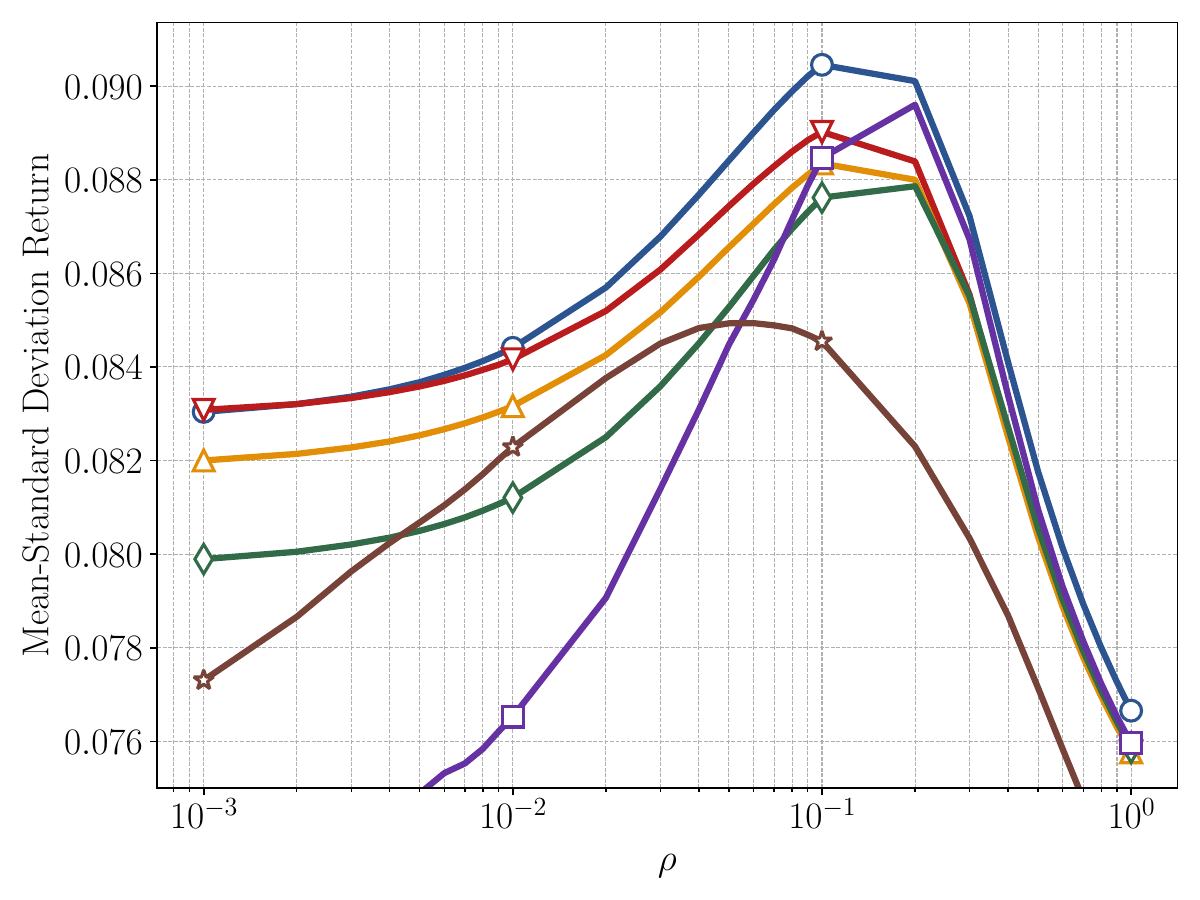}
        % \caption{Caption 2}
    \end{subfigure}
    \hfill
    \begin{subfigure}[b]{0.31\textwidth}
        \centering
        \includegraphics[width=\linewidth]{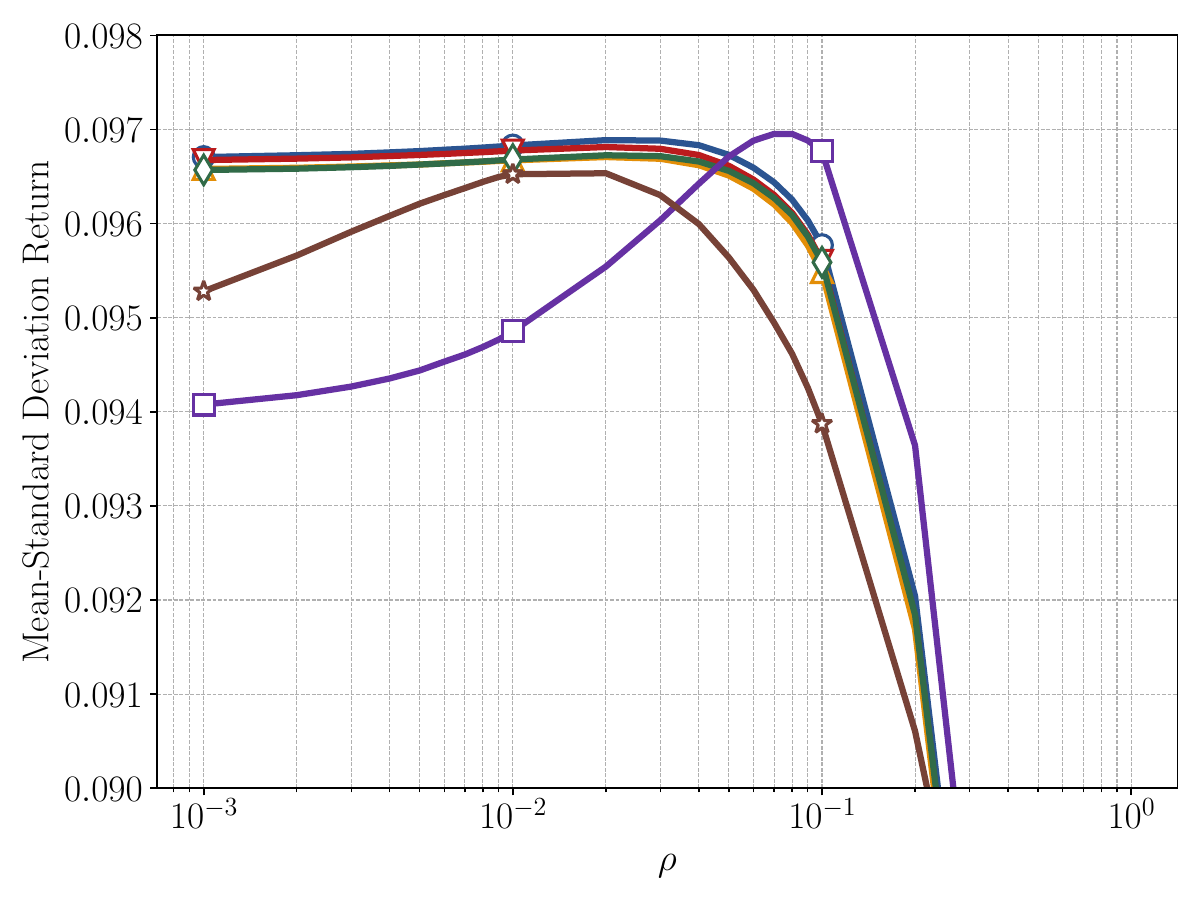}
        % \caption{Caption 3}
    \end{subfigure} \\
    \begin{subfigure}[b]{\textwidth}
        \centering
        \includegraphics[width=\linewidth]{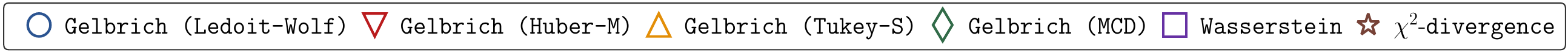}
    \end{subfigure}
    \caption{{Out-of-sample CVaR (top row) and mean-standard deviation return (bottom row) of distributionally robust portfolios as a function of $\rho$ for different training sample sizes: $N = 20$ (left), $N = 100$ (middle), and $N = 1{,}000$ (right).}}
    \label{fig:syntethic}
\end{figure}

\subsection{Index Tracking}

We next showcase the advantages of minimizing the Gelbrich risk in the context of an index-tracking application with real data, which involves a {\em nonlinear} loss function. The goal of index tracking is to construct a portfolio of $n-1$ assets with respective returns $\xi_1,\ldots,\xi_{n-1}$ that displays a similar performance as a market index with return $\xi_n$. We thus seek a portfolio from within the feasible set $\Omega = \{ \w \in \R^{n}: \sum_{i=1}^{n-1} \w_i = 1, \; \w_n = -1, \; \w_i \ge 0 \; \forall i<n \}$ that minimizes the Gelbrich error defined in Section~\ref{sec:indextracking}, that is, the worst-case expectation of $|\w^\top\xi|^p$ across all distributions in a Gelbrich ambiguity set. Our experiments are based on three standard instances of the index-tracking problem described in~\cite{ref:bruni2016dataset}. Table~\ref{ch4:table:dataset} summarizes the underlying datasets, which comprise weekly return time series of the market index (\texttt{DowJones}, \texttt{FF49Industries}, or \texttt{FTSE100}) and the $n-1$ assets used for replication. Each dataset covers a backtesting period partitioned into blocks of 12 weeks. We adopt the same rolling horizon approach as in~\cite{ref:bruni2016dataset}. At the beginning of each block, we use the previous 52 weeks of return data to estimate the sample mean $\msa$ and the sample covariance matrix $\covsa$, and we compute a portfolio vector $\w^\star\in\Omega$ with minimal Gelbrich error. At the beginning of each week in the current block, we rebalance the portfolio to the target allocation~$\w^\star$ and then calculate its out-of-sample error $|(\w^\star)^\top \widehat\xi|^p$ using the return data~$\widehat\xi$ over the next week.

\begin{table}[ht!]
\centering
\caption{{Characteristics of the standard index tracking instances from~\cite{ref:bruni2016dataset}.}}
\label{ch4:table:dataset}
{\begin{tabular}{l|c|c|c|c}
    Index &     \# of Assets  & Backtesting Interval    &      \# of Weeks      & \# of 12-Week Blocks    \\ \hline\hline
    \texttt{DowJones}    &     28    & Feb 1990--Apr 2016 &    1363    &110    \\ \hline
    \texttt{FF49Industries}   &     49    & Jul 1969--Jul 2015 &    2325    & 190    \\ \hline
    \texttt{FTSE100}     &     83    & Jul 2002--Apr 2016 &    717    & 56    \\ 
\end{tabular}}
\end{table}

We begin by studying the Gelbrich tracking problem with $p=1$, considering different estimators for the mean vector and the covariance matrix. In particular, we employ the \texttt{Ledoit-Wolf}, \texttt{Huber-M}, \texttt{Tukey-S}, and \texttt{MCD} estimators to construct the Gelbrich ambiguity sets. Our objective is to evaluate how the out-of-sample tracking error of the resulting Gelbrich portfolios varies with the radius of the ambiguity set. Table~\ref{tab:error_ratios} presents the ratio of the minimum tracking error---corresponding to the optimal radius of the ambiguity set---to the non-robust tracking error, obtained when the radius is set to zero. A ratio below one indicates that incorporating robustness improves the out-of-sample performance of the tracking portfolio. As in Section~\ref{sec:linear-portfolio-experiment}, we find that accounting for distributional uncertainty is advantageous. Specifically, for most estimators and datasets, the minimum tracking error is achieved at a strictly positive radius of the Gelbrich ambiguity set. In particular, working with Gelbrich ambiguity sets centered at \texttt{Huber-M}, \texttt{Tukey-S}, and \texttt{MCD} estimators leads to substantial and consistent improvements across all tested indices, underscoring once again the benefits of combining techniques from robust statistics with distributionally robust optimization models.

\begin{table}[!ht]
    \centering
    \caption{{Ratio of the tracking error attained by the best Gelbrich portfolio (optimal choice of $\rho$) to that attained by the corresponding non-robust portfolio ($\rho=0$).}}
    \label{tab:error_ratios}
    {\begin{tabular}{l|c|c|c|c}
        Index & \texttt{Ledoit-Wolf} & \texttt{Huber-M} & \texttt{Tukey-S} & \texttt{MCD} \\ \hline \hline
        \texttt{DowJones} & $\frac{0.3455}{0.3455} = 1.0000$ & $\frac{0.3428}{0.3776} \approx 	0.9078$ & $\frac{0.4199}{0.4910} \approx 0.8552$ & $\frac{0.3513}{0.4069} \approx 0.8633$ \\ \hline
        \texttt{FF49Industries} & $\frac{0.0068}{0.0085} = 0.8000$ & $\frac{0.0104}{0.0271} \approx 0.3838$ & $\frac{0.0075}{0.0119} \approx 0.6303$ & $\frac{0.0104}{0.0271} \approx 0.3838$ \\ \hline
        \texttt{FTSE100} & $\frac{0.2896}{0.2896} = 1.0000$ & $\frac{0.2544}{0.2709} \approx 0.9391$ & $\frac{0.2585}{0.2755} \approx 0.9383$ & $\frac{0.2544}{0.2709} \approx 0.9391$ 
    \end{tabular}}
\end{table}

In the next experiment, we benchmark the Gelbrich portfolio corresponding to the best-performing mean-covariance estimator for each dataset against the solutions of three standard distributionally robust optimization models. The first two models employ ambiguity sets of radius $\rho$ centered at the empirical asset return distribution, defined with respect to the Wasserstein distance and the $\chi^2$-divergence, respectively. Tractable reformulations of these optimization problems can be derived from \citep[Theorem~6.3]{ref:esfahani2018data} and \citep[Lemma~1]{duchi2021learning}. The third model relies on the celebrated Delage-Ye ambiguity set~\cite{ref:delage2010distributionally}, defined as
\[
    \mathcal A_\rho(\msa,\covsa) = \left\{ \QQ \in \mathcal M_2: \begin{array}{l} 
    (\EE_{\QQ}[\xi] - \msa)^\top \covsa^{-1} (\EE_{\QQ}[\xi]- \msa) \le \kappa \rho^2 \\
    \EE_{\QQ}[(\xi - \EE_{\QQ}[\xi])(\xi - \EE_{\QQ}[\xi])^\top ] \preceq (1 + (1-\kappa)\rho^2) \covsa
    \end{array}\right\}.
\]
Similar to the Gelbrich ambiguity set, the Delage-Ye ambiguity set imposes constraints only on the first and second moments of the asset return distribution. To ensure a fair comparison, in each experiment we use the same estimators for both ambiguity sets. The Delage-Ye ambiguity set involves two tuning parameters, that is, a size parameter $\rho \in \mathbb{R}_+$ and a weight parameter $\kappa \in [0,1]$, which balances the ambiguity in the mean against that in the covariance matrix. Using standard reformulation techniques from~\cite{ref:delage2010distributionally} together with the $\mathcal{S}$-lemma, the corresponding worst-case risk can be reformulated as a tractable semidefinite program. In all experiments, we set $\kappa = 0.5$. Numerical tests (not reported here) indicate that our results are robust to the choice of $\kappa \in \{0.5, 0.9, 0.99\}$.

Figure~\ref{fig:tracking} reports the average tracking error of the different portfolios over the backtesting interval as a function of~$\rho$. When $\rho$ is optimally chosen for each method, the Gelbrich portfolio consistently achieves lower out-of-sample tracking error than the benchmarks across all three datasets. This finding reaffirms that the Gelbrich approach, which relies solely on the first and second moments while accounting for their uncertainty, can outperform distributionally robust methods that exploit richer distributional information.

\begin{figure}[!tb]
    \centering
    \begin{subfigure}[b]{0.32\columnwidth}
        \centering
        \includegraphics[width=\linewidth]{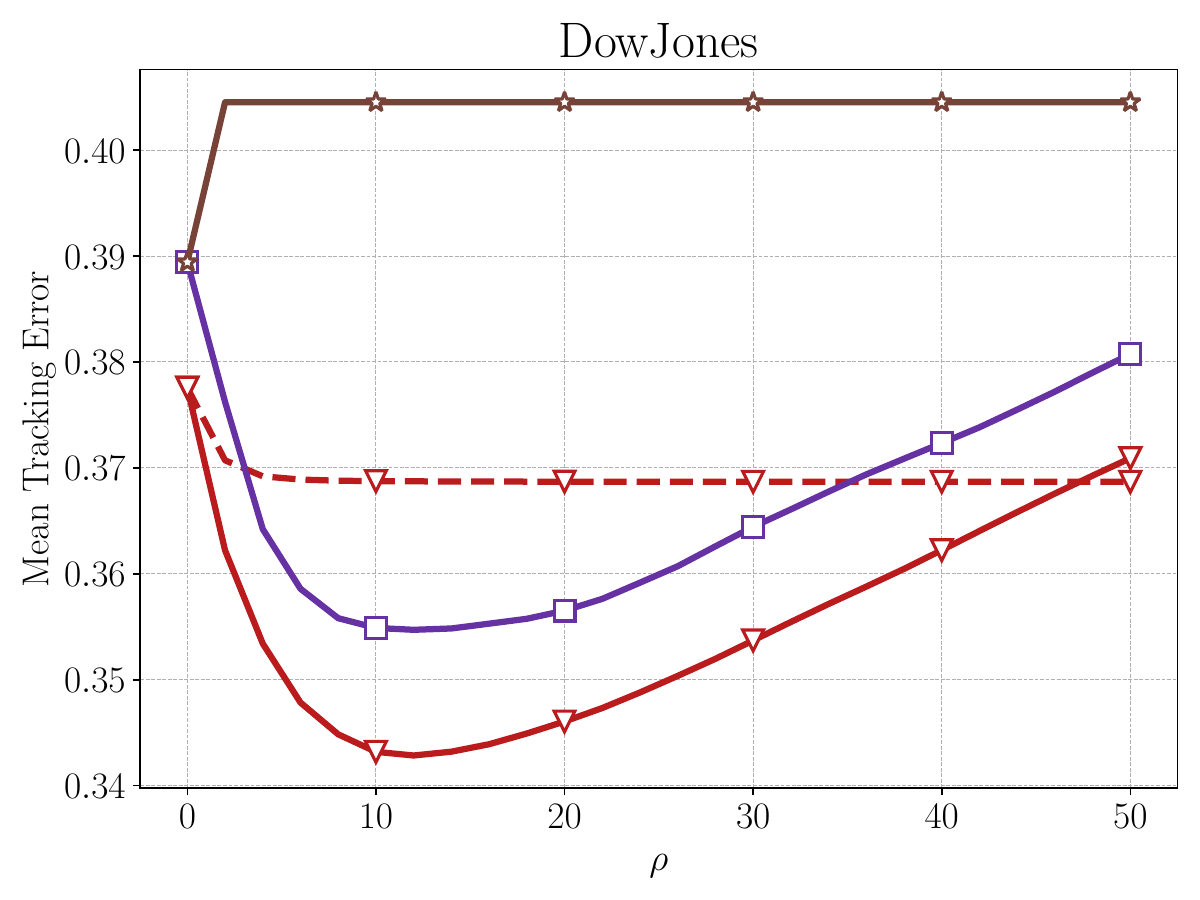}
    \end{subfigure}
    \begin{subfigure}[b]{0.32\columnwidth}
        \centering
        \includegraphics[width=\linewidth]{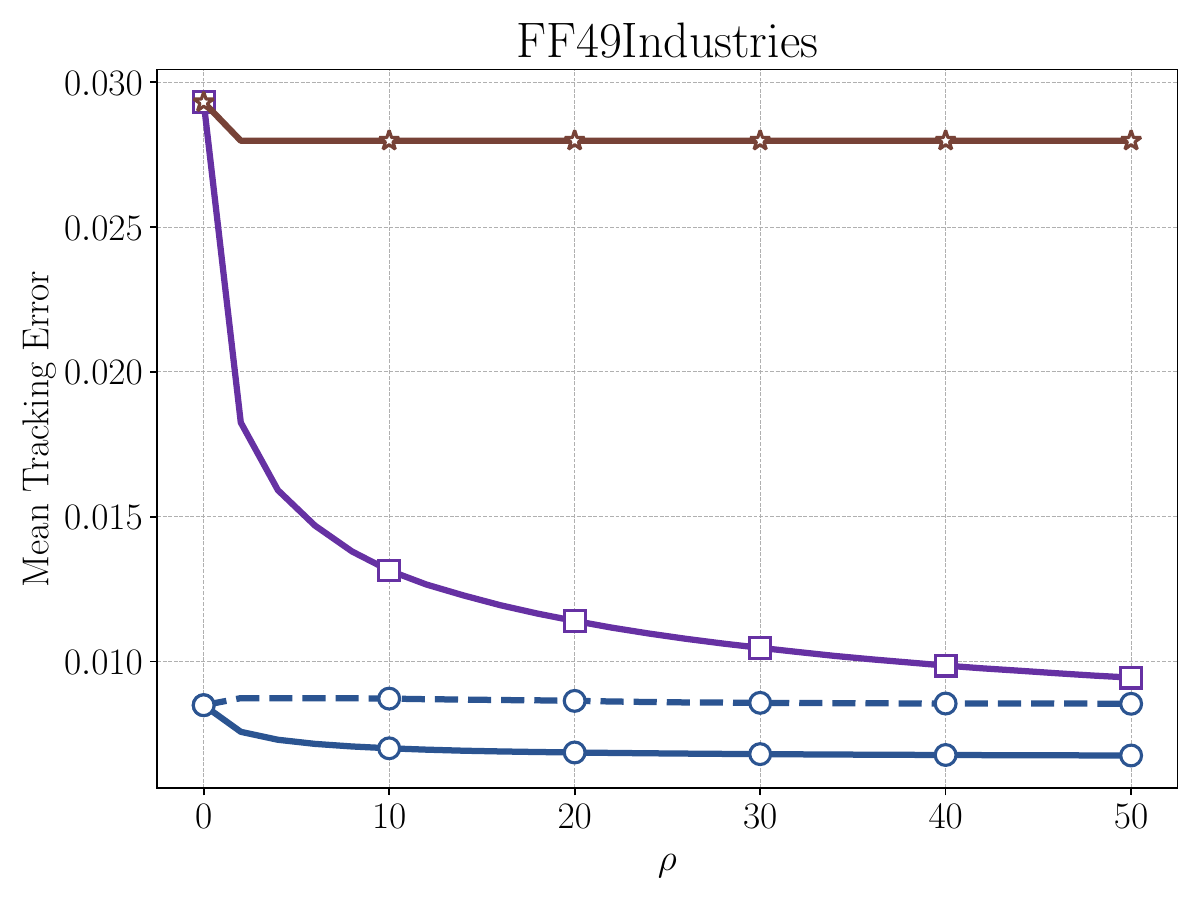}
    \end{subfigure}
    \begin{subfigure}[b]{0.32\columnwidth}
        \centering
        \includegraphics[width=\linewidth]{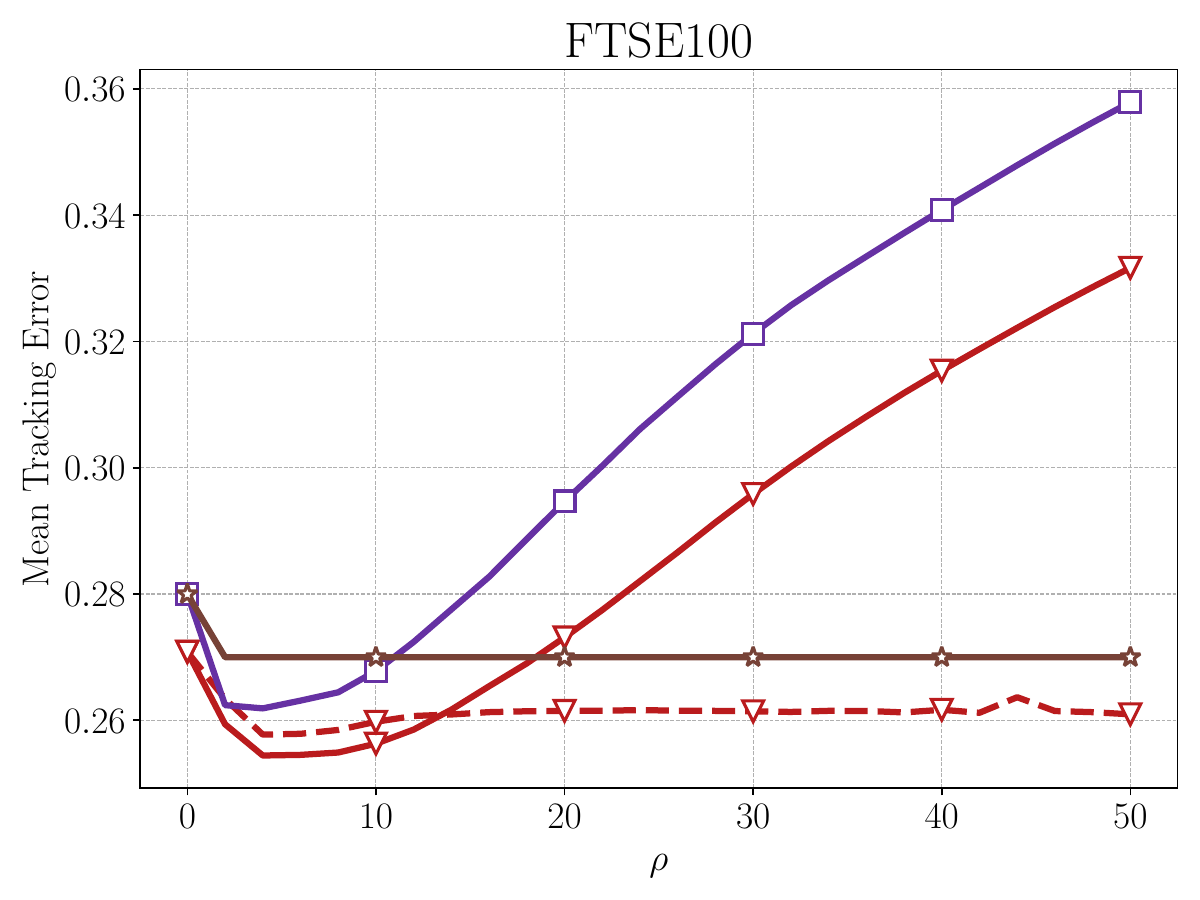}
    \end{subfigure} \\
    \begin{subfigure}[b]{\textwidth}
        \centering
        \includegraphics[width=\linewidth]{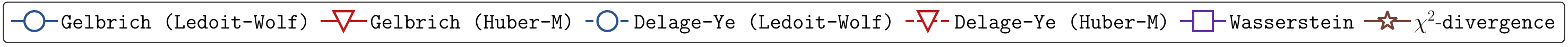}
    \end{subfigure}

    \caption{
        {Out-of-sample tracking error of distributionally robust tracking portfolios averaged over the backtesting interval as a function of~$\rho$ (for better readability, only the best-performing Gelbrich and Delage-Ye portfolios are shown)}.
    }
    \label{fig:tracking}
\end{figure}

In practice, the radius~$\rho$ of the ambiguity set must be estimated from data. To this end, we conduct a fully data-driven experiment in which $\rho$ is chosen for each method via 5-fold time-series cross-validation on the training dataset. Concretely, for each 52-week training window, we use the \texttt{TimeSeriesSplit} module from the \texttt{scikit-learn} library to generate five expanding-window folds. As all portfolio optimization problems are implemented with a \texttt{scikit-learn}-compatible API, this cross-validator can be directly integrated with \texttt{GridSearchCV}. The framework automatically evaluates candidate values $\rho \in \{0, 2, 4, \dots, 20\}$ across the five folds and selects the one that achieves the best average validation performance. The model is then re-fitted on the full 52-week training window using the selected $\rho$ and subsequently evaluated on the following 12-week test period.

By selecting $\rho$ in a data-driven manner at every step, our backtest mimics real investment practice, where model robustness is dynamically adjusted to reflect evolving market conditions in the training data. This approach yields a more reliable estimate of the achievable out-of-sample performance of each strategy. The aggregated mean tracking errors over the full backtesting horizon are reported in Table~\ref{table:backtest}. This experiment provides further evidence that combining robust estimation with robust optimization methods based on the first and second moments can be highly effective.}

\begin{table}[!tb]
    \centering
    \caption{{Out-of-sample tracking error of fully data-driven robust portfolios.}}
    \label{table:backtest}
    {\begin{tabular}{c|c|c|c|c} 
        Estimator & Ambiguity Set & \texttt{DowJones} & \texttt{FF49Industries} & \texttt{FTSE100} \\[1pt] \hline\hline 
        \multirow{2}{*}{\texttt{Ledoit-Wolf}} & \texttt{Gelbrich} & 0.3467 & 0.0069 & 0.2910 \\\cline{2-5} 
        & \texttt{Delage-Ye} & 0.3479 & 0.0085 & 0.2934 \\[1pt] \hline 
        \multirow{2}{*}{\texttt{Huber-M}} & \texttt{Gelbrich} & 0.3459 & 0.0128 & 0.2601 \\\cline{2-5} 
        & \texttt{Delage-Ye} & 0.3732 & 0.0277 & 0.2604 \\[1pt] \hline 
        \multirow{2}{*}{\texttt{Tukey-S}} & \texttt{Gelbrich} & 0.4270 & 0.0075 & 0.2645 \\\cline{2-5} 
        & \texttt{Delage-Ye} & 0.5002 & 0.0162 & 0.2606 \\[1pt] \hline 
        \multirow{2}{*}{\texttt{MCD}} & \texttt{Gelbrich} & 0.3610 & 0.0128 & 0.2589 \\\cline{2-5} 
        & \texttt{Delage-Ye} & 0.4052 & 0.0277 & 0.2595 \\[1pt] \hline 
        \multirow{2}{*}{\begin{tabular}{@{}c@{}} \texttt{Empirical} \\[-1ex] \texttt{Distribution} \end{tabular} } & \texttt{Wasserstein} & 0.3602 & 0.0116 & 0.2662 \\\cline{2-5} 
        & \texttt{$\chi^2$-divergence} & 0.4001 & 0.0282 & 0.2737 \\[1pt]
    \end{tabular}}
\end{table}

\textbf{Acknowledgments.} We thank Erick Delage for valuable comments on an earlier version of this paper. This research was supported by the Swiss National Science Foundation under the NCCR Automation (grant agreement 51NF40\_180545) and under an Early Postdoc.Mobility Fellowship awarded to the second author (grant agreement P2ELP2\_195149).

\bibliographystyle{myabbrvnat}
\bibliography{bibliography}

\newpage
% \renewcommand\appendixname{Online Appendix}
% \appendix
% \section*{Mean-Covariance Robust Risk Measures}

% \renewcommand\appendixname{Part}
% \renewcommand{\thesection}{\Roman{section}}
% \renewcommand{\thesubsection}{\Roman{section}.\roman{subsection}}
% \setcounter{section}{0}
\begin{center}
    \textbf{Appendix: Mean-Covariance Robust Risk Measurement}
\end{center}

\begin{appendix}
\section{Auxiliary Results}
In this section we analyze the uncertainty sets $\U_\rho(\msa, \covsa)$ and $\V_\rho(\msa, \covsa)$. Our results rely on the following semidefinite programming representation of the squared Gelbrich distance.
\begin{lemma}
    \label{lemma:Gelbrich-SDP}
    For any mean-covariance pairs~$(\m_1, \cov_1)$ and $(\m_2, \cov_2)$ in~$\R^n \times \PSD^n$, we have
    \begin{equation} \label{eq:Gelbrich-SDP}
         \Gelbrich^2\big( (\m_1, \cov_1), (\m_2, \cov_2) \big) = 
         \left\{
         \begin{array}{cl}
         \DS \min_{C \in \R^{n \times n}} &\DS \| \mu_1 - \mu_2 \|^2 + \Tr{\cov_1 + \cov_2 - 2C} \\
         \st & \begin{bmatrix} \cov_1 & C \\ C^\top & \cov_2 \end{bmatrix} \succeq 0. 
        \end{array}
        \right.
    \end{equation} 
\end{lemma}
\begin{proof}[Proof of Lemma~\ref{lemma:Gelbrich-SDP}]
    The claim follows from \cite[Proposition~2]{ref:malago2018wasserstein}. We emphasize that the same result has also been reported in~\cite{ref:cuesta1996lower, ref:dowson1982frechet, ref:givens1984class, ref:knott1984optimal, ref:olkin1982distance}. 
\end{proof}

\subsection{Support Functions of \texorpdfstring{$\U_\rho(\msa, \covsa)$}{U} and \texorpdfstring{$\V_\rho(\msa, \covsa)$}{V}}
We now show that the support functions of the convex sets $\U_\rho(\msa, \covsa)$ and $\V_\rho(\msa, \covsa)$ can be computed efficiently. This will allow us to examine the Gelbrich risk of nonlinear portfolio loss functions in Section~\ref{sec:nonlinearportfolio} and the Gelbrich mean-variance risk in Appendix~\ref{sect:mean:var}. The support functions of~$\mathcal U_{\rho}(\msa, \covsa)$ and~$\V_\rho(\msa, \covsa)$ are also relevant for robust optimization. For example, a robust constraint that requires a concave function $h(\mu,\cov)$ to be non\-positive for all $(\mu,\cov)\in \mathcal U_{\rho}(\msa, \covsa)$ can be reformulated as a convex constraint that involves the convex conjugate of~$-h(\mu, \cov)$ as well as the support function of $\mathcal U_{\rho}(\msa, \covsa)$ \cite[Theorem~2]{ref:ben2015deriving}. Formally, we have
\begin{align*}
h(\mu,\cov) \le 0\quad &\forall (\mu,\cov)\in \mathcal U_{\rho}(\msa, \covsa)\quad \\ 
&\iff \quad \exists \; q \in \R^n, Q \in \Sym^n \text{ such that } (- h)^* (q, Q) + \delta^*_{\mathcal U_{\rho}(\msa, \covsa)}(q, Q) \le 0.
\end{align*}
This constraint is computationally tractable for many commonly used constraint functions because the support function of $\mathcal U_{\rho}(\msa, \covsa)$ can be represented as the optimal value of a tractable conic minimization problem. 

\begin{proposition}
    \label{proposition:support:U}
    For any $\rho \ge 0$, $q\in\R^n$ and $Q\in\Sym^n$, the support function of $\mathcal U_{\rho}(\msa, \covsa)$ satisfies
    \begin{align*}
    \delta^*_{\mathcal U_{\rho}(\msa, \covsa)}(q, Q) =
    \left\{ \begin{array}{cl}
    \inf & q^\top \msa + \tau + \gamma \big(\rho^2 - \Tr{\covsa}\big) + \Tr{Z} \\[1ex]
    \st & \gamma \in \R_+, \; \tau \in\R_+,\; Z \in \mathbb{S}^n_+ \\[1ex]
    & \begin{bmatrix} \gamma I - Q & \gamma \covsa^{\frac{1}{2}} \\ \gamma \covsa^{\frac{1}{2}} & Z \end{bmatrix} \succeq 0, \quad  \Bigg\| \begin{bmatrix} q \\ \tau - \gamma \end{bmatrix} \Bigg\| \leq \tau + \gamma.
    \end{array} \right.
    \end{align*}
\end{proposition}

\begin{proof}[Proof of Proposition~\ref{proposition:support:U}]
    Evaluating the support function of $\mathcal U_{\rho}(\msa, \covsa)$ at a given point~$(q, Q) \in \R^n \times \Sym^n$ amounts to solving the finite convex program
    \begin{equation*}
    \delta^*_{\mathcal U_{\rho}(\msa, \covsa)}(q, Q) = 
    \left\{ \begin{array}{cl}
    \Sup{\m, \cov \succeq 0} & q^\top \m + \Tr{Q \cov} \\
    \st & \| \m - \msa \|^2 + \Tr{\cov + \covsa - 2 \big( \covsa^\half \cov \covsa^\half \big)^\half} \leq \rho^2.
    \end{array} \right.
    \end{equation*}
    Using the semidefinite programming reformulation for the squared Gelbrich distance on the right hand side of~\eqref{eq:Gelbrich-SDP} and introducing an auxiliary variable~$M$ that equals~$\mu\mu^\top$ at optimality, we then obtain
    \begin{align*}
    \delta^*_{\mathcal U_{\rho}(\msa, \covsa)}(q, Q) 
    =& \left\{
    \begin{array}{cl}
    \Sup{\m, C, \cov \succeq 0} & q^\top \m + \Tr{Q \cov} \\
    \st & \Tr{\mu \mu^\top - 2 \msa \mu^\top + \msa \msa^\top} + \Tr{\cov + \covsa - 2 C} \leq \rho^2 \\[1ex]
    & \begin{bmatrix} \cov & C \\ C^\top & \covsa \end{bmatrix} \succeq 0
    \end{array} 
    \right. \\
    =& \left\{
    \begin{array}{cl}
    \sup & q^\top \m + \Tr{Q \cov} \\[1ex]
    \st & \mu \in \R^n, \, C \in \R^{n \times n}, \, \cov \in \PSD^n, \, M \in \PSD^n \\ [1ex]
    & \Tr{M - 2 \msa \mu^\top} + \| \msa \|^2 + \Tr{\cov + \covsa - 2 C} \leq \rho^2 \\[1ex]
    & \begin{bmatrix} \cov & C \\ C^\top & \covsa \end{bmatrix} \succeq 0, ~ \begin{bmatrix} M & \mu \\ \mu^\top & 1 \end{bmatrix} \succeq 0.
    \end{array} 
    \right.
    \end{align*}
    By strong conic programming duality \cite[Theorem~1.4.2]{ref:ben2001lectures}, the resulting semidefinite program is equivalent to
    \begin{align}
    \label{eq:dual:of:sdp}
    \delta^*_{\mathcal U_{\rho}(\msa, \covsa)}(q, Q) 
    =& \left\{\begin{array}{cl}
    \inf & \gamma \left( \rho^2 - \| \msa \|^2 - \Tr{\covsa} \right) + \Tr{A_{22} \covsa} + \beta \\[1ex]
    \st & \gamma \in \R_+, ~ \begin{bmatrix} A_{11} & \gamma I \\ \gamma I & A_{22} \end{bmatrix} \in \PSD^{2n}, ~ \begin{bmatrix} B & \gamma \msa + \frac{q}{2} \\ \left( \gamma \msa + \frac{q}{2} \right)^\top & \beta \end{bmatrix} \in \PSD^{n+1}\\
    & A_{11} \preceq \gamma I - Q  , ~ B \preceq \gamma I.
    \end{array} \right.
    \end{align}
    Strong duality holds because $\gamma = \max \{ \lambda_{\max}(Q), 0 \} + 2$, $A_{11} = I$, $A_{22} = 2 \gamma^2 I$, $B = I$ and $\beta = \| \gamma \msa + \frac{q}{2} \|^2+1$ represents a Slater point for the dual problem. %and because the primal problem is trivially feasible. 
    Next, we simplify the semidefinite program~\eqref{eq:dual:of:sdp} by eliminating the decision variables $A_{11}$ and $B$, each of which appears in two opposing matrix inequalities. If $\gamma > 0$, then~$A_{22}$ has full rank thanks to \cite[Corollary~8.2.2]{ref:bernstein2009matrix}. In this case, we obtain the following equivalences by Schur complementing the two matrix inequalities involving~$A_{11}$.
    \begin{align}
        \label{eq:remove:A}
        \begin{array}{ll}
            \begin{bmatrix} A_{11} & \gamma I \\ \gamma I & A_{22} \end{bmatrix} \succeq 0, ~ A_{11} \preceq \gamma I - Q 
            &\iff
            \gamma^2 A_{22}^{-1} \preceq A_{11} \preceq \gamma I - Q \\
            &\iff
            \begin{bmatrix} \gamma I - Q & \gamma I \\[1ex] \gamma I & A_{22} \end{bmatrix} \succeq 0 
        \end{array}
    \end{align}
    If $\gamma = 0$, on the other hand, we trivially have
    \begin{align*}
        \begin{bmatrix} A_{11} & \gamma I \\ \gamma I & A_{22} \end{bmatrix} \succeq 0, ~ A_{11} \preceq \gamma I - Q 
        \iff
        0 \preceq A_{11} \preceq - Q, ~ A_{22} \succeq 0
        \iff
        \begin{bmatrix} - Q & 0 \\[1ex] 0 & A_{22} \end{bmatrix} \succeq 0.
    \end{align*}
    Therefore, all equivalences in~\eqref{eq:remove:A} hold for any $\gamma \geq 0$. 
    If $\beta > 0$, then Schur complementing the two matrix inequalities involving~$B$ yields
    \begin{align}
        \begin{bmatrix} B & \gamma \msa + \frac{q}{2} \\ \left( \gamma \msa + \frac{q}{2} \right)^\top & \beta \end{bmatrix} \succeq 0  , ~ B \preceq \gamma I
        &\iff 
        \frac{1}{\beta} \left( \gamma \msa + \frac{q}{2} \right) \left( \gamma \msa + \frac{q}{2} \right)^\top \preceq B \preceq \gamma I \notag\\
        &\iff
        \begin{bmatrix} \gamma I & \gamma \msa + \frac{q}{2} \\[1ex] (\gamma \msa + \frac{q}{2})^\top & \beta \end{bmatrix} \succeq 0. \label{eq:remove:B}
    \end{align}
    If $\beta = 0$, on the other hand, we have 
    \begin{align*}
        \begin{bmatrix} B & \gamma \msa + \frac{q}{2} \\ \left( \gamma \msa + \frac{q}{2} \right)^\top & \beta \end{bmatrix} \succeq 0  , ~ B \preceq \gamma I
        &\iff 
        0 \preceq B \preceq \gamma I
        \iff
        \begin{bmatrix} \gamma I & 0 \\[1ex] 0 & 0 \end{bmatrix} \succeq 0,
    \end{align*}
    where the first equivalence holds because $\beta=0$ implies via the first matrix inequality that~$\gamma \msa + \frac{q}{2} = 0$; see~\cite[Corollary~8.2.2]{ref:bernstein2009matrix}. Therefore, all equivalences in~\eqref{eq:remove:B} hold for any $\beta \geq 0$. 
    In summary, we have thus shown that
    \begin{align}
    \label{eq:support:U:schur}
    \delta^*_{\mathcal U_{\rho}(\msa, \covsa)}(q, Q) =
    \left\{\begin{array}{cl}
    \inf\limits_{\gamma, \beta, A_{22}} & \gamma \left( \rho^2 - \| \msa \|^2 - \Tr{\covsa} \right) + \Tr{A_{22} \covsa} + \beta \\[2ex]
    \st & \begin{bmatrix} \gamma I - Q & \gamma I \\[1ex] \gamma I & A_{22} \end{bmatrix} \succeq 0, ~
    \begin{bmatrix} \gamma I & \gamma \msa + \frac{q}{2} \\[1ex] (\gamma \msa + \frac{q}{2})^\top & \beta \end{bmatrix} \succeq 0.
    \end{array} \right.
    \end{align}
    For any fixed $\gamma \ge 0$, problem~\eqref{eq:support:U:schur} decomposes into two separate minimization problems over~$\beta$ and~$A_{22}$, respectively. If~$\gamma>0$, the partial minimization problem over~$\beta$ reduces~to
    \begin{align*}
        &\phantom{=} \inf \left\{ \beta: \begin{bmatrix} \gamma I & \gamma \msa + \frac{q}{2} \\[1ex] (\gamma \msa + \frac{q}{2})^\top & \beta \end{bmatrix} \succeq 0 \right\} \\
        &= \inf \left\{ \beta: \beta - \gamma \| \msa \|^2 - q^\top \msa - \frac{1}{4\gamma} \| q \|^2 \geq 0 \right\} \\
        &= \inf \left\{ \eta + \gamma \| \msa \|^2: \eta \ge q^\top \msa + \tau, \; \Bigg\| \begin{bmatrix} q \\ \tau - \gamma \end{bmatrix} \Bigg\| \leq \tau + \gamma, \; \tau \geq 0 \right\},
    \end{align*}
    where the first equality exploits a standard Schur complement argument, and the second equality follows from the substitution~$\eta\leftarrow \beta -\gamma\|\msa\|^2$ and from introducing an auxiliary variable~$\tau\ge 0$ subject to the hyperbolic constraint~$\tau\geq \|q\|^2/(4\gamma)$. Note also that the first and the last minimization problems in the above expression remain equivalent when $\gamma = 0$. Similarly, if $\gamma > 0$, then the partial minimization problem over~$A_{22}$ reduces to
    \begin{align*}
        &\phantom{=} \inf \left\{ \Tr{A_{22} \covsa}:  \begin{bmatrix} \gamma I - Q & \gamma I \\[1ex] \gamma I & A_{22} \end{bmatrix} \succeq 0 \right\} \\
        &= \inf \left\{ \Tr{Z}:  Z \succeq \covsa^{\half} A_{22} \covsa^{\half}, A_{22} \succeq \gamma^2 (\gamma I - Q)^{-1} \right\} \notag \\
        &= \inf \left\{ \Tr{Z}:  Z \succeq \gamma \covsa^{\half} (\gamma I - Q)^{-1} \gamma \covsa^{\half}  \right\} \notag \\
        &= \inf \left\{ \Tr{Z}:  \begin{bmatrix} \gamma I - Q & \gamma \covsa^{\half} \\[1ex] \gamma \covsa^{\half} & Z \end{bmatrix} \succeq 0 \right\},
    \end{align*}
    where the first and the third equalities follow from Schur complement arguments. As~$\covsa\succ 0$ by assumption, the first and the last minimization problems in the above expression remain equivalent when~$\gamma = 0$. The claim then follows by substituting the reformulated partial minimization problems into~\eqref{eq:support:U:schur} and eliminating~$\eta$.
\end{proof}

The next lemma shows that the unique maximizer of the optimization problem defining the support function of $\U_\rho(\msa, \covsa)$ can be computed in quasi-closed form.

\begin{lemma}
    \label{lemma:extremal:support:U}
    Suppose that $\covsa \succ 0$, $\rho>0$ and either $q \neq 0$ or $\lambda_{\max}(Q) > 0$.
    Then the optimization problem
    \begin{subequations}
    \begin{equation}
        \label{eq:support:U}
        \delta^*_{\mathcal U_{\rho}(\msa, \covsa)}(q, Q) = \Sup{(\m, \cov)\in \mathcal U_{\rho}(\msa, \covsa)}~ q^\top \m + \Tr{Q \cov} 
    \end{equation}
    is uniquely solved by
    \begin{align} \label{eq:extremal:support:U:value}
    \m\opt = \msa  + \frac{q}{2\gamma\opt} \quad \text{and}\quad \cov\opt = \left(I - \frac{Q}{\gamma\opt}\right)^{-1} \covsa \left(I - \frac{Q}{\gamma\opt}\right)^{-1},
    \end{align}
    where $\gamma\opt > \max \{ \lambda_{\max}(Q), 0 \}$ is the unique solution of the nonlinear algebraic equation
    \begin{align}
        \label{eq:extremal:support:U:FOC}
        \frac{\|q\|^2}{4\gamma^2} + \Tr{\covsa \left( I - \gamma (\gamma I - Q)^{-1} \right)^2} = \rho^2.
    \end{align}
    \end{subequations}
    In addition, if $Q \succeq 0$ then we have $\cov\opt \succeq \lambda_{\min}(\covsa) I$. 
\end{lemma}
\begin{proof}[Proof of Lemma~\ref{lemma:extremal:support:U}]
    From the proof of Proposition~\ref{proposition:support:U} we know that problem~\eqref{eq:support:U} is equivalent to the semidefinite program~\eqref{eq:support:U:schur}. We first prove that~$\gamma > 0$ for any 
    $(\gamma, \beta, A_{22})$ feasible in~\eqref{eq:support:U:schur}. To see this, assume for the sake of argument that~$\gamma=0$, in which case the two matrix inequalities in~\eqref{eq:support:U:schur} reduce to
    \begin{align*}
        \begin{bmatrix} - Q & 0 \\[1ex] 0 & A_{22} \end{bmatrix} \succeq 0\quad\text{and}\quad \begin{bmatrix} 0 & {q}/{2} \\[1ex] \left({q}/{2}\right)^\top & \beta \end{bmatrix} \succeq 0.
    \end{align*}
    However, these constraints are not satisfiable by any~$\beta$ and~$A_{22}$ under our assumption that either $q \neq 0$ or $\lambda_{\max}(Q) > 0$.
    Similarly, one can show that~$\gamma > \lambda_{\max}(Q)$ for any~$(\gamma, \beta, A_{22})$ feasible in~\eqref{eq:support:U:schur}. Indeed, we have
    \begin{align*}
        \begin{bmatrix} \gamma I - Q & \gamma I \\ \gamma I & A_{22} \end{bmatrix} \succeq 0, ~ \gamma > 0 \quad\implies\quad \gamma I - Q \succ 0 \quad \implies \quad \gamma > \lambda_{\max}(Q),
    \end{align*}
    where the first implication follows from~\cite[Corollary~8.2.2]{ref:bernstein2009matrix}, which requires $\gamma I - Q$ to have full rank whenever the off-diagonal block $\gamma I$ has full rank. Schur complementing the two matrix inequalities in~\eqref{eq:support:U:schur} yields
    \[
        A_{22}\succeq \gamma^2(\gamma I-Q)^{-1}\quad \text{and}\quad \textstyle \beta\geq \|\gamma \msa+\frac{q}{2}\|^2/\gamma,
    \]
    which is possible because $\gamma > \max \{ \lambda_{\max}(Q), 0 \}$. Using these relations to eliminate~$\beta$ and~$A_{22}$ from~\eqref{eq:support:U:schur} yields
    \begin{align}
        \label{eq:extremal:support:U}
        \begin{array}{ll}
        &\phantom{=} \DS \delta^*_{\mathcal U_{\rho}(\msa, \covsa)}(q, Q) \\
        &= \DS \inf\limits_{\substack{ \gamma > 0 \gamma > \lambda_{\max}(Q) }} q^\top \msa + \frac{\|q\|^2}{4\gamma} + \gamma \big(\rho^2 - \Tr{\covsa} \big) + \gamma^2 \Tr{(\gamma I - Q)^{-1} \covsa}.
        \end{array}
    \end{align}
    In the following we use~$f(\gamma)$ to denote the objective function of problem~\eqref{eq:extremal:support:U}. As $\covsa \succ 0$, $f(\gamma)$ satisfies 
    \[
        f(\gamma) \ge q^\top \msa + \frac{\|q\|^2}{4\gamma} + \gamma \big(\rho^2 - \Tr{\covsa} \big) + \lambda_{\min}(\covsa) \gamma^2 \Tr{(\gamma I - Q)^{-1}}
    \]
    for all $\gamma\ge \max\{\lambda_{\max}(Q), 0\}$.
    Thus, $f(\gamma)$ diverges as $\gamma$ decreases to $\max\{\lambda_{\max}(Q), 0\}$. Similarly, since~$\rho>0$, one readily verifies that~$f(\gamma)$ grows indefinitely as~$\gamma$ increases to infinity. 
    Noting that~$f(\gamma)$ is smooth and strictly convex on its domain, these insights reveal that problem~\eqref{eq:extremal:support:U} has a unique minimizer~$\gamma\opt \in ( \max\{0, \lambda_{\max}(Q)\},\infty)$. This minimizer can be found by solving the problem's first-order optimality condition~$f'(\gamma\opt)=0$, where
    \begin{align*}
        f'(\gamma) &= - \frac{\|q\|^2}{4\gamma^2} + \rho^2 - \Tr{\covsa} + 2 \gamma \Tr{(\gamma I - Q)^{-1} \covsa} - \gamma^2 \Tr{(\gamma I - Q)^{-2} \covsa } \\
        &= \rho^2 - \frac{\|q\|^2}{4\gamma^2} - \Tr{\covsa \big( I - \gamma (\gamma I - Q)^{-1} \big)^2}
    \end{align*}
    denotes the derivative of~$f(\gamma)$. Thus, the algebraic equation~\eqref{eq:extremal:support:U:FOC} represents the (necessary and sufficient) first-order optimality condition of problem~\eqref{eq:extremal:support:U}.
    In the remainder of the proof, we demonstrate that $(\m\opt, \cov\opt)$ as defined in~\eqref{eq:extremal:support:U:value} constitutes a global maximizer of problem~\eqref{eq:support:U}. To see this, note first that~$\cov\opt\succeq 0$ and
    \begin{align*}
    \Gelbrich^2\big((\m\opt, \cov\opt), (\msa, \covsa)\big)^2 = \frac{\|q\|^2}{4(\gamma\opt)^2} + \Tr{\covsa (I - \gamma\opt(\gamma\opt I - Q)^{-1})^2} = \rho^2,
    \end{align*}
    where the first equality follows from~\eqref{eq:extremal:support:U:value} and the definition of the Gelbrich distance, whereas the second equality holds because $\gamma\opt$ solves~\eqref{eq:extremal:support:U:FOC}. Thus, $(\m\opt, \cov\opt)$ is feasible in problem~\eqref{eq:support:U}. Furthermore, the objective function value of $(\m\opt, \cov\opt)$ in~\eqref{eq:support:U} evaluates to
    \begin{align*}
    &\phantom{=} q^\top \m\opt + \Tr{Q \cov\opt} \\
    &= q^\top \msa + \frac{\|q\|^2}{2\gamma\opt} +  (\gamma\opt)^2  \Tr{Q  (\gamma\opt I - Q)^{-1} \covsa (\gamma\opt I - Q)^{-1}} \\
    &= q^\top \msa + \frac{\|q\|^2}{2\gamma\opt} +  (\gamma\opt)^2 \Tr{(Q - \gamma\opt I + \gamma\opt I)  (\gamma\opt I - Q)^{-1} \covsa (\gamma\opt I - Q)^{-1}} \\
    &= q^\top \msa + \frac{\|q\|^2}{2\gamma\opt} - (\gamma\opt)^2 \Tr{\covsa (\gamma\opt I - Q)^{-1}} + (\gamma\opt)^3 \Tr{\covsa (\gamma I - Q)^{-2}} \\
    &= q^\top \msa + \frac{\|q\|^2}{2\gamma\opt} + \gamma\opt\left(\rho^2 - \frac{\|q\|^2}{4(\gamma\opt)^2} - \Tr{\covsa}\right) + (\gamma\opt)^2 \Tr{(\gamma\opt I - Q)^{-1} \covsa} \\
    &= q^\top \msa + \frac{\|q\|^2}{4\gamma\opt} + \gamma\opt\big(\rho^2  - \Tr{\covsa}\big) + (\gamma\opt)^2 \Tr{(\gamma\opt I - Q)^{-1} \covsa} = \delta^*_{\mathcal U_{\rho}(\msa, \covsa)}(q, Q),
    \end{align*}
    where the fourth equality holds because $\gamma\opt > 0$ solves~\eqref{eq:extremal:support:U:FOC}, 
    and the last equality follows from the optimality of $\gamma\opt$ in~\eqref{eq:extremal:support:U}.  Thus, $(\m\opt, \cov\opt)$ is optimal in~\eqref{eq:support:U}. As problem~\eqref{eq:support:U} has a linear objective function and a strictly convex feasible set, the maximizwer $(\m\opt, \cov\opt)$ is unique.
    To complete the proof, note that if~$Q \succeq 0$, then $(I - Q/\gamma\opt)^{-1} \succeq I$ because $\gamma\opt I \succ Q$, and thus it is easy to verify that $\cov\opt \succeq \lambda_{\min}(\covsa) I$.
\end{proof}

Next, we derive a tractable reformulation for the support function of $\V_\rho(\msa, \covsa)$.
% \subsection{Support Function of \texorpdfstring{$\V_\rho(\msa, \covsa)$}{V}}
% We now derive several technical results that will be needed to prove the results of Section~\ref{sec:nonlinearportfolio}. The uncertainty set $\mathcal V_{\rho}(\msa, \covsa)$ can conveniently be used in classical robust optimization. Indeed, a robust constraint that requires a concave function $h(\mu,M)$ to be nonpositive for all $(\mu,M)\in \mathcal V_{\rho}(\msa, \covsa)$ can be reformulated as a simple convex constraint involving the conjugate of $-h(\mu,M)$ and the support function of $\mathcal V_{\rho}(\msa, \covsa)$. 
% Formally, we have
% \begin{align*}
% h(\mu,\cov) \le 0\quad &\forall (\mu,\cov)\in \mathcal V_{\rho}(\msa, \covsa)\quad \\ 
% &\iff \quad \exists \; q \in \R^n, Q \in \Sym^n \text{ such that } (- h)^* (q, Q) + \delta^*_{\mathcal V_{\rho}(\msa, \covsa)}(q, Q) \le 0.
% \end{align*}
% This constraint is computationally tractable for many commonly used constraint functions because the support function of $\mathcal V_{\rho}(\msa, \covsa)$ is conic-representable.

\begin{proposition}
    \label{proposition:support:V}
    For any $\rho \in \RR_+$, $q\in\RR^n$ and $Q\in\Sym^n$, the support function of $\mathcal V_{\rho}(\msa, \covsa)$ satisfies
    \begin{align*}
    \delta^*_{\mathcal V_{\rho}(\msa, \covsa)}(q, Q) =
    \left\{ \begin{array}{cl}
    \inf & \gamma \big( \rho^2 - \| \msa \|^2 - \Tr{\covsa} \big) + \Tr{Z} + z \\[2ex]
    \st & \gamma \in \RR_+, \; \begin{bmatrix} \gamma I - Q & \gamma \covsa^\half \\[1ex] \gamma \covsa^\half & Z \end{bmatrix} \succeq 0, \; \begin{bmatrix} \gamma I - Q & \gamma \msa + \frac{q}{2} \\[1ex] (\gamma \msa + \frac{q}{2})^\top & z \end{bmatrix} \succeq 0.
    \end{array} \right.
    \end{align*}
\end{proposition}
\begin{proof}[Proof of Proposition~\ref{proposition:support:V}]
    Evaluating the support function of $\mathcal V_\rho(\msa, \covsa)$ at $(q, Q) \in \RR^n \times \Sym^n$ amounts to solving a maximization problem
    \begin{align*}
    \delta^*_{\mathcal V_{\rho}(\msa, \covsa)}(q, Q) =
    \left\{ \begin{array}{cl}
    \sup & \mu^\top q + \Tr{M Q} \\[1ex]
    \st & \mu \in \RR^n,\; M \in \PSD^n,\; \Gelbrich^2 \big( (\mu, M - \mu \mu^\top), (\msa, \covsa) \big) \leq \rho^2.
    \end{array} \right.
    \end{align*}
    The semidefinite programming representation of the squared Gelbrich distance in~\eqref{eq:Gelbrich-SDP} then yields
    \begin{equation*}
    \delta^*_{\mathcal V_{\rho}(\msa, \covsa)}(q, Q) =
    \left\{ \begin{array}{cl}
    \sup & \mu^\top q + \Tr{M Q} \\[1ex]
    \st & \mu \in \RR^n,\; M \in \PSD^n,\; C \in \RR^{n \times n} \\[1ex]
    & \| \mu - \msa \|^2 + \Tr{(M - \mu \mu^\top) + \covsa - 2 C} \leq \rho^2 \\[1ex]
    & \begin{bmatrix} M - \mu \mu^\top  & C \\ C^\top & \covsa \end{bmatrix} \succeq 0.
    \end{array} \right.
    \end{equation*}
    Note that the first constraint is equivalent to $\Tr{M - 2 \mu \msa^\top - 2 C} \leq \rho^2 - \| \msa \|^2 - \Tr{\covsa}$. Introducing a new decision variable $U \in \PSD^n$ with $U \succeq \mu \mu^\top$ and then applying a Schur complement argument, the above optimization problem can be recast as the following semidefinite program.
    \begin{equation*}
    \begin{array}{cl}
    \sup & \mu^\top q + \Tr{MQ} \\[1ex]
    \st & \mu \in \RR^n,\; M \in \PSD^n,\; U \in \PSD^n,\; C \in \RR^{n \times n} \\[1ex]
    & \Tr{M - 2 \mu \msa^\top - 2 C} \leq \rho^2 - \| \msa \|^2 - \Tr{\covsa} \\[1ex]
    & \begin{bmatrix} M - U  & C \\ C^\top & \covsa \end{bmatrix} \succeq 0 ,\; \begin{bmatrix} U & \mu \\ \mu^\top & 1 \end{bmatrix} \succeq 0
    \end{array}
    \end{equation*}
    By strong conic duality \cite[Theorem~1.4.2]{ref:ben2001lectures}, the above semidefinite program admits the dual
    \begin{align*}
    \begin{array}{cl}
    \DS \inf & \gamma \big( \rho^2 - \| \msa \|^2 - \Tr{\covsa} \big) + \Tr{\covsa A_{22}} + \beta \\ [2ex]
    \st & \gamma \in \RR_+ ,\; \begin{bmatrix} A_{11} & \gamma I \\ \gamma I & A_{22} \end{bmatrix} \succeq 0 ,\; \begin{bmatrix} B & \gamma \msa + \frac{q}{2} \\ (\gamma \msa + \frac{q}{2})^\top & \beta \end{bmatrix} \succeq 0 ,\; \gamma I - Q \succeq A_{11} \succeq B.
    \end{array}
    \end{align*}
    Strong duality holds because $\gamma = \max \{ \lambda_{\max}(Q), 0 \} + 4, A_{11} = 2I, A_{22} = \gamma^2 I, B = I, \beta = \| 2 \gamma \msa + q \|^2$ represents a Slater point for the dual problem. % and because the primal problem is trivially feasible. 
    As in the proof of Proposition~\ref{proposition:support:U}, one can eliminate the decision variables $A_{11}$ and $B$ by showing that the matrix inequalities $\gamma I - Q \succeq A_{11} \succeq B$ are binding at optimality. Thus, we can further simplify the dual problem to
    \begin{align*}
    \begin{array}{cl}
    \DS \inf & \gamma \big( \rho^2 - \| \msa \|^2 - \Tr{\covsa} \big) + \Tr{\covsa A_{22}} + \beta \\ [2ex]
    \st & \gamma \in \RR_+, \; \begin{bmatrix} \gamma I - Q & \gamma I \\ \gamma I & A_{22} \end{bmatrix} \in \PSD^{2n} ,\; \begin{bmatrix} \gamma I - Q & \gamma \msa + \frac{q}{2} \\ (\gamma \msa + \frac{q}{2})^\top & \beta \end{bmatrix} \in \PSD^{n+1}.
    \end{array}
    \end{align*}
    The claim then follows by renaming $\covsa^{\frac{1}{2}} A_{22} \covsa^{\frac{1}{2}}$ and $\beta$ as $Z$ and $z$, respectively.
\end{proof}

% \begin{remark}
%     One may replace the two matrix inequalities in the semidefinite programming reformulation of $\delta^*_{\mathcal V_{\rho}(\msa, \covsa)}(q, Q)$ derived in Proposition~\ref{proposition:support:V} with a single matrix inequality of the form
%     \begin{align*}
%     \begin{bmatrix} \gamma I - Q & \gamma \covsa^\half & \gamma \msa + \frac{q}{2} \\[1ex] \gamma \covsa^\half & Z & 0 \\[1ex] (\gamma \msa + \frac{q}{2})^\top & 0 & z \end{bmatrix} \succeq 0.
%     \end{align*}
%     This can potentially lead to better numerical stability.
% \end{remark}

In analogy to Lemma~\ref{lemma:extremal:support:U}, the next lemma shows that the unique maximizer of the optimization problem defining the support function of $\V_\rho(\msa, \covsa)$ can be computed in quasi-closed form.
\begin{lemma}
    \label{lemma:extremal:support:V}
    Suppose that $\covsa \succ 0$, $\rho>0$ and $\lambda_{\max}(Q) > 0$.
    Then the optimization problem
    \begin{subequations}
    \begin{equation}
        \label{eq:support:V}
        \delta^*_{\mathcal V_{\rho}(\msa, \covsa)}(q, Q) = \Sup{(\m, M)\in \mathcal V_{\rho}(\msa, \covsa)}~ q^\top \m + \Tr{Q M} 
    \end{equation}
    is uniquely solved by
    \begin{align} \label{eq:extremal:support:V:value}
    \m\opt = ( \gamma\opt I - Q)^{-1} \left(\gamma\opt \widehat \mu + \frac{q}{2} \right) \quad \text{and} \quad M \opt = \left( I - \frac{Q}{\gamma\opt}\right)^{-1} \covsa \left( I - \frac{Q}{\gamma\opt}\right)^{-1} + \m\opt {\m\opt}^\top, 
    \end{align}
    where $\gamma\opt > \lambda_{\max}(Q)$ is the unique solution of the nonlinear algebraic equation
    \begin{align}
        \label{eq:extremal:support:V:FOC}
        \| \widehat \mu - (\gamma I -Q)^{-1} (q/2 + \gamma \widehat \mu) \|^2 + \Tr{\covsa \left( I - \gamma (\gamma I - Q)^{-1} \right)^2} =\rho^2.
    \end{align}
    \end{subequations}
\end{lemma}
\begin{proof}[Proof of Lemma~\ref{lemma:extremal:support:V}]
    Using Schur complement argument, one can further simplifies the result of Proposition~\ref{proposition:support:V} and conclude that the support function of $\V_{\rho}(\msa, \covsa)$ is equivalent to the optimal value of the following univariate optimization problem
    \begin{equation} \label{eq:dual:single}
    \begin{array}{cl}
    \inf & \gamma \big(\rho^2 - \|\msa\|^2 - \Tr{\covsa}\big) + \gamma^2 \Tr{(\gamma I - Q)^{-1}\covsa} + (\gamma\msa + \frac{q}{2})^\top [\gamma I - Q]^{-1} (\gamma \msa + \frac{q}{2})\\
    \st & \gamma \in \RR_+, \;  \gamma I \succ Q.
    \end{array}
    \end{equation}
    Denote momentarily the objective function of~\eqref{eq:dual:single} as $f(\gamma)$. The gradient of $f$ for any feasible solution $\gamma$ satisfies
    \begin{align*}
    \nabla_\gamma f &=\rho^2 - \Tr{\covsa \left( I - \gamma (\gamma I - Q)^{-1} \right)^2} - \Big \| \msa -  (\gamma I -Q)^{-1} \left(\gamma\msa + \frac{q}{2} \right) \Big\|^2.
    \end{align*}
    Thus, if $\gamma\opt$ solves the nonlinear algebraic equation~\eqref{eq:extremal:support:V:FOC}, then $\gamma\opt$ also solves the first-order optimality condition of problem~\eqref{eq:dual:single}. This in turn implies that $\gamma\opt$ is the minimizer of problem~\eqref{eq:dual:single}. 
    
    As we have assumed that $\lambda_{\max}(Q) > 0$, it is easy to verify that as $\gamma$ approaches $\lambda_{\max}(Q)$, $\nabla_\gamma f$ tends to $-\infty$, and as $\gamma$ tends to infinity, $\nabla_\gamma f$ tends to $\rho^2 > 0$. This asserts that there exists a finite value $\gamma\opt$ that solves~~\eqref{eq:extremal:support:V:FOC}. Let $\cov\opt = M\opt - \mu\opt{\mu\opt}^\top$.
    Notice that
    \[
    \Gelbrich\big( (\mu\opt, \cov\opt), (\msa, \covsa) \big)^2 =  \Big \| (\gamma\opt I -Q)^{-1} \left(\gamma\opt \msa + \frac{q}{2} \right) - \msa  \Big \|^2 + \Tr{\covsa \left( I - \gamma\opt (\gamma\opt I - Q)^{-1} \right)^2} = \rho^2,
    \]
    where the first equality follows from substituting the value of $\mu\opt$ and $\cov\opt$ from~\eqref{eq:extremal:support:V:value} into the definition of the Gelbrich distance $\Gelbrich$, and the second equality is because $\gamma\opt$ solves~\eqref{eq:extremal:support:V:FOC}. This implies that $(\mu\opt, \cov\opt) \in \U_{\rho}(\msa, \covsa)$. Next, we show that $(\mu\opt, M\opt)$ will attain the value $\delta^*_{\mathcal V_{\rho}(\msa, \covsa)}(q, Q)$. Observe that
    \begin{subequations}
        \begin{align}
        \Tr{Q \cov\opt}
        =& \Tr{Q (\gamma\opt)^2 (\gamma\opt I - Q)^{-1} \covsa (\gamma\opt I - Q)^{-1}} \label{eq:indef:1} \\
        =&  \Tr{(Q - \gamma\opt I + \gamma\opt I) (\gamma\opt)^2 (\gamma\opt I - Q)^{-1} \covsa (\gamma\opt I - Q)^{-1}} \notag \\
        =& - (\gamma\opt)^2 \Tr{(\gamma\opt I - Q)^{-1} \covsa} + (\gamma\opt)^3 \Tr{(\gamma\opt I - Q)^{-1} \covsa (\gamma\opt I - Q)^{-1}} \notag \\
        =& \gamma\opt \big( \rho^2 - \| \msa - (\gamma\opt I - Q)^{-1}(\gamma\opt\msa + q/2) \|^2 - \Tr{\covsa} + \gamma\opt \Tr{(\gamma\opt I - Q)^{-1} \covsa} \big) \label{eq:indef:2} \\
        =& \gamma\opt \big( \rho^2 - \| \msa - \mu\opt \|^2 - \Tr{\covsa} + \gamma\opt \Tr{(\gamma\opt I - Q)^{-1} \covsa} \big) \label{eq:indef:3} \\
        =& \gamma\opt \big( \rho^2 - \| \msa\|^2 - \Tr{\covsa} \big) + (\gamma\opt)^2 \Tr{(\gamma\opt I - Q)^{-1} \covsa} + 2\gamma\opt \msa^\top \mu\opt - \gamma\opt \| \mu\opt \|^2, \notag
        \end{align}
    \end{subequations}
    where equality~\eqref{eq:indef:1} is from the definition of $\cov\opt$, equality~\eqref{eq:indef:2} follows from the fact that $\gamma\opt$ solves~\eqref{eq:extremal:support:V:FOC} and thus we can write
    \begin{align*}
    &- \gamma\opt \Tr{(\gamma\opt I - Q)^{-1} \covsa} + (\gamma\opt)^2 \Tr{(\gamma\opt I - Q)^{-1} \covsa (\gamma\opt I - Q)^{-1}} \\
    &\hspace{2cm} =
    \rho^2 - \| \msa - (\gamma\opt I - Q)^{-1}(q/2 + \gamma\opt\msa) \|^2 - \Tr{\covsa} + \gamma\opt \Tr{(\gamma\opt I - Q)^{-1} \covsa}.
    \end{align*}
    Finally, equality~\eqref{eq:indef:3} is from the definition of $\mu\opt$. We thus find
    \begin{align*}
    & q^\top \mu\opt +\Tr{Q (\cov\opt + \mu\opt(\mu\opt)^\top)}  \\
    &=\gamma\opt \big( \rho^2 - \| \msa\|^2 - \Tr{\covsa} \big) + (\gamma\opt)^2 \Tr{(\gamma\opt I - Q)^{-1} \covsa} + 2(\gamma\opt \msa + q/2)^\top \mu\opt + (\mu\opt)^\top (Q - \gamma\opt I) \mu\opt \\
    &= \gamma\opt \big( \rho^2 - \| \msa\|^2 - \Tr{\covsa} \big) + (\gamma\opt)^2 \Tr{(\gamma\opt I - Q)^{-1} \covsa} + (\gamma\opt \msa + q/2)^\top (\gamma\opt I - Q)^{-1} (\gamma\opt\msa + q/2)
    \end{align*}
    which equals the optimal value of the dual program~\eqref{eq:dual:single} because $\gamma\opt$ is also the minimizer of problem~\eqref{eq:dual:single}. This concludes the proof.
\end{proof}

% We highlight that evaluating the worst-case probability over optimal transport ambiguity sets has been previously studied in~\cite{ref:esfahani2018data, ref:blanchet2016robust}. In particular, a nonconvex formulation has been provided in~\cite[\S~2.4]{ref:blanchet2019quantifying}, which is then reformulated as mixed integer linear programs in \cite{ref:xie2019distributionally, ref:chen2018data, ref:ho2020distributionally, ref:ho2020strong, ref:shen2020chance}.
% Thanks to Corollary~\ref{cor:gelbrich-risk}, Corollary~\ref{corollary:WC-prob} provides a convex conservative approximation of the worst-case probability over the type-2 Wasserstein ambiguity set. 

\section{Proofs of Section~\ref{sec:properties}}

We give a short proof of Theorem~\ref{theorem:gelbrich} in order to keep this paper self-contained.

\begin{proof}[Proof of Theorem~\ref{theorem:gelbrich}]
    By the definition of the 2-Wasserstein distance, we have
    \begin{align*}
     \Wass^2(\QQ_1, \QQ_2) = &\DS \min_{\pi \in \Pi(\QQ_1, \QQ_2)}~\DS \int_{\R^n \times \R^n} \| \xi_1 - \xi_2 \|^2 \pi(\diff \xi_1, \diff \xi_2) \\[1em]
     =& \left\{
     \begin{array}{cl}
     \DS \min &\DS \| \mu_1 - \mu_2 \|^2 + \Tr{\cov_1 + \cov_2 - 2C} \\[1ex]
     \st &  \pi \in \Pi(\QQ_1, \QQ_2),\; C \in \R^{n \times n} \\[1ex]
     &\DS \int_{\R^n \times \R^n} \begin{bmatrix} \xi_1 \\ \xi_2 \end{bmatrix} \begin{bmatrix} \xi_1 \\ \xi_2 \end{bmatrix}^\top \pi(\diff \xi_1, \diff \xi_2) = \begin{bmatrix} \cov_1 & C \\ C^\top & \cov_2 \end{bmatrix},\quad  \begin{bmatrix} \cov_1 & C \\ C^\top & \cov_2 \end{bmatrix}\succeq 0,
    \end{array} \right. 
    \end{align*}
    where the second equality follows from the observations that $C$ is uniquely determined by $\pi$ (thanks to the equality constraint in the last line) and that the conic constraint is redundant (because the second-order momement matrix of $\pi$ is always positive semidefinite). 
    Relaxing the last optimization problem by removing all constraints that involve the original decision variable $\pi$, which in turn allows us to remove $\pi$ itself, we find
    \begin{equation*}
         \Wass^2(\QQ_1, \QQ_2) \geq 
         \left\{
         \begin{array}{cl}
         \DS \min_{C \in \R^{n \times n}} &\DS \| \mu_1 - \mu_2 \|^2 + \Tr{\cov_1 + \cov_2 - 2C} \\
         \st & \begin{bmatrix} \cov_1 & C \\ C^\top & \cov_2 \end{bmatrix} \succeq 0. 
        \end{array}
        \right.
    \end{equation*}
    By \cite[Proposition~2]{ref:malago2018wasserstein}, the above semidefinite program can be solved analytically, and its optimal value is given by $\Gelbrich^2( (\m_1, \cov_1), (\m_2, \cov_2))$. We emphasize that the same analytical formula has also been reported in~\cite{ref:cuesta1996lower, ref:dowson1982frechet, ref:givens1984class, ref:knott1984optimal, ref:olkin1982distance}. The claim then follows by taking square roots on both sides of~\eqref{eq:Gelbrich-SDP}. 
\end{proof}

Our proof of Theorem~\ref{theorem:Wasserstein=Gelbrich} relies on the following preparatory lemma.

\begin{lemma}[{2-Wasserstein distances of perfectly correlated distributions \cite[Theorem~2.13]{ref:cuestaalbertos1993optimal}}]
    \label{lemma:affine-coupling}
    If $\QQ_1\in\mc M_2$ has mean~0 and covariance matrix~$\cov_1 \in \PSD^n$, and if $\QQ_2=\QQ_1\circ g^{-1}$ is the pushforward distribution of $\QQ_1$ under the positive semidefinite linear transformation $g(\xi)=A\xi$ with $A \in \PSD^n$, then
    \begin{equation}
    \label{eq:affine-wasserstein}
    \Wass(\QQ_1, \QQ_2) = \sqrt{\EE_{\QQ_1} \left[ \| \xi - A \xi \|^2 \right] } = \sqrt{\Tr{\cov_1 + A \cov_1 A - 2 \cov_1^{\half} A \cov_1^{\half} }} .
    \end{equation}
\end{lemma}

Lemma~\ref{lemma:affine-coupling} implies that if $\xi$ follows a distribution $\QQ_1$ with vanishing mean and if $\QQ_2$ is defined as the distribution of $A\xi$ for some $A\in\PSD^n$, then the 2-Wasserstein distance between the perfectly correlated distributions $\QQ_1$ and $\QQ_2$ is given by~\eqref{eq:affine-wasserstein}. We can now use this (nontrivial) result to prove Theorem~\ref{theorem:Wasserstein=Gelbrich}.

\begin{proof}[Proof of Theorem~\ref{theorem:Wasserstein=Gelbrich}]
    By assumption, we have $\QQ_2=\QQ_1\circ f^{-1}$, where $f(\xi)=A\xi+b$ is an affine transformation with parameters $A\in\PSD^n$ and $b\in\R^n$. In the following, it will be more convenient to re-express this affine transformation as $f(\xi) = A (\xi-\m_1)+b'$, which involves the auxiliary parameter $b'=A\m_1+b$. Thus, we have 
    \[\mu_2 = \EE_{\QQ_2}[\xi] = \EE_{\QQ_1}[f(\xi)] = b',\]
    where the three equalities follow from the definition of $\m_2$, the integration formula for pushforward distributions and the definitions of $\m_1$ and $f$, respectively. Similarly, we find
    \[
    \cov_2 = \EE_{\QQ_2}[(\xi - \mu_2) (\xi - \mu_2)^\top] = \EE_{\QQ_1}[(f(\xi) - b') (f(\xi)- b')^\top ] = A \cov_1 A,
    \]
    where the second equality exploits our earlier insight that $\mu_2=b'$. Multiplying the above expression from both sides with $\cov_1^\frac{1}{2}$ yields the quadratic equation $ \cov_1^{\half} \cov_2 \cov_1^{\half} = (\cov_1^{\half} A \cov_1^{\half})^2$. As $\cov_1\succ 0$ by assumption, this equation is uniquely solved by $A = \cov_1^{-\half}( \cov_1^{\half} \cov_2 \cov_1^{\half} )^{\half} \cov_1^{-\half}$. This confirms that the affine function $f$ is uniquely determined by the first- and second-order moments of $\QQ_1$ and $\QQ_2$ and that its parameters are given by~\eqref{eq:optimal-affine}.
    
    Next, we define two distributions $\overline\QQ_1,\overline \QQ_2\in\mc M_2$ through the relations $\overline\QQ_1[\xi\in B]=\QQ_1[(\xi+\m_1)\in B]$ and $\overline\QQ_2[\xi\in B]=\QQ_2[(\xi+\m_2)\in B]$ for all Borel sets $B\in \mc B(\R^n)$. Thus, $\overline \QQ_1$ and $\overline \QQ_2$ are obtained by shifting $\QQ_1$ and $\QQ_2$ so that their mean vectors vanish. By construction, we then have $\overline\QQ_2=\overline \QQ_1\circ g^{-1}$ where $g(\xi)=A\xi$. By the definition of the 2-Wasserstein distance and the shifted distributions $\overline \QQ_1$ and $\overline \QQ_2$, we further have
    \begin{equation} \label{eq:wass-mu}
        \Wass^2(\QQ_1, \QQ_2) = \| \mu_1 - \mu_2 \|^2 + \Wass^2(\overline \QQ_1, \overline \QQ_2).
    \end{equation}
    Finally, as $\overline\QQ_2=\overline \QQ_1\circ g^{-1}$, we may use Lemma~\ref{lemma:affine-coupling} to conclude that
    \begin{align*}
        \Wass^2 (\overline \QQ_1, \overline \QQ_2) &= {\Tr{\cov_1 + A \cov_1 A - 2 \cov_1^{\half} A \cov_1^{\half} }} \\
        &= \Tr{\cov_1 + \Big( \cov_1^{-\half} \big( \cov_1^{\half} \cov_2 \cov_1^{\half} \big)^{\half} \cov_1^{-\half} \Big) \cov_1 \Big( \cov_1^{-\half} \big( \cov_1^{\half} \cov_2 \cov_1^{\half} \big)^{\half} \cov_1^{-\half} \Big)} \\
        & \quad - 2 \Tr{ \cov_1^{\half} \Big( \cov_1^{-\half} \big( \cov_1^{\half} \cov_2 \cov_1^{\half} \big)^{\half} \cov_1^{-\half} \Big) \cov_1^{\half} } \\
        &= \Tr{\cov_1 + \cov_2 - 2 \big( \cov_1^{\half} \cov_2 \cov_1^{\half} \big)^{\half} },
    \end{align*}
    where the second equality uses the expression for~$A$ in~\eqref{eq:optimal-affine}. The claim then follows by substituting the above expression into~\eqref{eq:wass-mu} and taking square roots on both sides.
\end{proof}

\begin{proof}[Proof of Theorem~\ref{theorem:finite-sample}]
    By~\cite[Theorem~1]{ref:bhatia2018bures}, the Gelbrich distance satisfies
    \begin{align*}
    \Gelbrich \big( (\msa_N, \covsa_N), (\mu, \cov) \big) &\leq \sqrt{\| \msa_N - \mu \|^2 + \| \covsa_N^{\frac{1}{2}} - \cov^{\frac{1}{2}} \|_F^2} \\
    &\leq \| \msa_N - \mu \| + \| \covsa_N^{\frac{1}{2}} - \cov^{\frac{1}{2}} \|_F \\
    &\leq \| \msa_N - \mu \| + \frac{1}{ \lambda_{\min}(\covsa_N) + \lambda_{\min}(\cov) } \, \| \covsa_N - \cov \|_F \\
    &\leq \| \msa_N - \mu \| + \frac{1}{\lambda_{\min}(\cov)} \, \| \widehat M_N - M + \mu \mu^\top - \msa_N \msa_N^\top \|_F \\
    &\leq \| \msa_N - \mu \| + \frac{1}{\lambda_{\min}(\cov)} \| \msa_N \msa_N^\top - \mu \mu^\top \|_F + \frac{\sqrt{n}}{\lambda_{\min}(\cov)} \, \| \widehat M_N - M \| ,
    \end{align*}
    where the second inequality holds because $\sqrt{a^2+b^2} \leq a + b$ for all $a,b\ge 0$, and the third inequality follows from~\cite[Equation~(1.2)]{ref:schmitt1992perturbation}. The last inequality exploits the triangle inequality and the observation that $ \| A \|_F \leq\sqrt{ n} \| A \| $ for all $A \in \mathbb S^n$. Note that all divisions by~$\lambda_{\min}(\cov)$ are well-defined because $\cov\succ 0$ by assumption. An elementary calculation further shows that
    \begin{align*}
    \| \msa_N \msa_N^\top - \mu \mu^\top \|_F 
    &= \| \msa_N (\msa_N - \mu)^\top + (\msa_N - \mu) \mu^\top \|_F \\
    &\leq \| \msa_N (\msa_N - \mu)^\top \|_F + \| (\msa_N - \mu) \mu^\top \|_F \\
    &\leq  \| \msa_N \| \cdot \| \msa_N - \mu \| + \| \msa_N - \mu \| \cdot \| \mu \| \\
    &= \| \msa_N - \mu \| \, (\| \msa_N \| + \| \mu \|) \\
    &\leq 2 \| \mu \| \cdot \| \msa_N - \mu \| + \| \msa_N - \mu \|^2.
    \end{align*}
    Next, set $K=1 + 2 \| \mu \| / \lambda_{\min}(\cov)$, and introduce an auxiliary function $f_\mu:\R_+\rightarrow\R_+$ defined through $f_\mu(x)= Kx+x^2/ \lambda_{\min}(\cov)$ for all $x\in\R_+$.
    %$f_\mu(x)= K x +x^2/ \lambda_{\min}(\cov)$ for all $x\in\R_+$.
    Similarly, introduce an auxiliary function $f_M:\R_+\rightarrow\R_+$ defined through $f_M(x)= \sqrt{n}x/\lambda_{\min}(\cov)$ for all $x\in\R_+$. The above estimates imply that
    \begin{equation*} 
    \Gelbrich \big( (\msa_N, \covsa_N), (\mu, \cov) \big) 
    \leq  f_\mu\big( \| \msa_N - \mu \| \big) + f_M\big(\| \widehat M_N - M \|\big). 
    \end{equation*}
    For any~$\eta_\mu,\eta_M\in(0,1]$ and for any $\rho\geq f_\mu(\rho_\mu(\eta_\mu))+f_M(\rho_M(\eta_M))%= c_1 \rho_{\mu}(\eta_\mu) + c_2 \rho_{\mu}(\eta_\mu)^2 + c_3 \rho_M(\eta_M)
    $, we thus have
    \begin{align}
    & \phantom{\geq} \PP^N \left[ \Gelbrich \big( (\msa, \covsa), (\mu, \cov) \big) \leq \rho \right] \nonumber \\
    & \geq \PP^N \left[ f_\mu\big( \| \msa_N - \mu \| \big) + f_M\big( \| \widehat M_N - M \|\big) \leq \rho \right] \nonumber \\
    & \geq \PP^N \left[ f_\mu\big( \| \msa_N - \mu \| \big) \leq f_\mu(\rho_\mu(\eta_\mu)) ~\wedge ~f_M\big(\| \widehat M_N - M \|\big) \leq f_M(\rho_M(\eta_M)) \right] \label{eq:Gelbrich-guarantee} \\
    & = \PP^N \left[ \| \msa_N - \mu \| \leq \rho_\mu(\eta_\mu) ~\wedge ~ \| \widehat M_N - M \| \leq \rho_M(\eta_M) \right] \geq 1 - \eta_\mu - \eta_M, \nonumber
    \end{align}
    where the equality holds because~$f_\mu$ and~$f_M$ are strictly increasing on their domains, and the last inequality follows from the reverse union bound. The claim thus follows if we define $\rho(\eta_\mu,\eta_M)= f_\mu(\rho_\mu(\eta_\mu))+f_M(\rho_M(\eta_M))$, which equals $c_1 \rho_{\mu}(\eta_\mu) + c_2 \rho_{\mu}(\eta_\mu)^2 + c_3 \rho_M(\eta_M)$ for $c_1=1+2\|\mu\|/\lambda_{\min}(\cov)$, $c_2=1/\lambda_{\min}(\cov)$ and $c_3=\sqrt{n}/\lambda_{\min}(\cov)$. Note that these constants are indeed strictly positive.
\end{proof}    

The proof of Corollary~\ref{corollary:finite-sample:SAA} requires two preparatory lemmas.
    
\begin{lemma}[Concentration inequality for the sample mean]
    \label{lemma:sample-mean}
    Suppose that $\PP$ is sub-Gaussian with variance proxy $\sigma^2$, mean~$\mu$ and covariance matrix~$\cov\succ 0$, and denote by $\msa_N$ the sample mean corresponding to~$N$ independent points sampled from~$\PP$. Then, there exists a positive constant~$C \leq \sigma / \sqrt{\| \cov \|}$ such that $\PP^N [ \| \msa_N - \mu \| \leq \rho_\mu(\eta)  ] \geq 1 - \eta$ for any significance level~$\eta \in (0,1]$, where
    \begin{equation*}
    \rho_\mu(\eta) = C \left( \sqrt{\frac{\Tr{\cov}}{N}} + \sqrt{ \frac{2 \, \| \cov \| \log \left( 1 / \eta \right)}{N} } \right).
    \end{equation*}
\end{lemma}

Lemma~\ref{lemma:sample-mean} establishes a variant of the Hanson-Wright inequality \cite{ref:hanson1971bound} for sub-Gaussian distributions.

\begin{proof}[Proof of Lemma~\ref{lemma:sample-mean}]
    Set~$\tilde\xi_i=\cov^{-\frac{1}{2}}(\xi_i-\mu)$, which is well-defined because~$\cov\succ 0$, and note that~$\tilde\xi_i$ is isotropic in the sense that~$\EE_\PP[\tilde \xi_i]=0$ and~$\EE_\PP[\tilde \xi_i\tilde \xi_i^\top]=I$ for all~$i=1,\ldots,N$. 
    Then, the probability distribution of~$\tilde\xi_i$ is also sub-Gaussian with the variance proxy $C^2$ satisfying $C^2 \leq \sigma^2 / \| \cov \|$ because
    \begin{align*}
        \EE_{\PP} \Big[ \exp \big( z^\top  ( \tilde \xi_i - \EE_{\PP} [\tilde \xi_i] ) \big) \Big] 
        &= \EE_{\PP} \Big[ \exp \big( z^\top \cov^{-\frac{1}{2}}  ( \xi_i - \EE_{\PP} [\xi_i] ) \big) \Big]\\
        &\leq \exp \left({\textstyle\frac{1}{2}}\| \cov^{-\frac{1}{2}} z \|^2 \sigma^2 \right)
        \leq \exp \left({\textstyle\frac{1}{2}}\| z \|^2 \sigma^2/\| \cov \| \right)
    \end{align*}
    for all~$z \in \R^n$. By the definition of the sample mean, we then find
    \begin{equation*}
    \msa_N - \mu = \frac{1}{N} \sum_{i=1}^N \xi_i - \mu = \cov^{\frac{1}{2}} \left( \frac{1}{N} \sum_{i=1}^{N} \tilde \xi_i \right).
    \end{equation*}
    Note that the random vector $N^{-1} \sum_{i=1}^{N} \tilde \xi_i$ has zero mean and covariance matrix~$N^{-1} I$, and thus it concentrates around~$0$ for large~$N$. 
    In addition, as $\tilde \xi_1, \dots, \tilde\xi_N$ are mutually independent, one easily verifies that $N^{-1} \sum_{i=1}^{N} \tilde \xi_i$ is also sub-Gaussian with variance proxy $C^2 / N$. Therefore, \cite[Theorem~2.1]{ref:hsu2012tail} guarantees that
    \begin{align*}
    &\PP^N \left[ \left\| \cov^{\frac{1}{2}} \left( \frac{1}{N} \sum_{i=1}^{N} \tilde \xi_i \right) \right\|^2 \leq \frac{C^2}{N} \left( \Tr{\cov} + 2 \sqrt{\Tr{\cov^2} \log \left( 1 / \eta \right) } + 2 \| \cov \| \log \left( 1 / \eta \right) \right)  \right] \geq 1 - \eta.
    \end{align*}
    The elementary inequality $\Tr{\cov^2} \leq \| \cov \|   \Tr{\cov}$ for any $\cov\succeq 0$ further implies that
    $$ \Tr{\cov} + 2 \sqrt{\Tr{\cov^2} \log \left( 1 / \eta \right) } + 2 \| \cov \| \log \left( 1 / \eta \right) \leq \left( \sqrt{\Tr{\cov}} + \sqrt{2 \| \cov \| \log \left( 1 / \eta \right)} \right)^2. $$
    Combining this inequality with the above concentration bound yields
    \begin{align*}
    &\PP^N \left[ \| \msa_N - \mu \| \leq \rho_\mu(\eta) \right] \\
    =\; &\PP^N \left[ \left\| \cov^{\frac{1}{2}} \left( \frac{1}{N} \sum_{i=1}^{N} \tilde \xi_i \right) \right\| \leq C \left( \sqrt{\frac{\Tr{\cov}}{N}} + \sqrt{ \frac{2 \| \cov \| \log \left( 1 / \eta \right)}{N} } \right) \right] \\
    \geq\; & \PP^N \left[ \left\| \cov^{\frac{1}{2}} \left( \frac{1}{N} \sum_{i=1}^{N} \tilde \xi_i \right) \right\|^2 \leq \frac{C^2}{N} \left( \Tr{\cov} + 2 \sqrt{\Tr{\cov^2} \log \left( 1 / \eta \right) } + 2 \| \cov \| \log \left( 1 / \eta \right) \right)  \right] \\
    \geq \;& 1-\eta,
    \end{align*}
    and thus the claim follows.
\end{proof}
    
\begin{lemma}[Concentration inequality for the sample second moment matrix]
    \label{lemma:sample-cov}
    Suppose that~$\PP$ is sub-Gaussian with variance proxy~$\sigma^2$ and second moment matrix $M$, and denote by $\wh M_N$ the sample second moment matrix corresponding to~$N$ independent points sampled from~$\PP$. Then, there are universal  constants $C_1>0$, $C_2\geq 1$, 
    $C_3>0$ such that $\PP^{N} [ \| \widehat M_N - M \| \leq \rho_M(\eta)] \geq 1 - \eta$ for any significance level~$\eta \in (0,1]$, where
    \begin{equation*}
    \rho_M(\eta) = \sigma^2 C_1 \left( \sqrt{\frac{n}{N}} + \frac{n}{N} \right) + \sigma^2\left( \sqrt{\frac{\log(C_2/\eta)}{C_3 N}} + \frac{\log(C_2/\eta)}{C_3 N} \right).
    \end{equation*}
\end{lemma}

\begin{proof}[Proof of Lemma~\ref{lemma:sample-cov}]
    By~\cite[Theorem~6.5]{ref:wainwright2019high}, for any $\delta > 0$ there exist universal constants $C_1>0$, $C_2\ge 1$, $C_3 > 0$ with
    \begin{equation*}
    \PP^{N} \left[ \frac{\| \widehat M_N - M \|}{\sigma^2} > C_1 \left( \sqrt{\frac{n}{N}} + \frac{n}{N} \right) + \delta \right] \leq C_2 \exp \left( -C_3 N \min\{ \delta, \delta^2 \} \right),
    \end{equation*}
    which implies that
    \begin{align*}
        \PP^{N} \!\! \left[ \| \widehat M_N - M \| \leq \sigma^2C_1 \left( \sqrt{\frac{n}{N}} + \frac{n}{N} \right) + \sigma^2\max \left\{ \sqrt{\frac{\log(C_2/\eta)}{C_3 N}} , \frac{\log(C_2/\eta)}{C_3 N} \right\} \right] \! \geq \! 1 - \eta.
    \end{align*}
    The claim then follows from the inequality $\max\{a, b\} \leq a + b $ for all $a,b \ge~0$.
\end{proof}

We are now armed to prove Corollary~\ref{corollary:finite-sample:SAA}.
\begin{proof}[Proof of Corollary~\ref{corollary:finite-sample:SAA}]
    Define $f_\mu$ and $f_M$ as in the proof of Theorem~\ref{theorem:finite-sample}. Lemma~\ref{lemma:sample-mean} then implies that
    \begin{align*}
    &f_\mu(\rho_\mu(\eta_\mu)) 
    = C K \left( \sqrt{\frac{\Tr{\cov}}{N}} + \sqrt{ \frac{2 \| \cov \| \log \left( 1 / \eta_\mu \right)}{N} } \right) + C^2 \left( \sqrt{\frac{\Tr{\cov}}{\lambda_{\min}(\cov) N}} + \sqrt{ \frac{2 \| \cov \| \log \left( 1 / \eta_\mu \right)}{\lambda_{\min}(\cov) N} } \right)^2 \\
    &\leq \frac{CK \sqrt{\Tr{\cov}} \lambda_{\min}(\cov) + CK \lambda_{\min}(\cov) + 2 C^2 \Tr{\cov} }{\lambda_{\min}(\cov) \sqrt{N}} + \frac{\left( 2 C K \| \cov \| \lambda_{\min}(\cov) + 4 C^2 \| \cov \| \right) \log \left( 1 / \eta_\mu \right)}{\lambda_{\min}(\cov) \sqrt{N}},
    \end{align*}
    where the inequality holds because $(a + b)^2 \leq 2 a^2 + 2 b^2$ for all $a,b\geq 0$, while $1/N\leq 1/\sqrt{N}$ for all $N \in \mathbb N$ and $\sqrt{x} \leq 1 + x$ for all $x \geq 0$.
    Similarly, Lemma~\ref{lemma:sample-cov} implies that
    \begin{align*}
    f_M(\rho_M(\eta_M)) 
    &= \frac{C_1 \sigma^2 \sqrt{n}}{\lambda_{\min}(\cov)} \left( \sqrt{\frac{n}{N}} + \frac{n}{N} \right) + \frac{\sigma^2 \sqrt{n}}{\lambda_{\min}} \left( \sqrt{\frac{\log(C_2/\eta_M)}{C_3 N}} + \frac{\log(C_2/\eta_M)}{C_3 N} \right) \\
    &\leq \frac{n \sigma^2 C_1 C_3 (1 + \sqrt{n}) + \sigma^2 \sqrt{C_3 n}}{\lambda_{\min}(\cov)C_3 \sqrt{N}}  
    + \frac{ \sigma^2 \sqrt{n} ( \sqrt{C_3} + 1) \log(C_2/\eta_M)}{\lambda_{\min}(\cov) C_3 \sqrt{N}},
    \end{align*}
    where the inequality holds again because $1/N\leq 1/\sqrt{N}$ for all $N \in \mathbb N$ and $\sqrt{x} \leq 1 + x$ for all $x \geq 0$. %, while from the proof of \cite[Theorem~6.5]{ref:wainwright2019high} we know that $C_2 \geq 1$.
    Setting $\eta_\mu = \eta_M = \eta / 2$, the sum of the above upper bounds on $f_\mu(\rho_\mu(\eta_\mu))$ and $f_M(\rho_M(\eta_M))$ equals $\rho(\eta) = (c_1 + c_2 \log(1/\eta)) / \sqrt{N}$ for some positive constants~$c_1$ and $c_2$ that depend only on $\mu$, $\cov$, $\sigma^2$ and $n$. The claim then follows from Theorem~\ref{theorem:finite-sample}. 
\end{proof}    

\begin{proof}[Proof of Proposition~\ref{prop:gelbrich-projection}]
The inclusion~\eqref{eq:Gelbrich-inclusion} follows immediately from the Gelbrich bound of Theorem~\ref{theorem:gelbrich}. 
It thus suffices to prove the reverse inclusion for~$\covsa \succ 0$. To this end, select any $(\m, \cov) \in \U_\rho(\msa, \covsa)$, and construct the pushforward distribution $\QQ=\Pnom \circ f^{-1}$ using the affine function $f(\xi) = \covsa^{-\half} (\covsa^\half \cov \covsa^\half )^\half \covsa^{-\half} (\xi -\msa) + \m$. As the structural ambiguity set~$\mc S$ is closed under positive semidefinite affine pushforwards, we have $\QQ \in \mc S$.
In addition, Theorem~\ref{theorem:Wasserstein=Gelbrich} implies that $\Wass(\QQ, \Pnom) = \Gelbrich\big((\m, \cov), (\msa, \covsa) \big) \leq \rho$, where the inequality holds because $(\m, \cov) \in \U_\rho(\msa, \covsa)$. We may thus conclude that $\QQ \in \mc W_{\rho}(\Pnom)$. Finally, an elementary calculation reveals that $\QQ$ has mean $\m$ and covariance matrix $\cov$.
\end{proof}

\begin{proof}[Proof of Theorem~\ref{theorem:gelbrich:amb}]
    The claim trivially holds if $\mc W_\rho(\Pnom)$ is empty. From now on we thus assume that~$\mc W_\rho(\Pnom)$ is non-empty. For any distribution~$\QQ \in \mc W_{\rho}(\Pnom)$ with mean~$\m \in \R^n$ and covariance matrix~$\cov \in \PSD^n$, Theorem~\ref{theorem:gelbrich} then implies that
    \[
        \Gelbrich\big( (\m, \cov), (\msa, \covsa) \big) \le \Wass(\QQ, \Pnom) \le \rho.
    \]
    As $\QQ\in \mc S$, we thus have $\QQ \in \mc G_{\rho}(\msa, \covsa)$. Hence, we may conclude that $\mc W_{\rho}(\Pnom) \subseteq \mc G_{\rho}(\msa, \covsa)$.

    Assume now that the structural ambiguity set~$\mc S$ is generated by~$\Pnom$ and that~$\covsa\succ 0$. Next, select any distribution~$\QQ \in \mc G_{\rho}(\msa, \covsa)$ with mean~$\m \in \R^n$ and covariance matrix~$\cov \in \PSD^n$. As~$\QQ \in \mc S$ constitutes a positive semidefinite affine pushforward of~$\Pnom$, we thus have
    \[
         \Wass(\QQ, \Pnom) = \Gelbrich\big( (\m, \cov), (\msa, \covsa) \big) \le \rho,
    \]
    where the equality follows from Theorem~\ref{theorem:Wasserstein=Gelbrich}. This implies that $\QQ \in \mc W_{\rho}(\Pnom)$ and, as~$\QQ \in \mc G_{\rho}(\msa, \covsa)$ was chosen arbitrarily, that $\mc{G}_{\rho}(\msa, \covsa)\subseteq \mc W_{\rho}(\Pnom)$. Recalling the first part of the proof, we then obtain~$\mc W_{\rho}(\Pnom) = \mc G_{\rho}(\msa, \covsa)$.
\end{proof}

\begin{proof}[Proof of Proposition~\ref{prop:compact:U}]
    The non-negativity of the Gelbrich distance implies that
    \[
    \U_{\rho}(\msa, \covsa) = 
    \left\{ (\m,\cov) \in \R^n\times \PSD^n: \Gelbrich^2 \big((\m, \cov), (\msa, \covsa) \big) \leq \rho^2 \right\}.
    \] 
    Recall that the squared Gelbrich distance is convex in $(\m, \cov) \in \R^n \times \PSD^n$. The H\"{o}lder continuity of the matrix square root established in~\cite[Lemma~A.2]{ref:nguyen2019bridging} ensures that the squared Gelbrich distance is also continuous. Therefore, $\U_{\rho}(\msa, \covsa)$ is convex and closed. Finally, one can show that for any $(\m, \cov) \in \U_{\rho}(\msa, \covsa)$ we have $\| \m - \msa \| \leq \rho$ and $0 \preceq \cov \preceq \big( \rho + \Trace \big[\covsa \big]^{\half} \big)^2 I$, see~\cite[Lemma~A.6]{ref:shafieezadeh2018wasserstein}. This implies that $\U_{\rho}(\msa, \covsa)$ is also compact.
\end{proof}

\begin{proof}[Proof of Proposition~\ref{prop:compact:V}]
    Recall from Proposition~\ref{prop:compact:U} that~$\mathcal U_{\rho}(\msa, \covsa)$ is compact, and note that~$\mathcal V_{\rho}(\msa, \covsa)$ can be viewed as the image of~$\mathcal U_{\rho}(\msa, \covsa)$ under the continuous transformation~$f(\mu,\cov)=(\mu, \cov+\mu\mu^\top)$ defined on~$\R^n\times\PSD^n$. Thus, the transformed set~$\mathcal V_{\rho}(\msa, \covsa)$ inherits compactness from~$\mathcal U_{\rho}(\msa, \covsa)$. To show that~$\mathcal V_{\rho}(\msa, \covsa)$ is convex, recall from the proof of Theorem~\ref{theorem:gelbrich} that the squared Gelbrich distance satisfies 
    \begin{align*}
    \Gelbrich^2\big( (\mu, \Sigma), (\msa, \covsa) \big) &= \left\{ 
    \begin{array}{cl}
    \DS \inf & \DS \| \mu - \msa \|^2 + \Tr{\Sigma + \widehat \Sigma - 2 C} \\ [1ex]
    \st & \DS C \in \R^{n \times n}, \;  \begin{bmatrix} \Sigma & C \\ C^\top & \widehat \Sigma \end{bmatrix} \succeq 0
    \end{array} \right. \\
     &= \left\{ \begin{array}{cl}
    \DS \sup & \| \mu - \msa \|^2 + \DS \Tr{\cov(I - A_{11}) + \covsa(I - A_{22})} \\ [1ex]
    \st & \DS A_{11} \in \PSD^{n}, \; A_{22} \in \PSD^{n}, \; \begin{bmatrix} A_{11} & -I \\ -I & A_{22} \end{bmatrix} \succeq 0.
    \end{array}\right.
    \end{align*}
    Here, the second equality follows from strong semidefinite duality 
    \cite[Theorem~1.4.2]{ref:ben2001lectures}, which holds because $A_{11} = A_{22} = 2 I$ constitutes a Slater point for the dual problem. 
    Thus, we have
    \begin{align*}
    &\phantom{=} \Gelbrich^2\big( (\mu, M - \mu \mu^\top), (\msa, \covsa) \big) \\
    &= \left\{ 
    \begin{array}{cl}
    \DS \sup & \DS \Tr{M(I - A_{22})} + \mu^\top A_{22} \mu - 2 \mu^\top \msa + \Tr{\covsa(I - A_{22})} + \| \msa \|^2 \\ [1ex]
    \st & \DS A_{11} \in \PSD^{n}, \; A_{22} \in \PSD^{n}, \; \begin{bmatrix} A_{11} & -I \\ -I & A_{22} \end{bmatrix} \succeq 0
    \end{array} \right.
    \end{align*}
    for any~$\mu\in\R^n$ and~$M\in\PSD^n$ with $M\succeq \mu\mu^\top$. Note that the objective function of the above maximization problem is
    jointly convex in $\m$ and $M$ for any feasible~$A_{11}$ and~$A_{22}$. As convexity is preserved under maximization, we may thus conclude that~$\Gelbrich^2( (\mu, M - \mu \mu^\top), (\msa, \covsa))$ is also jointly convex in~$\m$ and~$M$. Hence, the set
    \[
        \mathcal V_{\rho}(\msa, \covsa) = 
        \left\{ (\m,M) \in \R^n\times \PSD^n:  M\succeq\m\m^\top,~ \Gelbrich^2\big( (\mu, M - \mu \mu^\top), (\msa, \covsa)\big)\leq \rho^2 \right\}
    \]
    is convex because it is representable as the feasible set of convex constraints.
\end{proof}

\section{Proofs of Section~\ref{sec:linearportfolio}}
    
\begin{proof}[Proof of Proposition~\ref{proposition:alpha}]
Define $\mc F$ as the set of all cumulative distribution functions on~$\R$. Thus, $\mc F$ contains all non-decreasing and right-continuous functions $F : \R \rightarrow [0, 1]$ with $\lim_{t \downarrow \infty} F (t ) = 0$ and $\lim_{t \uparrow \infty} F (t ) =  1$. For any loss function~$\ell\in\Lc_0$ and probability distribution~$\QQ\in\mc M$, we use~$F_{\ell(\xi)}^\QQ\in\mc F$ to denote the distribution function of the random variable~$\ell(\xi)$ under~$\QQ$. As the family of risk measures $\{\risk_\QQ\}_{\QQ \in \mc M}$ is law-invariant, there exists a distribution functional $\varrho: \mc F \to \R$ with~$\risk_{\QQ}(\ell) = \varrho(F_{\ell(\xi)}^{\QQ})$ for all~$\ell\in\Lc_0$ and~$\QQ\in\mc M$.

Define now $\mc F(0,1)\subseteq\mc F$ as the family of all cumulative distribution functions of the random variables of the form~$v^\top \xi$ for some~$v\in\R^n$ under any probability distribution~$\mathbb D\in\mc S$ under which~$v^\top\xi$ has zero mean and unit variance, that is, we set
\[
    \mc F(0,1) = \left\{ F^\mathbb D_{v^\top\xi} : v\in\R^n,~ \mathbb D\in\mc S \text{ such that }\EE_\mathbb D[v^\top\xi] =0 \text{ and } \EE_\mathbb D[(v^\top\xi)^2] =1 \right\}.
\]
Next, choose any $\m \in \R^n$, $\cov \in \PSD^n$ and $\w \in \R^n$ with~$\w^\top \cov \w > 0$, and set~$m = \w^\top \m$ and~$s = \sqrt{\w^\top \cov \w}$. Recalling that the structured Chebyshev ambiguity set~$\mc C(\mu,\cov)$ contains all distributions~$\QQ\in\mc S$ with mean~$\mu$ and covariance matrix~$\cov$, we will first show that
    \begin{equation}
        \label{eq:cdf-families}
        \mc F(0, 1) = \left\{ F_{-(w^\top \xi-m)/s}^\QQ : \QQ \in \mathcal C(\m, \cov) \right\}.
    \end{equation}
The proof proceeds in two steps. In the first step, we prove that the left hand side of~\eqref{eq:cdf-families} is a subset of the right hand side. To this end, select any~$F \in \mathcal F(0,1)$. Thus, there exists a random vector~$\xi_1\in\R^n$ with probability distribution~$\mathbb D_1\in\mc S$ and a deterministic vector~$v\in\R^n$ such that~$v^\top\xi_1$ has zero mean and unit variance under~$\mathbb D_1$ and such that~$F=F^{\mathbb D_1}_{v^\top \xi_1}$. Using similar ideas as in~\cite[Theorems~1 and~2]{ref:yu2009projection}, we construct another random vector $\xi_2\in\R^n$ with distribution~$\mathbb D_2\in\mathcal C(0, I)$ that is independent of~$\xi_1$. We can then construct a new random vector~$\xi\in\R^n$ as the following linear combination of~$\xi_1$ and~$\xi_2$.
    \begin{align*}
        \textstyle \xi = \mu- \frac{1}{s} \cov wv^\top\xi_1+ \left(I - \frac{1}{s^2}\cov \w \w^\top \right) \cov^{\frac{1}{2}} \xi_2
    \end{align*} 
    Next, we define~$\QQ$ as the probability distribution of~$\xi$. As~$\mathbb D_1$ and~$\mathbb D_2$ belong to the structural ambiguity set~$\mathcal S$ and as~$\mathcal S$ is stable and therefore closed under (not necessarily positive semidefinite) affine pushforwards and under convolutions, we may conclude that~$\QQ\in\mathcal S$. In addition, an elementary calculation exploiting our knowledge that~$v^\top\xi_1$ and~$\xi_2$ are standardized under the distributions~$\mathbb D_1$ and~$\mathbb D_2$, respectively, reveals that~$\xi$ has mean~$\mu$ and covariance matrix~$\cov$ under~$\QQ$. Hence, we have shown that~$\QQ\in\mathcal C(\mu,\cov)$. By the construction of~$\xi$ and the definitions of~$m$ and~$s$, we finally have~$-(w^\top\xi-m)/s=v^\top\xi_1$ and thus~$F_{-(w^\top \xi-m)/s}^\QQ=F$. This implies that~$F$ belongs to the set on the right hand side of~\eqref{eq:cdf-families}. %, hence completing the proof of the first step.
    
    In the second step, we prove that the right hand side of~\eqref{eq:cdf-families} is a subset of the left hand side. To this end, select any $\QQ \in \mathcal C(\m, \cov)$, and use~$F$ as a shorthand for $F^\QQ_{- (w^\top\xi-m)/s}$. To show that~$F\in\mc F(0,1)$, set~$v=-w/s$, and define the pushforward distribution~$\mathbb D=\QQ\circ f^{-1}$ with respect to the transformation~\mbox{$f(\xi)=\xi+mv/(s\|v\|^2)$}. Note first that~$\mathbb D\in\mc S$ because the structural ambiguity set~$\mc S$ is closed under affine pushforwards. In addition, one readily verifies that~$- (w^\top\xi-m)/s=v^\top f(\xi)$ has zero mean and unit variance under~$\QQ$, which means that~$v^\top\xi$ has zero mean and unit variance under~$\mathbb D$. We thus have~$F\in\mc F(0,1)$. In combination, the first and the second step of the proof establish equation~\eqref{eq:cdf-families}.

    We may now conclude that the standard risk coefficient satisfies
    \begin{align*}
        \alpha = \Sup{\QQ \in \mc C(\m,\cov)}  \; \risk_\QQ \left(- \frac{1}{s}(w^\top\xi-m) \right) = \Sup{\QQ \in \mc C(\m,\cov)} \; \varrho \left( F^\QQ_{- (w^\top\xi-m)/s} \right) =  \Sup{F \in \mc F (0, 1)} \; \varrho ( F),
    \end{align*}
    where the second equality re-expresses the risk measures~$\risk_\QQ$ for $\QQ\in\mc M$ in terms of the distribution functional~$\varrho$, and the second equality exploits~\eqref{eq:cdf-families}. The last expression is manifestly independent of $\m$, $\cov$ and $\w$. This observation completes the proof.
\end{proof}

\begin{proof}[Proof of Theorem~\ref{theorem:meanstd}]
    Decomposing the Gelbrich ambiguity set into disjoint Chebyshev ambiguity sets allows us to rewrite the Gelbrich risk as in~\eqref{eq:two-layer-a}. As the loss function~$\ell(\xi)=-w^\top\xi$ is linear, the inner maximization problem in~\eqref{eq:two-layer-a} further simplifies to
    \begin{align*}
        &\Sup{\QQ \in \mathcal C(\m,\cov)}  \; \risk_\QQ \left( - \w^\top \xi \right) =  - \w^\top \m + \Sup{\QQ \in \mathcal C(\m,\cov)}  \; \risk_\QQ \left( - \w^\top (\xi-\m) \right) \\
        &\qquad= - \w^\top \m + \sqrt{\w^\top \cov \w} ~\Sup{\QQ \in \mathcal C(\m,\cov)}  \; \risk_\QQ \left( - \frac{\w^\top (\xi-\m)}{\sqrt{\w^\top \cov \w}} \right) = - \w^\top \m +\alpha \sqrt{\w^\top \cov \w}
    \end{align*}
    where the first two equalities exploit the translation invariance and the positive homogeneity of the risk measures~$\risk_\QQ$, $\QQ \in \mc M$, respectively, whereas the third equality follows from the definition of the standard risk coefficient~$\alpha$.
    Note that~$\alpha$ is independent of~$\mu$, $\cov$ and~$w$ thanks to Proposition~\ref{proposition:alpha}, which applies because the structural ambiguity set is stable and the family of risk measures is law-invariant. Using the definition of the set~$\U_{\rho}(\msa, \covsa)$, the outer maximization problem in~\eqref{eq:two-layer-a} can then be reformulated as
    \begin{align} 
    \Sup{\QQ \in \mc{G}_{\rho}(\msa, \covsa)} \; \risk_\QQ ( -\w^\top \xi ) 
    &= \Sup{(\m, \cov) \in \U_{\rho}(\msa, \covsa)} \; -\m^\top \w + \alpha \sqrt{\w^\top \cov \w} \notag \\
    &= 
    \left\{
    \begin{array}{cl}
    \Sup{\m, \cov \succeq 0} & -\m^{\top} \w + \alpha \sqrt{\w^{\top} \cov \w} \\
    \st & \| \m - \msa \|^2 + \Tr{\cov + \covsa - 2 \big( \covsa^\half \cov \covsa^\half \big)^\half} \leq \rho^2.
    \end{array}
    \right. \label{eq:primal}
    \end{align}
    If~$\rho=0$, then~$(\m,\cov)=(\msa, \covsa)$ is the only feasible solution of problem~\eqref{eq:primal}, in which case~\eqref{eq:MS:1} trivially holds. Similarly, if~$\alpha=0$, then~$\cov=\covsa$ and~$\m=\msa-\rho w/\|w\|$ are optimal in~\eqref{eq:primal}, and~\eqref{eq:MS:1} again trivially holds. From now on, we may thus assume without loss of generality that~$\rho > 0$ and~$\alpha>0$. In this case problem~\eqref{eq:primal} is equivalent to
    \begin{align*}
    &\Sup{\m, \cov \succeq 0} \Inf{\gamma \geq 0} \quad -\m^{\top} \w + \alpha \sqrt{\w^\top \cov \w} + \gamma \left[ \rho^2 - \|\m - \msa\|^2 - \Tr{\cov + \covsa - 2 \big( \covsa^\half \cov \covsa^\half \big)^\half} \right] \\
    =& \Inf{\gamma \geq 0} \Bigg\{ \gamma \big( \rho^2 - \Tr{\covsa} \big) + \Sup{\m} \left\{ -\m^{\top} \w - \gamma \|\m - \msa \|^2  \right\}+ \Sup{\cov \succeq 0} \left\{ \alpha \sqrt{\w^\top \cov \w} + \gamma \Tr{-\cov + 2 \big( \covsa^\half \cov \covsa^\half \big)^\half }  \right\} \Bigg\},
    \end{align*}
    where the first equality follows from strong duality, which holds because $(\msa, \covsa)$ constitutes a Slater point for the primal convex program~\eqref{eq:primal}, and from a simple rearrangement. As~$\gamma \geq 0$, the embedded quadratic maximization problem over~$\m$ is convex and can be solved analytically. Indeed, for~$\gamma>0$ we have $\sup_{\m} \left\{ -\m^\top \w - \gamma \| \m - \msa \|^2 \right\} = - \msa^{\top}\w + \frac{\|\w\|^2}{4\gamma}$. For $\gamma = 0$, on the other hand, the supremum over~$\mu$ evaluates to~$0$ if~$w=0$ and to~$+\infty$ otherwise. Thus, the formula on the right-hand side of the above expression remains valid for~$\gamma=0$ if we interpret it as the limit when $\gamma$ tends to~0 from above. Similarly, as~$\gamma\ge 0$ and~$\alpha> 0$, one may introduce an auxiliarly epigraphical variable~$t$ to reformulate the embedded maximization problem over $\cov$ as the convex program
    \[
        \begin{array}{cl}
            \Sup{t, \cov} & ~\alpha t + \gamma  \Tr{-\cov + 2 \big( \covsa^\half \cov \covsa^\half \big)^\half}  \\
            \st &~ t \geq 0, \; \cov \succeq 0, \; t^2 - \w^\top \cov \w \leq 0,
        \end{array}
    \]
    which manifestly satisfies Slater's condition.    Suppose now temporarily that~$\covsa \succ 0$. By invoking strong duality and using the variable transformation~$B \gets (\covsa^{\frac{1}{2}} \cov \covsa^{\frac{1}{2}})^{\frac{1}{2}}$, the above optimization problem can be recast as
    \begin{align}
       % & \Sup{t \ge 0, \cov \succeq 0} \Inf{\lambda \ge 0} ~\alpha t + \gamma  \Tr{-\cov + 2 \big( \covsa^\half \cov \covsa^\half \big)^\half} + \lambda (\w^\top \cov \w - t^2) \notag \\
        & \Inf{\lambda \ge 0} \Sup{t \ge 0, \cov \succeq 0}~ \alpha t - \lambda t^2 + \Tr{\cov (\lambda \w\w^\top - \gamma I)} + 2\gamma \Tr{\big( \covsa^\half \cov \covsa^\half \big)^\half} \notag \\
        =& \Inf{\lambda \ge 0} \Sup{t \geq 0, B \succeq 0}~
        \alpha t - \lambda t^2  + \Tr{B^2 \Delta_\lambda} + 2 \gamma \Tr{B} , \label{eq:mstd:2}
    \end{align}
    where~$\Delta_\lambda = \covsa^{-\half}(\lambda \w\w^{\top} - \gamma I)\covsa^{-\half}$ for any~$\lambda\ge 0$. Note that~$\Delta_\lambda$ is well-defined because~$\covsa$ is invertible. Note also that the inner maximization problem in~\eqref{eq:mstd:2} is separable with respect to~$t$ and~$B$. Consider first the maximization problem over~$t$. As~$\alpha>0$, its supremum evaluates to~$\alpha^2/(4\lambda)$ and is attained at~$t^\star=\alpha/(2\lambda)$ whenever~$\lambda>0$. Otherwise, if~$\lambda=0$, then its supremum evaluates to~$+\infty$ for. From now on we may thus assume without loss of generality that the outer minimization over~$\lambda$ in~\eqref{eq:mstd:2} is subject to the strict constraint~$\lambda>0$. Consider now the maximization problem over~$B$. The proof of~\cite[Proposition~2.8]{ref:nguyen2018distributionally} implies that if~$\Delta_\lambda\not\prec 0$, then the supremum over~$B$ evaluates to~$+\infty$. From now on we may thus assume without loss of generality that the outer minimization over~$\lambda$ in~\eqref{eq:mstd:2} is subject to the strict constraint $\gamma I - \lambda \w\w^{\top} \succ 0$, which is equivalent to~$\lambda<\gamma\|w\|^{-2}$ and guarantees that~$\Delta_\lambda\prec 0$. As~$\lambda>0$, this in turn implies that~$B^\star=-\gamma\Delta_\lambda^{-1}$ is strictly positive definite and satisfies the first-order optimality condition $B \Delta_\lambda + \Delta_\lambda B + 2\gamma I =0$.     
%     and that the respective maximizers~$t^\star$ and~$B^\star$ must satisfy the first-order conditions
%     \be
%         \alpha - 2\lambda t^\star =0\quad \text{and}\quad  B^\star \Delta_\lambda + \Delta_\lambda B^\star + 2\gamma I =0. \notag
%     \ee
%    The first condition requires that $\lambda > 0$, and in this case it implies that $t^\star = \alpha/(2\lambda) > 0$. A necessary condition for the optimal value of $B$ is that $0 \succ \Delta$, which in turns requires $\gamma I - \lambda \w\w^{\top} \succ 0$, or equivalently $\lambda < \gamma \| \w\|^2$. 
This condition can be interpreted as a continuous Lyapunov equation, and therefore its solution~$B^\star$ is in fact unique; see, {\em e.g.},~\cite[Theorem~12.5]{ref:hespanha2009linear}. By making the implicit constraints on~$\lambda$ explicit and by eliminating the supremum operator by evaluating the objective function at~$t^\star$ and~$B^\star$, problem~\eqref{eq:mstd:2} can now be reformulated~as
    \begin{align*}
    & \Inf{0 < \lambda < \gamma \|\w\|^{-2}} \;
    \frac{\alpha^2}{4\lambda} + \gamma^2 \Tr{\covsa^\half (\gamma I - \lambda \w\w^\top)^{-1} \covsa^{\half} }
    \\ &= \Inf{0 < \lambda < \gamma \|\w\|^{-2}} \;  \frac{\alpha^2}{4\lambda} + \gamma \Tr{\covsa} + \frac{\w^\top \covsa \w}{\lambda^{-1} -\|w\|^2/\gamma}=
    \gamma \Tr{\covsa} + \frac{\alpha^2}{4} \frac{\|w\|^2}{\gamma} + \alpha \sqrt{\w^\top \covsa \w}.        
    \end{align*}
    Here, the first equality exploits the Sherman-Morrison formula \cite[Corollary~2.8.8]{ref:bernstein2009matrix} to rewrite the inverse matrix, and the second equality solves the resulting minimization problem analytically. Indeed, the infimum is attained in the interior of the feasible set at the unique solution~$\lambda\opt$ of the first-order condition $\frac{1}{\lambda}  = \frac{\|w\|^2}{\gamma} + \frac{2}{\alpha} \sqrt{\w^\top \covsa \w}$. In summary, we have derived closed-form solutions for both subproblems over~$\mu$ and~$\cov$ in the objective function of the problem dual to~\eqref{eq:primal}. Hence, the Gelbrich risk satisfies
    \begin{align*}
        \Sup{\QQ \in \mc{G}_{\rho}(\msa, \covsa)} \; \risk_\QQ ( -\w^\top \xi )  &=
        \Inf{\gamma \ge 0} ~ -\msa^\top \w + \alpha \sqrt{\w^\top \covsa \w} +  \gamma \rho^2 + \frac{1 + \alpha^2}{4} \frac{\|w\|^2}{\gamma} \\
        &=-\msa^\top \w + \alpha \sqrt{\w^\top \covsa \w } + \rho \sqrt{1+ \alpha^2}\, \| \w \|,
    \end{align*}
    where the second equality holds because the unique solution of the minimization problem in the first line is given by~$\gamma\opt = (2\rho)^{-1} \sqrt{1 + \alpha^2}\, \|w\|$. We have thus established~\eqref{eq:MS:1} for~$\covsa \succ 0$.
    
    To demonstrate that~\eqref{eq:MS:1} remains valid for all $\covsa \succeq 0$, we denote by~$J(\covsa)$ the optimal value of problem~\eqref{eq:primal} as a function of~$\covsa$. Thus, $J(\covsa)$ coincides with the Gelbrich risk.
%     \[
%         J(\covsa) =     \left\{
%         \begin{array}{cl}
%         \Sup{\m, \cov \succeq 0} & -\m^{\top} \w + \alpha \sqrt{\w^{\top} \cov \w} \\
%         \st & \| \m - \msa \|^2 + \Tr{\cov + \covsa - 2 \big( \covsa^\half \cov \covsa^\half \big)^\half} \leq \rho^2.
%         \end{array}
%         \right.
%     \]
    Applying Berge's maximum theorem~\cite[pp.~115--116]{ref:berge1963topological} to problem~\eqref{eq:primal}, it is easy to show that~$J(\covsa)$ is continuous on~$\PSD^n$. Next, define $\bar{J}(\covsa) = -\msa^\top \w + \alpha \sqrt{\w^\top \covsa \w } + \rho \sqrt{1+ \alpha^2}\, \| \w \|$ as the right hand side of~\eqref{eq:MS:1}, which is manifestly continuous on~$\PSD^n$. From the first part of the proof we know that~$J(\covsa) = \bar{J}(\covsa)$ for all $\covsa \succ 0$. As both~$J(\covsa)$ and~$\bar{J}(\covsa)$ are continuous on~$\PSD^n$ and as any positive semidefinite matrix can be expressed as a limit of positive definite matrices, we thus have~$J(\covsa) = \bar{J}(\covsa)$ for all~$\covsa \in \PSD^n$. This proves~\eqref{eq:MS:1} for all~$\covsa \succeq 0$.
    
    Finally, if~$\mc S$ is the structural ambiguity set generated by a Gaussian nominal distribution~$\Pnom$ with~$\covsa \succ 0$, then the Wasserstein risk equals the Gelbrich risk by Corollary~\ref{cor:gelbrich-risk}.
\end{proof}

\begin{proof}[Proof of Proposition~\ref{prop:alpha:positive}]
    Fix any~$\mu\in\R^n$, $\cov\in\PSD^n$ and~$w\in\R^n$, and select any symmetric probability distribution~$\QQ\in\mc S$, which exists by assumption. As~$\mc S$ is closed under positive semidefinite affine pushforwards, we may assume without loss of generality that~$\xi$ has mean~$\mu$ and covariance matrix~$\cov$ under~$\QQ$. We then find
    \begin{align*}
        0 
        & = \textstyle \risk_\QQ(0) 
        = \risk_\QQ\left( \frac{1}{2} \frac{\w^\top (\xi- \m)}{\sqrt{\w^\top \cov \w}} - \frac{1}{2} \frac{\w^\top (\xi- \m)}{\sqrt{\w^\top \cov \w}}\right ) \\
        & \leq \frac{1}{2} \risk_\QQ\left(\frac{\w^\top (\xi- \m)}{\sqrt{\w^\top \cov \w}} \right) + \frac{1}{2} \risk_\QQ\left( - \frac{\w^\top (\xi- \m)}{\sqrt{\w^\top \cov \w}} \right) = \risk_\QQ \left( - \frac{\w^\top (\xi- \m)}{\sqrt{\w^\top \cov \w}} \right) \leq\alpha,
    \end{align*}
    where the first equality and the first inequality follow from the positive homogeneity and convexity of the coherent risk measure~$\risk_\QQ$, respectively. The last equality exploits the law-invariance of~$\risk_\QQ$ and the symmetry of~$\QQ$, which implies that $\w^\top (\xi- \m)/\sqrt{\w^\top \cov \w}$ and $-\w^\top (\xi- \m)/\sqrt{\w^\top \cov \w}$ have the same distribution under~$\QQ$. Finally, the last inequality follows from the definition of~$\alpha$ and the observation that~$\QQ\in \mc C(\mu,\cov)$.
\end{proof}

\begin{proof}[Proof of Corollary~\ref{corol:1/N}]
    Note that problem~\eqref{eq:optimize-w-gelbrich} has a unique minimizer for every~$\rho>0$ because its feasible set~$\Omega$ is closed and its objective function is strictly convex and coercive on~$\Omega$. This minimizer coincides with the unique optimal solution~$\w\opt(\lambda)$ of
    \[
        \Min{\w \in \Omega}~ - \lambda \msa^\top \w + \lambda \alpha \sqrt{\w^\top \covsa \w} + \sqrt{1 + \alpha^2} \| \w\|,
    \]
    where~$\lambda=1/\rho$. By Berge's maximum theorem~\cite[pp.~115--116]{ref:berge1963topological}, the function~$\w\opt(\lambda)$ is continuous on~$[0,1]$, and therefore~$w^\star(\lambda)$ converges to~$w^\star(0)$ as~$\lambda$ tends to~$0$. As~$\frac{1}{n} e\in\Omega$, however, one readily verifies that~$w^\star(0)=\frac{1}{n} e$. As~$\lambda=1/\rho$, this reasoning shows that the unique minimizer of~\eqref{eq:optimize-w-gelbrich} converges to~$\frac{1}{n}e$ as~$\rho$ tends to~$\infty$.
\end{proof}

\begin{proof}[Proof of Proposition~\ref{prop:WC}]
    If $\rho = 0$, then all distributions in the Gelbrich ambiguity set, and in particular all maximizers that attain the Gelbrich risk,  have mean~$\msa$ and covariance matrix~$\covsa$. The claim then follows because, for~$\rho=0$, the formulas for~$\m\opt$ and~$\cov\opt$ reduce to~$\msa$ and $\covsa$, respectively. Assume from now on that~$\rho > 0$. As all assumptions of Theorem~\ref{theorem:meanstd} are satisfied, we may proceed as in the proof of Theorem~\ref{theorem:meanstd} to show that
    \begin{align*}
    \Sup{\QQ \in \mc{G}_{\rho}(\msa, \covsa)} \; \risk_\QQ \left( -\w^\top \xi \right) 
    & = \left\{\begin{array}{cl} 
    \max & -\m^\top \w + \alpha t \\ \text{s.t.} & (\m, \cov) \in \U_{\rho}(\msa, \covsa), ~t\ge 0,~ t^2\leq \w^\top \cov \w
    \end{array}\right. \\
    & = \Min{\gamma \geq 0, \lambda \geq 0} \Sup{\m, \cov \succeq 0, t \geq 0} L(\m, \cov, t, \gamma, \lambda),
    \end{align*}
    where the last equality follows from strong duality, which applies because the primal problem has a compact feasible set. The Lagrangian in the last expression is defined through
    \begin{align*}
        L(\m, \cov, t, \gamma, \lambda) =\;& \gamma \big( \rho^2 - \|\m - \msa \|^2 - \Tr{\covsa} \big) -\m^{\top} \w + \alpha t - \lambda t^2 \\
        & + \Tr{\cov (\lambda \w\w^\top - \gamma I)} + 2\gamma \Tr{\big( \covsa^\half \cov \covsa^\half \big)^\half}.
    \end{align*}
    From the proof of Theorem~\ref{theorem:meanstd} we know that the dual problem is uniquely solved by
    \[
        \gamma\opt = (2\rho)^{-1} \sqrt{1 + \alpha^2}\, \|w\| \quad \text{and} \quad   \frac{1}{\lambda\opt}  = \frac{\|w\|^2}{\gamma \opt} + \frac{2}{\alpha} \sqrt{\w^\top \covsa \w}.
    \]
    In addition, \cite[Theorem~D.4.1]{ref:ben2001lectures} implies that any primal maximizer $(\mu^\star, \cov^\star, t^\star)$ must also be a maximizer of
    \begin{align}
        \label{eq:lagrangian}
        \max_{\m, \cov \succeq 0, t \geq 0}~L(\m, \cov, t, \gamma\opt, \lambda\opt).
    \end{align}
    As the Lagrangian is additively separable in~$\mu$, $\cov$ and~$t$, the maximizers $\mu\opt$, $\cov\opt$ and $t\opt$ can be determined separately. Indeed, one readily verifies that~$\mu\opt$ must be a maximizer of the problem $\max_{\m} \{ -\m^\top \w - \gamma\opt \| \m - \msa \|^2\}$, which is uniquely solved by $\mu^\star = \msa - \w / (2 \gamma^\star)$. Similarly, $\cov\opt$ must be a maximizer of the problem
    \begin{align}
        \label{eq:aux:cov}
        \Max{\cov \succeq 0} \Tr{\cov (\lambda \w\w^\top - \gamma I)} + 2\gamma \Tr{\big( \covsa^\half \cov \covsa^\half \big)^\half} = \Max{B \succeq 0}
        \Tr{B^2 \Delta} + 2 \gamma \Tr{B},
    \end{align}
    where the equality exploits the substitution $B \gets ( \covsa^{\frac{1}{2}} \cov \covsa^{\frac{1}{2}})^{\frac{1}{2}}$ and the definition $\Delta = \covsa^{-\half}(\lambda\opt \w\w^{\top} - \gamma\opt I)\covsa^{-\half}$. As in the proof of Theorem~\ref{theorem:meanstd}, one can show that the second maximization problem in~\eqref{eq:aux:cov} is uniquely solved by $B^\star = \covsa^{\frac{1}{2}} ( I - \frac{\lambda\opt}{\gamma\opt} ww^\top )^{-1} \covsa^{\frac{1}{2}}$, which implies that the first maximization problem in~\eqref{eq:aux:cov} is uniquely solved~by 
    \begin{align*}
        \cov^\star 
        &= \left( I - \frac{\lambda\opt}{\gamma\opt} ww^\top \right)^{-1} \covsa \left( I - \frac{\lambda\opt}{\gamma\opt} ww^\top \right)^{-1} \\
        &= \left(I + \frac{\lambda\opt}{\gamma\opt - \lambda\opt \| \w \|^2} \w \w^\top \right) \covsa \left(I + \frac{\lambda\opt}{\gamma\opt - \lambda\opt \| \w \|^2} \w \w^\top \right).
    \end{align*}
    Here, the last equality exploits the Sherman-Morrison-Woodbury identity~\cite[Corollary~2.8.8]{ref:bernstein2009matrix}.
    Finally, $t\opt$ must be a maximizer of $\max_{t\geq 0} \alpha t-\lambda\opt t^2$, which is uniquely solved by $t\opt = \alpha / (2 \lambda\opt)$. The claim now follows from the formulas for the dual variables~$\gamma\opt$ and $\lambda\opt$ substituted into the formulas for~$\mu\opt$, $\cov\opt$ and~$t\opt$.
\end{proof}

\begin{proof}[Proof of Proposition~\ref{proposition:VaR}]
    Proposition~1 in~\cite{ref:yu2009projection} provides an analytical formula for the Chebyshev risk of a portfolio loss function with respect to VaR. As the Chebyshev risk coincides with the Gelbrich risk for~$\rho=0$, the standard risk coefficient is readily found by comparison with~\eqref{eq:MS:1}. 
\end{proof}

\begin{proof}[Proof of Proposition~\ref{proposition:CVaR}]
    This follows from~\cite[Proposition~2]{ref:yu2009projection}, similar to Proposition~\ref{proposition:VaR}.
\end{proof}

\begin{proof}[Proof of Proposition~\ref{proposition:meanstd}]
    Equation~\eqref{eq:MS:1} for~$\rho=0$ reveals that~$\alpha=\beta$.
\end{proof}

\begin{proof}[Proof of Proposition~\ref{proposition:spectral}]
    This follows from~\cite[Theorem~2]{ref:li2018risk}, similar to Proposition~\ref{proposition:VaR}.
\end{proof}

\begin{proof}[Proof of Proposition~\ref{proposition:coherent}]
    This follows from~\cite[Theorem~3]{ref:li2018risk}, similar to Proposition~\ref{proposition:VaR}.
\end{proof}

\begin{proof}[Proof of Proposition~\ref{proposition:distortion}]
    This follows from~\cite[Theorem~3.10]{ref:cai2020distributionally}, similar to Proposition~\ref{proposition:VaR}. The same theorem reveals that the Chebyshev risk does not change if~$h$ is replaced with its right-continuous modification. In light of~\eqref{eq:two-layer}, this invariance remains valid if uncertainty is modeled by a Gelbrich ambiguity set. Thus, we may always assume without loss of generality that~$h$ is right-continuous.
\end{proof}

\section{Proofs of Section~\ref{sec:nonlinearportfolio}}
\begin{proof}[Proof of Theorem~\ref{theorem:WC-exp}]
    By the two-layer decomposition~\eqref{eq:two-layer-a} of the Gelbrich ambiguity set, we find
    \begin{subequations}
        \begin{align}
            \Sup{\QQ \in \mathbb{G}_{\rho}(\widehat \mu, \covsa)} \EE_\QQ \left[ \ell(\xi) \right] 
            &= \Sup{(\mu, \cov) \in \U_{\rho}(\widehat \mu, \covsa)}\Sup{\QQ \in \Ambi(\mu, \cov)} \EE_\QQ[ \ell(\xi)] \label{eq:exp:1} \\
            &= \Sup{(\mu, \cov) \in \U_{\rho}(\widehat \mu, \covsa)} \Inf{(y_0, y, Y) \in \mathcal Y} \; y_0 + 2 \mu^\top y + \Tr{(\cov + \mu \mu^\top)Y} \label{eq:exp:2} \\
            &= \Sup{(\mu, M) \in \V_{\rho}(\widehat \mu, \covsa)} \Inf{(y_0, y, Y) \in \mathcal Y} \; y_0 + 2 \mu^\top y + \Tr{MY} \label{eq:exp:3} \\
            &= \Inf{(y_0, y, Y) \in \mathcal Y} \Sup{(\mu, M) \in \V_{\rho}(\widehat \mu, \covsa)} \; y_0 + 2 \mu^\top y + \Tr{MY} \label{eq:exp:4} \\
            &= \Inf{(y_0, y, Y) \in \mathcal Y} \; y_0 + \delta^*_{\V_\rho(\widehat \mu, \covsa)} (2y, Y), \label{eq:exp:5}
        \end{align}
    \end{subequations}
    where~\eqref{eq:exp:2} exploits duality. Note that the dual problem in~\eqref{eq:exp:2} can be constructed as in the proof of \cite[Lemma~1]{ref:calafiore2009parameter}. Strong duality holds because $\rho>0$, which implies that the uncertainty set $\U_{\rho}(\widehat \mu, \covsa)$ contains a point $(\mu,\cov)$ with $\cov\succ 0$. Similar arguments were used in~\cite[\S~4.2]{kuhn2025distributionally}, and thus details are omitted. In addition, \eqref{eq:exp:3} holds because~$(\m,\cov)\in \mathcal U_\rho(\msa, \covsa)$ if and only if~$(\m,\cov +\m\m^\top)\in \mathcal V_\rho(\msa, \covsa)$, \eqref{eq:exp:4} follows from Sion's minimax theorem~\cite{ref:sion1958minimax}, which applies because $\V_{\rho}(\widehat \mu, \covsa)$ is convex and compact by virtue of Proposition~\ref{prop:compact:V}, and~\eqref{eq:exp:5} exploits the definition of the support function. By Proposition~\ref{proposition:support:V}, we thus obtain the optimization problem given in the theorem statement. 
\end{proof}

The proof of Theorem~\ref{theorem:poly-var} relies on the following corollary, which generalizes \cite[Lemma~A.2]{ref:zymler2013distributionally}.

\begin{corollary}[Worst-case probabilities over Gelbrich ambiguity sets] \label{corollary:WC-prob}
    If $\Xi \subseteq \mathbb R^n$ as a (not necessarily convex) Borel set, $\mc S = \mc M_2$ and~$\rho>0$, then we have
    \begin{align*}
        \Sup{\QQ \in \mathcal{G}_{\rho}(\widehat \mu, \covsa)} \QQ( \xi \in \Xi) = 
        \left\{
        \begin{array}{cl}
        \inf &  y_0 + \gamma \big( \rho^2 - \| \widehat \mu \|^2 - \Tr{\covsa} \big) + z + \Tr{Z} \\
        \st &\gamma \in \mathbb R_+, \; y_0 \in \mathbb R, \; y \in \mathbb R^n, \; Y \in \mathbb{S}^n, \; z \in \mathbb R_+, \; Z \in \PSD^n \\
        &\begin{bmatrix} \gamma I - Y & \gamma \covsa^\half \\ \gamma \covsa^\half & Z \end{bmatrix} \succeq 0, \; \begin{bmatrix} \gamma I - Y & \gamma \widehat \mu + y \\ (\gamma \widehat \mu + y)^\top & z \end{bmatrix} \succeq 0 , \; \begin{bmatrix} Y & y \\[0.5ex] y^\top  & y_0 \end{bmatrix} \succeq 0 \\[2.5ex]
        & y_0 + 2y^\top \xi + \xi^\top Y \xi \geq 1 \quad \forall \xi \in \Xi.
        \end{array}
        \right.
    \end{align*}
\end{corollary}

\begin{proof}[Proof of Corollary~\ref{corollary:WC-prob}]
    The probability $\QQ( \xi \in \Xi)$ can be viewed as the expected value of the indicator loss function~$\ell =\mathds{1}_{\Xi}$ corresponding to the set~$\Xi$, which is defined through $\mathds{1}_{\Xi}(\xi) =1$ if~$\xi \in \Xi$ and $\mathds{1}_{\Xi}(\xi) =0$ if~$\xi \not\in \Xi$. Thus, we can address worst-case probability problems with the techniques of Theorem~\ref{theorem:WC-exp}. Specifically, the convex set $\mathcal Y$ associated with the indicator loss function can be represented as
    \begin{align*}
        \mathcal Y &= \left\{ y_0 \in \mathbb R, \; y \in \mathbb R^n, \; Y \in \Sym^n: 
        y_0 + 2y^\top \xi + \xi^\top Y \xi \geq \mathds{1}_{\Xi}(\xi) ~ \forall \xi \in \mathbb R^n
        \right\} \\
        &= \left\{ y_0 \in \mathbb R, \; y \in \mathbb R^n, \; Y \in \Sym^n: y_0 + 2y^\top \xi + \xi^\top Y \xi \geq 0 ~ \forall \xi \in \mathbb R^n,~
        y_0 + 2y^\top \xi + \xi^\top Y \xi \geq 1 ~ \forall \xi \in \Xi
        \right\} \\
        &= \left\{ y_0 \in \mathbb R, \; y \in \mathbb R^n, \; Y \in \Sym^n: \begin{bmatrix} Y & y \\[0.5ex] y^\top  & y_0 \end{bmatrix} \succeq 0, \; y_0 + 2y^\top \xi +  \xi^\top Y \xi \geq 1 ~ \forall \xi \in \Xi \right\}.
    \end{align*}
    The claim follows by substituting this representation of~$\mathcal Y$ into the problem of Theorem~\ref{theorem:WC-exp}.
\end{proof}

\begin{proof}[Proof of Theorem~\ref{theorem:poly-var}]
    By introducing an auxiliary variable $\tau$, the Gelbrich VaR can be recast as 
    \begin{align*}
        \Sup{\QQ \in \mathbb{G}_{\rho}(\widehat \mu, \covsa)} \QQ\text{-}\VaR_{\beta}( \ell(\xi)) 
        &= \inf_{\tau \in \R} \bigg\{ \tau : \Sup{\QQ \in \mathbb{G}_{\rho}(\widehat \mu, \covsa)} \QQ\text{-}\VaR_{\beta}( \ell(\xi)) \leq \tau \bigg\}\\
        &= \inf_{\tau \in \R} \bigg\{ \tau : \QQ\text{-}\VaR_{\beta}( \ell(\xi)) \leq \tau ~ \forall \QQ \in \mathbb{G}_{\rho}(\widehat \mu, \covsa) \bigg\} \\
        &= \inf_{\tau \in \R} \bigg\{ \tau : \QQ \left( \ell(\xi) > \tau \right) \leq \beta ~ \forall \QQ \in \mathbb{G}_{\rho}(\widehat \mu, \covsa) \bigg\} \\
        &= \inf_{\tau \in \R} \bigg\{ \tau : \sup_{\QQ \in \mathbb{G}_{\rho}(\widehat \mu, \covsa)} \QQ \left( \ell(\xi) > \tau \right) \leq \beta \bigg\}\\
        &= \inf_{\tau \in \R} \bigg\{ \tau : \sup_{\QQ \in \mathbb{G}_{\rho}(\widehat \mu, \covsa)} \QQ \left( \ell(\xi) \geq \tau \right) \leq \beta \bigg\},
    \end{align*} 
    where the last equality follows from a straightforward adaptation of the arguments in \cite[p.~175]{ref:zymler2013distributionally}, which imply that the mapping $\tau\mapsto \sup_{\QQ \in \mathbb{G}_{\rho}(\widehat \mu, \covsa)} \QQ( \ell(\xi) > \tau )$ is lower semi-continuous and non-increasing. Corollary~\ref{corollary:WC-prob} then allows us to reformulate the worst-case probability on the last line~as
    \begin{align*}
        \Sup{\QQ \in \mathbb G_{\rho}(\widehat \mu, \covsa)} \QQ \left( \ell(\xi) \geq \tau \right) 
        = \left\{
        \begin{array}{cl}
            \inf &  y_0 + \gamma \big( \rho^2 - \| \widehat \mu \|^2 - \Tr{\covsa} \big) + z + \Tr{Z} \\[1ex]
            \st &\gamma,z \in \mathbb R_+, \; y_0 \in \mathbb R, \; y \in \mathbb R^n, \; Y \in \mathbb{S}^n, \; Z \in \PSD^n \\[1ex]
            &\begin{bmatrix} \gamma I - Y & \gamma \covsa^\half \\ \gamma \covsa^\half & Z \end{bmatrix} \succeq 0, \; \begin{bmatrix} \gamma I - Y & \gamma \widehat \mu + y \\ (\gamma \widehat \mu + y)^\top & z \end{bmatrix} \succeq 0, \; \begin{bmatrix} Y & y \\[0.5ex] y^\top  & y_0 \end{bmatrix} \succeq 0 \\[3ex]
            &  y_0 + 2y^\top \xi + \xi^\top Y \xi \geq 1 \qquad \forall \xi \in \Xi_\tau,
        \end{array}
        \right. %\label{eq:WC-prob:1}
    \end{align*}
    where the uncertainty set $\Xi_\tau=\{ \xi \in \mathbb R^n: \ell(\xi) \geq \tau \}$ of the emerging robust constraint depends on the decision variable~$\tau$. Recalling the polyhedrality of $\ell(\xi)$, this constraint is satisfied if and only if 
    \begin{align*}
        \left.\begin{array}{cl}
            \DS \inf_{\xi \in \mathbb R^n} & \DS y_0 + 2y^\top \xi + \xi^\top Y \xi \\[1ex]
            \st & \DS \tau + \w^\top \xi + \w^\top\max\{ A\xi + a, B\xi + b \} \leq 0
        \end{array}\right\}\geq 1.
    \end{align*}
    By standard duality arguments akin to those used in \cite[Theorem~2.3]{ref:zymler2013distributionally}, this inequality holds if and only if there exists a scalar $\eta\geq 0$ and a vector $\zeta\in\R^n$, $\zeta\leq \w$, such that
    \begin{align*}
       \begin{bmatrix} Y & y \\ y^\top & y_0 \end{bmatrix} + \begin{bmatrix} 0 & \frac{\eta}{2} ( \w + (A - B)^\top \zeta + B^\top \w ) \\ \frac{\eta}{2} ( \w + (A - B)^\top \zeta + B^\top \w )^\top & \eta(\tau + (a-b)^\top \zeta + b^\top \w)-1 \end{bmatrix} \succeq 0.
    \end{align*}
    Substituting the emerging matrix inequality into our reformulation of the worst-case probability $\sup_{\QQ \in \mathbb G_{\rho}(\widehat \mu, \covsa)} \QQ ( \ell(\xi) \geq \tau)$ then yields the following reformulation of the Gelbrich VaR problem.
    \begin{align}
        \label{eq:WC-prob:2}
        %\Sup{\QQ \in \mathbb{G}_{\rho}(\widehat \mu, \covsa)} \QQ\text{-}\VaR_{\beta}( \ell(\xi)) = \left\{ 
        \begin{array}{cl}
            \inf &  \tau \\[-1ex]
            \st & \gamma, \eta, z \in \mathbb R_+, ~ \tau, y_0, v_0 \in \mathbb R, ~ v, y \in \mathbb R^n, ~ \zeta \in \mathbb R_+^k, ~ Y \in \mathbb{S}^n, ~ Z \in \PSD^n \hspace{-3em} \\
            & \zeta \leq \w, ~ y_0 + \gamma \big( \rho^2 - \| \widehat \mu \|^2 - \Tr{\covsa} \big) + z + \Tr{Z} \leq \beta \\
            & v =  \frac{1}{2} \left( \w + (A - B)^\top \zeta + B^\top \w \right), ~ v_0 = \tau + (a-b)^\top \zeta + b^\top \w \\
            &\begin{bmatrix} \gamma I - Y & \gamma \covsa^\half \\ \gamma \covsa^\half & Z \end{bmatrix} \succeq 0, \; \begin{bmatrix} \gamma I - Y & \gamma \widehat \mu + y \\ (\gamma \widehat \mu + y)^\top & z \end{bmatrix} \succeq 0 \\[3ex]
            & \begin{bmatrix} Y & y \\[0.5ex] y^\top  & y_0 \end{bmatrix} \succeq 0, \; \begin{bmatrix} Y & y + \eta v \\ y^\top + \eta v^\top & y_0 + \eta v_0 - 1 \end{bmatrix} \succeq 0.
        \end{array} 
    \end{align}
    Note that problem~\eqref{eq:WC-prob:2} is still non-convex because of the bilinear terms in the last constraint. To eliminate these bilinear terms, we first show that~$\eta>0$ at any feasible solution of~\eqref{eq:WC-prob:2}. Assume to the contrary that there exists a feasible solution with $\eta=0$. The last semidefinite constraint then implies that $y_0 \ge 1$. In addition, as $Y\succeq 0$ and thus $\gamma I - Y \preceq \gamma I$, the first semidefinite constraint implies via a Schur complement argument that $\Tr{Z} \ge \gamma \Tr{\covsa}$. Similarly, the second and the fourth semidefinite constraints 
    imply that $z \ge \gamma \|\widehat \mu\|^2 + 1 - y_0$. These insights collectively show that 
    \[
        y_0 + \gamma \big( \rho^2 - \| \widehat \mu \|^2 - \Tr{\covsa} \big) + z + \Tr{Z} \geq  1+\gamma\rho^2\geq 1,
    \]
    which contradicts the second constraint in~\eqref{eq:WC-prob:2} for any $\beta \in (0, 1)$. Hence, we may conclude that our assumption was false and that~$\eta>0$. 
    % for any It can be shown that any feasible solution of with vanishing $\eta$-component will satisfy
    % \[
    % y_0 + \gamma \big( \rho^2 - \| \widehat \mu \|^2 - \Tr{\covsa} \big) + z + \Tr{Z} \ge 1.
    % \]
    % To see this, observe that when $\eta = 0$, the last constraint in~\eqref{eq:WC-prob:2} implies that any feasible solution should have the component $y_0$ satisfy $y_0 \ge 1$. Furthermore, a Schur complement argument reveals that the first two semidefinite constraints of problem~\eqref{eq:WC-prob:2} imply $\Tr{Z} \ge \gamma \Tr{\covsa}$ and $z \ge \gamma \|\widehat \mu\|^2$ because $\gamma I - Y \preceq \gamma I$.
    % However, this is in conflict with the constraint 
    % \[
    % y_0 + \gamma \big( \rho^2 - \| \widehat \mu \|^2 - \Tr{\covsa} \big) + z + \Tr{Z} \leq \beta
    % \]
    % for $\beta \in (0, 1)$. This implies that any feasible solution of~\eqref{eq:WC-prob:2} will have $\eta > 0$. 
    Consequently, we can divide all constraints of problem~\eqref{eq:WC-prob:2} by~$\eta$, and we can apply the substitution $(\gamma, \eta, z, y_0, y, Y, Z) \leftarrow (\gamma/\eta, 1/\eta, z/\eta, y_0/\eta, y/\eta, Y/\eta, Z/\eta)$. In doing so, problem~\eqref{eq:WC-prob:2} reduces to the desired (tractable) semidefinite program.
\end{proof}

\begin{proof}[Proof of Theorem~\ref{theorem:poly-cvar}]
We have
    \begin{align*}
        \Sup{\QQ \in \mathbb{G}_{\rho}(\widehat \mu, \covsa)} \QQ\text{-}\CVaR_{\beta}( \ell(\xi)) 
        &= \Sup{(\mu, \cov) \in \U_{\rho}(\widehat \mu, \covsa)}\Sup{\QQ \in \Ambi(\mu, \cov)} \QQ\text{-}\CVaR_{\beta}( \ell(\xi)) \\
        & = \Sup{(\mu, \cov) \in \U_{\rho}(\widehat \mu, \covsa)}\Sup{\QQ \in \Ambi(\mu, \cov)} \QQ\text{-}\VaR_{\beta}( \ell(\xi))\\
        &=\Sup{\QQ \in \mathbb{G}_{\rho}(\widehat \mu, \covsa)} \QQ\text{-}\VaR_{\beta}( \ell(\xi)),
    \end{align*}
where the first and the third equalities follow from the two-layer decomposition~\eqref{eq:two-layer-a} of the Gelbrich ambiguity set, while the second equality holds because of \cite[Theorem~2.2]{ref:zymler2013distributionally}.
\end{proof}

\begin{proof}[Proof of Theorem~\ref{theorem:quad-var}]
    The proof widely parallels that of Theorem~\ref{theorem:poly-var} and can thus be kept short. First, the Gelbrich VaR can be expressed as the optimal value of a minimization problem over $\tau\in\R$ that involves the worst-case probability $\sup_{\QQ \in \mathbb{G}_{\rho}(\widehat \mu, \covsa)} \QQ( \ell(\xi) \geq \tau)$.
    % For any $\tau \in \mathbb R$, define the following set
    % \[
    % \Xi_\tau = \left\{ \xi \in \mathbb R^n:  \tau + \theta(\w) + \Delta(\w)^\top \xi + \half \xi^\top \Gamma(\w) \xi \le 0 \right\}.
    % \]
    By Corollary~\ref{corollary:WC-prob}, this worst-case probability can again be expressed as the optimal value of a minimization problem with a robust constraint. Using strong convex duality, this robust constraint is equivalent to a matrix inequality. In summary, 
    the Gelbrich VaR of a quadratic loss function thus equals
    \begin{align}
        \label{eq:WC-quadprob:3}
        \begin{array}{cl}
            \inf & \tau  \\[-1ex]
            \st &\gamma, z, \eta \in \mathbb R_+,  \; y_0 \in \mathbb R, \; y \in \mathbb R^n, \; Y \in \mathbb{S}^n, \; Z \in \PSD^n, \; \tau \in \mathbb R \\
            & y_0 + \gamma \big( \rho^2 - \| \widehat \mu \|^2 - \Tr{\covsa} \big) + z + \Tr{Z} \leq \beta \\
            &\begin{bmatrix} Y & y \\ y^\top & y_0 \end{bmatrix} + \begin{bmatrix} \eta \Gamma(\w) & \eta \Delta(\w) \\ \eta \Delta(\w)^\top & -1 + 2\eta(\tau + \theta(\w)) \end{bmatrix}  \succeq 0 \\[3ex]
            &\begin{bmatrix} \gamma I - Y & \gamma \covsa^\half \\ \gamma \covsa^\half & Z \end{bmatrix} \succeq 0, \; \begin{bmatrix} \gamma I - Y & \gamma \widehat \mu + y \\ (\gamma \widehat \mu + y)^\top & z \end{bmatrix} \succeq 0, \; \begin{bmatrix} Y & y \\[0.5ex] y^\top  & y_0 \end{bmatrix} \succeq 0.
        \end{array}
    \end{align}
    An analogous reasoning as in the proof of Theorem~\ref{theorem:poly-var} reveals that every feasible solution of problem~\eqref{eq:WC-quadprob:3} satisfies $\eta > 0$. The claim then follows by dividing the constraints of problem~\eqref{eq:WC-quadprob:3} by $\eta$ and by applying a similar variable substitution as in the proof of Theorem~\ref{theorem:poly-var}.
\end{proof}

\begin{proof}[Proof of Theorem~\ref{theorem:quad-cvar}]
    The proof is analogous to that of Theorem~\ref{theorem:poly-cvar} and thus omitted.
\end{proof}

The tractability of the generic worst-case expectation problem addressed in Theorem~\ref{theorem:WC-exp} depends on the properties of the underlying loss function~$\ell$. The following corollary shows that this worst-case expectation problem is equivalent to a tractable semidefinite programming whenever~$\ell$ is defined as the pointwise maximum of finitely many (possibly indefinite) quadratic functions.

\begin{corollary}[Piecewise quadratic loss] \label{corollary:WC-pq}
    Suppose that $\ell(\xi) = \max_{j=1,\ldots,J} \{\xi^\top Q_j \xi + 2 q_j^\top \xi + q^0_{j} \}$ for some $Q_j \in \mathbb S^n$, $q_j \in \mathbb R^n$ and $q_{j}^0 \in \mathbb R$, $j =1,\ldots,J$.
    If $\mc S = \mc M_2$ and $\rho>0$, then we have
    \begin{equation*}
        \Sup{\QQ \in \mc G_{\rho}(\widehat \mu, \covsa)} \EE_{\QQ}[\ell(\xi)] = \left\{
        \begin{array}{cl}
            \inf & y_0 + \gamma \big( \rho^2 - \| \widehat \mu \|^2 - \Tr{\covsa} \big) + z + \Tr{Z} \\
            \st &\gamma,z \in \mathbb R_+,  \; y_0 \in \mathbb R, \; y \in \mathbb R^n, \; Y \in \mathbb{S}^n, \; Z \in \PSD^n \\
            &\begin{bmatrix} \gamma I - Y & \gamma \covsa^\half \\ \gamma \covsa^\half & Z \end{bmatrix} \succeq 0, \; \begin{bmatrix} \gamma I - Y & \gamma \widehat \mu + y \\ (\gamma \widehat \mu + y)^\top & z \end{bmatrix} \succeq 0\\[3ex]
            &\begin{bmatrix} Y - Q_j & y - q_j \\[0.5ex] y^\top - q_j^\top & y_0 - q^0_{j} \end{bmatrix} \succeq 0 \quad \forall j =1,\ldots,J.
        \end{array}
        \right.
    \end{equation*}
\end{corollary}

\begin{proof}[Proof of Corollary~\ref{corollary:WC-pq}]
    By Theorem~\ref{theorem:WC-exp}, the worst-case expectation problem at hand can be reformulated as an equivalent minimization problem. The convex set $\mathcal Y$ appearing in this reformulation satisfies
    \begin{align*}
        \mathcal Y %&= \Big\{ y_0 \in \mathbb R, \; y \in \mathbb R^n, \; Y \in \Sym^n:  y_0 + 2y^\top \xi + \xi^\top Y \xi \geq \ell(\xi) ~ \forall \xi \in \mathbb R^n \Big\} \\[1ex]
        &= \left\{ y_0 \in \mathbb R, \; y \in \mathbb R^n, \; Y \in \Sym^n: y_0 + 2y^\top \xi + \xi^\top Y \xi \geq \xi^\top Q_j \xi + 2 q_j^\top \xi + q^0_{j} ~ \forall \xi \in \mathbb R^n,\, \forall j =1,\ldots,J
        \right\} \\
        &= \left\{ y_0 \in \mathbb R, \; y \in \mathbb R^n, \; Y \in \Sym^n: \begin{bmatrix} Y - Q_j & y - q_j \\[0.5ex] y^\top - q_j^\top  & y_0 - q_j^{0} \end{bmatrix} \succeq 0 ~ \forall j=1,\ldots,J \right\}.
    \end{align*}
    Thus, $\mathcal Y$ is semidefinite representable. The claim follows by substituting this representation of~$\mathcal Y$ into the optimization problem derived in Theorem~\ref{theorem:WC-exp}.
\end{proof}

\section{Mean-Variance Risk Measures}
\label{sect:mean:var}
Mean-variance risk measures appear in static portfolio selection~\cite{ref:markowitz1952portfolio, ref:follmer2008convex} and in myopic reformulations of dynamic portfolio optimization~\cite{ref:merton1969lifetime}, due to their analytic tractability. For any probability distribution~$\QQ \in \mc M$, the mean-variance risk measure with risk-aversion coefficient $\beta \ge 0$ of any loss function~$\ell \in \Lc_0$ is defined as
\[
    \risk_\QQ(\ell(\xi)) = \EE_\QQ[ \ell(\xi)] + \beta \Var_\QQ(\ell(\xi)),
\]
where $\Var_\QQ(\ell(\xi))$ denotes the variance of the loss $\ell(\xi)$ under the probability distribution~$\QQ$. 
Even though the mean-variance risk measure gives rise to a law-invariant family of translation invariant risk measures, it fails to be positive homogeneous, and therefore Theorem~\ref{theorem:meanstd} is not applicable. Nevertheless, the Gelbrich risk evaluation problem can still be reformulated as a tractable second-order cone program.

\begin{theorem}[Gelbrich mean-variance risk of linear loss functions]
    \label{theorem:meanvar}
    Suppose that $\{ \risk_\QQ\}_{\QQ \in \mc M}$ is a family of mean-variance risk measures with coefficient $\beta > 0$. If $\covsa \succ 0$, then the Gelbrich risk and the Wasserstein risk of any portfolio loss function~$\ell(\xi) = -\w^\top \xi$ coincide and are equal to the optimal value of the second-order cone program 
    %equivalent to the Wasserstein risk
    \be \label{eq:MV:refor}
    \begin{array}{cl}
        \inf &~ \gamma \rho^2 - \msa^\top \w + \frac{1}{4} y + \beta z \\
        \st &~ \gamma \in \R_{+}, \; y \in \R_{+}, \; z \in \R_+ \\[2ex]
        &~ \left\Vert \begin{pmatrix} 2 \covsa^{\frac{1}{2}} \w \\ z + \beta y - 1 \end{pmatrix} \right\Vert \leq z - \beta y + 1, \;
        \left\Vert \begin{pmatrix} 2  \w \\ y - \gamma \end{pmatrix} \right\Vert \leq y + \gamma, \; \beta y \leq 1.
    \end{array}
    \ee
\end{theorem}

\begin{proof}[Proof of Theorem~\ref{theorem:meanvar}]
    The equivalence of the Gelbrich risk and the Wasserstein risk for portfolio loss functions follows from Proposition~\ref{prop:gelbrich-projection} and the observation that the mean-variance risk measure depends on the distribution of the asset returns only through its first and second moments. It remains to be shown that the Gelbrich risk coincides with~\eqref{eq:MV:refor}.
    
    If~$w=0$, the Gelbrich risk evaluates to~$0$, and problem~\eqref{eq:MV:refor} is solved by~$\gamma=y=z=0$. Thus, the claim is trivially satisfied. From now on, we may assume without loss of generality that~$w\neq 0$. Using the decomposition~\eqref{eq:two-layer-a}, the Gelbrich risk can be recast as
    \begin{align*}
    \risk_{\mc{G}_{\rho}(\msa, \covsa)}(\ell) = \Sup{(\m, \cov) \in \U_\rho(\msa, \covsa)} \Sup{\QQ\in \mc C(\m, \cov)} \risk_{\QQ}(-\w^\top \xi) =
    \Sup{(\m, \cov) \in \U_\rho(\msa, \covsa)} \left\{ -\w^\top \m + \beta \w^\top \cov \w \right\}, 
    \end{align*}
    where the second equality uses the definition of the mean-variance risk measure. Note that the last optimization problem in the above expression evaluates the support function of $\U_\rho(\msa, \covsa)$ at the point~$(-\w, \beta \w\w^\top)$, and by Proposition~\ref{proposition:support:U} we thus have
    \begin{align*}
    \Sup{\QQ \in \mc G_{\rho}(\msa, \covsa)}  \; \risk_\QQ ( -\w^\top \xi ) 
    = \left\{
        \begin{array}{cl}
        \inf & -\msa^\top \w + \tau + \gamma \big(\rho^2 -  \Tr{\covsa} \big) + \Tr{Z} \\[1ex]
        \st & \gamma \in \R_+, \; \tau \in\R_+,\; Z \in \mathbb{S}^n_+ \\[1ex]
         & \begin{bmatrix} \gamma I - \beta \w\w^\top & \gamma \covsa^{\frac{1}{2}} \\ \gamma \covsa^{\frac{1}{2}} & Z \end{bmatrix} \succeq 0, \;  \left\| \begin{pmatrix} -\w \\ \tau - \gamma \end{pmatrix} \right\| \leq \tau + \gamma.
        \end{array}
    \right.
    \end{align*}
    Fix any feasible solution~$(\gamma, \tau, Z)$ of the resulting semidefinite program. The first matrix inequality implies that~$\gamma I\succeq \beta ww^\top$. As~$\beta>0$ and~$w\neq 0$, this is only possible if~$\gamma>0$, which in turn implies that~$\gamma\covsa^{\frac{1}{2}}\succ 0$. By \cite[Corollary~8.2.2]{ref:bernstein2009matrix}, the first matrix inequality in the above semidefinite program, therefore, also implies that~$\gamma I - \beta \w\w^\top \succ 0$, which is equivalent to~$\gamma > \beta \| \w \|^2$.
%     \begin{align*}
%         \begin{bmatrix} \gamma I - \beta \w\w^\top & \gamma \covsa^{\frac{1}{2}} \\ \gamma \covsa^{\frac{1}{2}} & Z \end{bmatrix} \succeq 0\quad \implies \quad \gamma I - \beta \w\w^\top \succ 0\quad \implies \quad \gamma > \beta \| \w \|^2,
%     \end{align*}
%    where the first step holds because $\covsa$ is regular and due to \cite[Corollary~8.2.2]{ref:bernstein2009matrix}, the matrix $\gamma I - \beta \w\w^\top$ has to be full rank.
    It is easy to verify that, at optimality, $\tau$ coincides with~$\|w\|^2/(4\gamma)$ and~$Z$ coincides with its Schur complement~$\gamma^2(\gamma I-\beta ww^\top)^{-1}\covsa$, which is well-defined because~$\gamma I - \beta \w\w^\top \succ 0$. Thus, $\tau$ and~$Z$ can be eliminated together with their constraints to obtain
    %Applying the Schur complement argument and eliminating the decision variable $\tau$, the Gelbrich mean-variance risk measure simplifies to
    \begin{align*}
        &\phantom{=} \Sup{\QQ \in \mc G_{\rho}(\msa, \covsa)}  \; \risk_\QQ ( -\w^\top \xi ) \\
        &= \inf_{\gamma > \beta \| \w \|^2} - \msa^\top \w + \frac{\| \w \|^2}{4\gamma} + \gamma \big(\rho^2 - \Tr{\covsa} \big) + \gamma^2 \Tr{(\gamma I - \beta \w\w^\top)^{-1} \covsa} \\
        &= \Inf{\gamma > \beta \|\w\|^2} \gamma \rho^2 - \msa^\top \w + \frac{\|w\|^2}{4\gamma}  + \beta (1 - \beta \gamma^{-1} \| w\|^2)^{-1} \w^\top \covsa \w,
    \end{align*}
    where the second equality follows from the Sherman-Morrison formula \cite[Corollary~2.8.8]{ref:bernstein2009matrix}. Introducing an auxiliary variable $y \geq 0$ and rewriting the constraint in the last expression as~$ \| \w\|^2/ \gamma \le y \le \beta^{-1}$ yields
    \be \notag
        \Sup{\QQ \in \mc G_{\rho}(\msa, \covsa)} \risk_\QQ \left( -\w^\top \xi \right) = 
        \left\{
        \begin{array}{cl}
        \inf & \gamma \rho^2 - \msa^\top \w + \frac{y}{4}  + \beta (1 - \beta y)^{-1} \w^\top \covsa \w  \\
        \st &\| \w\|^2/ \gamma \le y \le \beta^{-1}, \; \gamma > 0.
    \end{array}
    \right.
    \ee
    As~$\gamma\ge 0$, $y\ge 0$ and $\beta y\leq 1$, we may use \cite[Equation~(8)]{ref:lobo1998socp} to reformulate the hyperbolic constraint and the quadratic-over-linear term in the objective function in terms of second-order cone constraints, that is,
    \begin{align*}
        \| \w\|^2 \le \gamma y~
        &\iff~ \left\Vert \begin{pmatrix} 2  \w \\ y - \gamma \end{pmatrix} \right\Vert \leq y + \gamma
        \quad\text{and}\quad \frac{\w^\top \covsa \w}{1 - \beta y} \leq z 
        ~ \\
        &\iff
        ~
        \left\Vert \begin{pmatrix} 2 \covsa^{\frac{1}{2}} \w \\ z + \beta y - 1 \end{pmatrix} \right\Vert \leq z - \beta y + 1,
    \end{align*}
    where~$z\ge 0$ is an auxiliary epigraphical variable. Thus, the claim follows.
\end{proof}

In analogy to Proposition~\ref{prop:WC}, we can determine the first and second moments of the worst-case probability distributions that maximize the Gelbrich mean-variance risk of a fixed linear loss function. 

\begin{proposition}[Worst-case moments]
    \label{prop:WC:meanvar}
    Suppose that $\{ \risk_\QQ\}_{\QQ \in \mc M}$ is a family of mean-variance risk measures with $\beta > 0$. If $\covsa \succ 0$, $w\neq 0$ and~$\rho>0$, then any extremal distribution $\QQ\opt$ that attains the Gelbrich risk of the linear loss function~$\ell(\xi) = -\w^\top \xi$ has the same mean $\m\opt \in \R^n$ and covariance matrix $\cov\opt \in \PSD^n$, where
    \begin{align*}
    \m\opt = \msa  - \frac{\w}{2\gamma\opt}, \quad \cov\opt = \left(I - \frac{\beta \w \w^\top}{\gamma\opt}\right)^{-1} \covsa \left(I - \frac{\beta \w \w^\top}{\gamma\opt}\right)^{-1}
    \end{align*}
    and $\gamma\opt > \beta \| \w \|^2$ is the unique solution of the nonlinear algebraic equation
    \begin{align*}
        \frac{\| w \|^2}{4\gamma^2} + \Tr{\covsa \left( I - \gamma (\gamma I - \beta \w \w^\top)^{-1} \right)^2} = \rho^2.
    \end{align*}
\end{proposition}

We emphasize again that there may be multiple extremal distributions in the Gelbrich ambiguity set $\mc G_{\rho}(\msa, \covsa)$ that share the unique extremal mean~$\m\opt$ and covariance matrix $\cov\opt$ identified in Proposition~\ref{prop:WC:meanvar}. 

\begin{proof}[Proof of Proposition~\ref{prop:WC:meanvar}]
    From the proof of Theorem~\ref{theorem:meanvar} we know that the first and second moments of any extremal distribution~$\QQ\opt$ are maximizers of the support function evaluation problem~$\sup \{ q^\top \m + \Tr{Q\cov}: (\m, \cov) \in \U_\rho(\msa, \covsa)\}$ with~$q = - w$ and~$Q = \beta \w \w^\top$. The claim thus follows from Lemma~\ref{lemma:extremal:support:U}, which applies because~$w\neq 0$. 
\end{proof}

Distributionally robust mean-variance portfolio optimization problems with 2-Wasser\-stein ball ambiguity sets were also studied by \cite{ref:blanchet2020distributionally}. Instead of minimizing a worst-case mean-variance risk measure, they minimize the worst-case variance of the portfolio return subject to a lower bound on the worst-case mean, resulting in a more conservative model because the extremal distributions in the objective function and the constraints may differ. The additional conservatism enhances analytical tractability.

\end{appendix}

\end{document}